\documentclass[a4paper]{article}

\usepackage[english]{babel}
\usepackage{hyperref}

\usepackage{amsmath, amsfonts,xfrac, amsthm,amssymb,fancyhdr}

\usepackage{mathtools}
\usepackage{todonotes}
\usepackage{setspace}

\usepackage{geometry}
\usepackage{xypic,colortbl}

\usepackage{tikz}
\usetikzlibrary{matrix, arrows}

\usepackage{ifthen, sidecap, wrapfig, xspace}

\usepackage{mathrsfs, stmaryrd}

\usepackage{sectsty}

\usepackage{
 calc, enumitem, makeidx,bbold
}

\usepackage[font=small,format=plain,labelfont=sc,textfont=it]{caption}
\usepackage[T1]{fontenc}
\usepackage[utf8]{inputenc}
\usepackage[all]{xy}
\usepackage[geometry]{ifsym}

\usepackage{palatino}
\usepackage{mathpazo}
\usepackage{comment}
\usepackage{enumitem}

\usepackage{xfrac}

\def\cqedsymbol{\ifmmode$\lrcorner$\else{\unskip\nobreak\hfil
\penalty50\hskip1em\null\nobreak\hfil$\lrcorner$
\parfillskip=0pt\finalhyphendemerits=0\endgraf}\fi} 
\newcommand{\cqed}{\renewcommand{\qed}{\cqedsymbol}}

\newcommand{\Oof}{\mathcal{O}}
\renewcommand{\preceq}{\preccurlyeq}
\renewcommand{\succeq}{\succcurlyeq}
\newcommand{\wh}{\widehat}
\newcommand{\wt}{\widetilde}

\newcommand{\ind}[2][]{%
  \mathrel{
    \mathop{
      \vcenter{
        \hbox{\oalign{\noalign{\kern-.3ex}\hfil$\vert$\hfil\cr
              \noalign{\kern-.7ex}
              $\smile$\cr\noalign{\kern-.3ex}}}
      }
    }^{#2}\displaylimits_{#1}
  }
}

\newcommand{\from}{\colon}
\newcommand{\str}[1]{\mathbf{#1}}
\renewcommand{\cal}[1]{\mathcal {#1}}
\newcommand{\CC}{\mathscr C}
\newcommand{\DD}{\mathscr D}
\newcommand{\EE}{\mathscr E}
\renewcommand{\le}{\leqslant}
\renewcommand{\ge}{\geqslant}
\renewcommand{\leq}{\leqslant}
\renewcommand{\geq}{\geqslant}
\renewcommand{\phi}{\varphi}

\newcommand{\wcol}{\mathrm{wcol}}
\newcommand{\scol}{\mathrm{scol}}
\newcommand{\WReach}{\mathrm{WReach}}
\newcommand{\SReach}{\mathrm{SReach}}

\newcommand{\mathsym}[1]{{}}

{\begin{figure}[#1]
\centering 
\includegraphics[scale=#3]{#2}
}
{\end{figure}}

\newlist{enumeratep}{enumerate}{10}
\setlist[enumeratep]{label=\quad\textit{\arabic*'.},ref=\arabic*',leftmargin=*}

\newenvironment{romanlist'}[0]
{\begin{list}{\makebox[0.5cm][l]{\textit{\roman{enumi}')}}}{\usecounter{enumi}}}
{\end{list}}

\newcommand{\savelabel}[2]{\expandafter\newtoks\csname#1\endcsname
  \global\csname#1\endcsname={#2} \label{#1} #2}
\newcommand{\loadlabel}[1]{\noindent {\bf Lemma~\ref{#1}. } \textit{\the\csname#1\endcsname}
\medskip

}
\renewcommand{\setminus}{-}

\newcommand{\loadlabelthm}[1]{\medskip\noindent {\bf Theorem~\ref{#1}. }
  \noindent  \textit{\the\csname#1\endcsname}
\medskip
}

\newcommand{\loadlabelprop}[1]{\noindent {\bf Proposition~\ref{#1}. }
  \noindent  \textit{\the\csname#1\endcsname}
\medskip

}

\newcommand{\N}{\mathbb{N}}

\renewcommand{\subset}{\subseteq}

\newcommand{\atleast}[1]{{\ge n}}
\newcommand{\less}[1]{{<n}}

	\newcommand{\notacol}[2]{}

\newsavebox{\quoteitbox}

{\hspace*{\fill}{\upshape(\usebox{\quoteitbox})}\end{quote}%
\end{sloppypar}\noindent\ignorespacesafterend}

\newenvironment{quoteit*} 
{\begin{sloppypar}\noindent\slshape\begin{quote}\itshape} 
	{\end{quote}\ignorespaces\end{sloppypar}\noindent\ignorespacesafterend}

\newenvironment{quotetag*}
{~\par
	\begingroup                  
	\begin{equation*}
		 \begin{minipage}[c]{115mm}
			\it\noindent{\par}
}
{
		\end{minipage}
	\end{equation*}
	\endgroup                        
\par
\textnormal
\medskip
}

\newcommand{\Ff}{{\mathcal F}}

\newcommand\set[1]{\ensuremath{\{#1\}}}
\newcommand{\setof}[2]{\set{#1\mid#2}}

\DeclareMathOperator{\tp}{tp}

\newcommand{\quo}[2]{\sfrac{#1}{#2}}

\newtheoremstyle{theoremstyle}
  {3pt}
  {3pt}
  {\itshape}
  {0pt}
  {\bfseries}
  {.}
  {4pt}
  {}

\theoremstyle{theoremstyle}
\newtheorem{theorem}{Theorem}[section]
\newtheorem{conjecture}{Conjecture}
\newtheorem*{theorem*}{Theorem}

\newtheorem{lemma}[theorem]{Lemma}
\newtheorem{corollary}[theorem]{Corollary}
\newtheorem{proposition}[theorem]{Proposition}
\newtheorem{claim}{Claim}

\newtheoremstyle{remarkstyle}
  {3pt}
  {10pt}
  {}
  {0pt}
  {\itshape}
  {}
  {4pt}
  {\thmname{#1}\thmnumber{ #2}\thmnote{ (#3)}.}

\theoremstyle{remarkstyle}
\newtheorem{example}{Example}[section]

\newtheoremstyle{definitionstyle}
  {3pt}
  {3pt}
  {}
  {0pt}
  {\bfseries}
  {}
  {4pt}
  {\thmname{#1}\thmnumber{ #2}\thmnote{ (#3)}.}

\theoremstyle{definitionstyle}
\newtheorem{definition}{Definition}

\numberwithin{equation}{section}

\newlength{\wideaslength}

\renewcommand{\subset}{\subseteq}

\newcommand{\seta}[1]{}

\def\lsim{\mathrel{\rlap{\lower4pt\hbox{\hskip1pt$\sim$}}
    \raise1pt\hbox{$<$}}}                

\definecolor{gray1}{rgb}{0.99,0.99,0.99}
\definecolor{gray2}{rgb}{0.97,0.97,0.97}
\definecolor{gray3}{rgb}{0.95,0.95,0.95}
\definecolor{gray4}{rgb}{0.93,0.93,0.93}
\definecolor{gray5}{rgb}{0.91,0.91,0.91}
\definecolor{gray6}{rgb}{0.89,0.89,0.89}
\definecolor{gray7}{rgb}{0.87,0.87,0.87}
\definecolor{gray8}{rgb}{0.85,0.85,0.85}
\definecolor{gray9}{rgb}{0.83,0.83,0.83}
\definecolor{gray10}{rgb}{0.81,0.81,0.81}
\definecolor{gray20}{rgb}{0.71,0.71,0.71}
\definecolor{gray40}{rgb}{0.51,0.51,0.51}

\newcommand{\ERCagreement}{
\footnotetext{\noindent
{\begin{minipage}[t]{0.7\textwidth}
\vspace{-15pt}
\small This paper is a part of projects {\sc{LIPA}} (JG) and {\sc{BOBR}} (MP, SzT) that have received funding from the European Research Council (ERC) under the European Union's Horizon 2020 research and innovation programme (grant agreements No 683080 and 948057, respectively). 
\vspace{10pt}
 \end{minipage} \hspace{5pt}
 \begin{minipage}{.23\textwidth} \vspace{5pt} 
 \includegraphics[width=\textwidth]{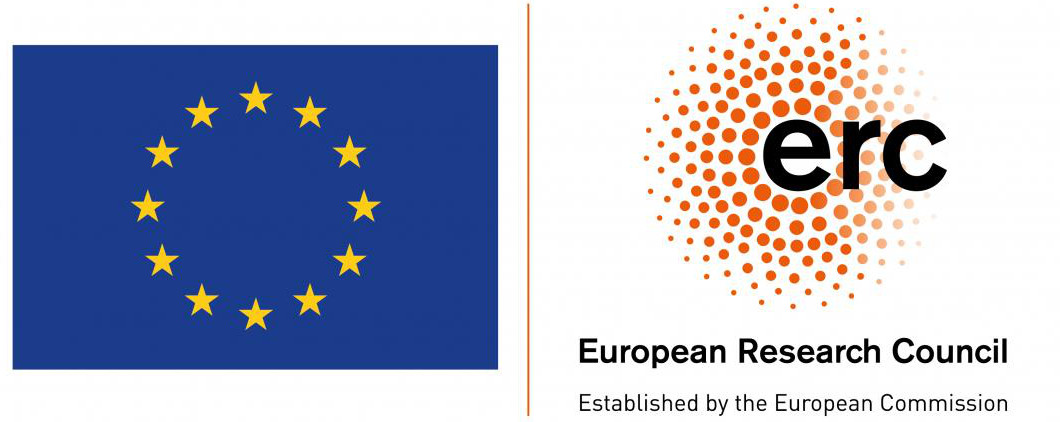}\end{minipage}\hfill}}}

 \title{Stable graphs of bounded twin-width\ERCagreement}
 \author{Jakub Gajarsk{\'y}\thanks{University of Warsaw, Poland, \texttt{jakub.gajarsky@mimuw.edu.pl}}, 
 Micha{\l} Pilipczuk\thanks{University of Warsaw, Poland, \texttt{michal.pilipczuk@mimuw.edu.pl}}, Szymon Toru{\'n}czyk\thanks{University of Warsaw, Poland, \texttt{szymtor@mimuw.edu.pl}}}
 
\begin{document}
\maketitle

\begin{abstract}
 We prove that every class of graphs $\CC$ that is monadically stable
 and has bounded twin-width can be transduced from some class with bounded sparse twin-width. This generalizes analogous results for classes of bounded linear cliquewidth~\cite{nesetril2021linrw_stable} and of bounded cliquewidth~\cite{nesetril2021rw_stable}. It also implies that monadically stable classes of bounded twin-width are linearly $\chi$-bounded.
\end{abstract}


\section{Introduction}
\newcommand{\msole}{\preceq_{\mathrm{MSO}}}
\newcommand{\cmsole}{\preceq_{\mathrm{CMSO}}}
\newcommand{\fole}{\preceq_{\mathrm{FO}}}
\newcommand{\msoge}{\succeq_{\mathrm{MSO}}}
\newcommand{\cmsoge}{\succeq_{\mathrm{CMSO}}}
\newcommand{\foge}{\succeq_{\mathrm{FO}}}
\newcommand{\foeq}{\equiv_{\mathrm{FO}}}

A line of work in  structural graph theory
 seeks to generalize results obtained for 
sparse graphs to graphs which are possibly dense, but also well-structured in some sense.
A classic example of this principle is the case of tree-like graphs. The standard graph parameter measuring tree-likeness for sparse graphs is treewidth, while its natural analogue in the dense setting is cliquewidth (or, equivalently, rankwidth). By now, this analogy has been well-understood from multiple points of view. For instance, the boundedness of treewidth and of cliquewidth delimits the area of algorithmic tractability of two natural variants of the monadic second-order logic (MSO) on graphs, in the sense of the existence of a fixed-parameter algorithm for model checking~\cite{courcelle90tw,courcelle2000cw}. Further, both parameters admit duality theorems linking them to the largest size of a grid that can be embedded in the considered graph as a minor (for treewidth) or as a vertex-minor (for cliquewidth)~\cite{robertson86gridminor,geelen2020vertex_gridminor}. Finally, cliquewidth ``projects'' to treewidth once we restrict attention to sparse graphs in the following sense: every class of graphs $\CC$ that has bounded cliquewidth and is weakly sparse, in fact has bounded treewidth. Here, we say that a class $\CC$ has bounded parameter $\pi$ if there is a universal upper bound on the value of $\pi$ in the members of~$\CC$, and $\CC$ is {\em{weakly sparse}} if there is $s\in \N$ such that all members of $\CC$ exclude the biclique $K_{s,s}$ as a subgraph.

%
%



Arguably, requiring that a class of graphs has bounded  treewidth or cliquewidth is very restrictive, as even very simple graph classes, such as grids, have unbounded values of these parameters. 
While treewidth and cliquewidth explain well the limits of tractability of problems expressible in MSO, the analogous realm for the first-order logic (FO) is much broader, and  not yet fully understood. The ultimate goal of completing this understanding is the fundamental motivation behind this work.

So far, the limit of tractability of model-checking FO  has been thoroughly explored in classes of sparse structures. In this context, {\em{nowhere denseness}} has been identified as the main dividing line. Roughly, a class of graphs $\CC$ is {\em{nowhere dense}} if for every $r\in \N$, one cannot obtain arbitrarily large complete graphs by contracting mutually disjoint connected subgraphs of radius at most $r$ in graphs from $\CC$. This notion is very general, as it encompasses most well-studied concepts of sparsity in graphs, including having bounded treewidth, bounded degree, excluding a fixed (topological) minor, or having bounded expansion.
As it turns out, under plausible complexity-theoretic assumptions, for every subgraph-closed class of graphs $\CC$ the model-checking problem for FO is fixed-parameter tractable on $\CC$ if and only if $\CC$ is nowhere dense~\cite{grohe2017fo_nd}. 

In this statement, the assumption that $\CC$ is subgraph-closed is crucial. For instance, FO model-checking is fixed-parameter tractable on any class of bounded cliquewidth, however these classes are not nowhere dense. The explanation here is that they are not subgraph-closed either. Identifying the dividing line for fixed-parameter tractability of model checking FO on all classes of graphs is the central open problem in the area.

%

\paragraph*{Transductions.}
Drawing inspiration from model theory,
to study the expressive power of FO on a given class $\CC$ of graphs, we look at the classes $\DD$ of graphs which can be obtained from graphs from $\CC$ using transformations definable in FO. This idea is best formalized by the notion of an (FO) \emph{transduction}. Write $\DD\fole\CC$, and say that $\DD$ can be {\em{transduced}} from $\CC$,
if every graph  $H\in \DD$ can be obtained from
some graph $G\in \CC$ by first creating a fixed number of copies of $G$,  then coloring the vertices of these copies arbitrarily,  applying a fixed FO-formula $\phi(x,y)$  (which can use the colors just introduced, and distinguish copies of the same vertex), thus defining a new edge relation, and finally, taking an induced subgraph of the resulting graph.
For example, if $\DD$ is the class consisting of edge-complementations of graphs from $\CC$, then $\DD\fole\CC$, as we can take the formula $\phi(x,y)=\neg E(x,y)$, where $E(x,y)$ is the edge relation.

The relation $\fole$ defines a quasi-order on graph classes. Two graph classes $\CC$ and $\DD$ are \emph{transduction equivalent}, denoted $\CC\foeq \DD$, if $\CC\fole \DD\fole \CC$, that is, each can be transduced from the other.


%
%
%

\begin{figure}[h]\centering
    \includegraphics[page=1,scale=0.60]{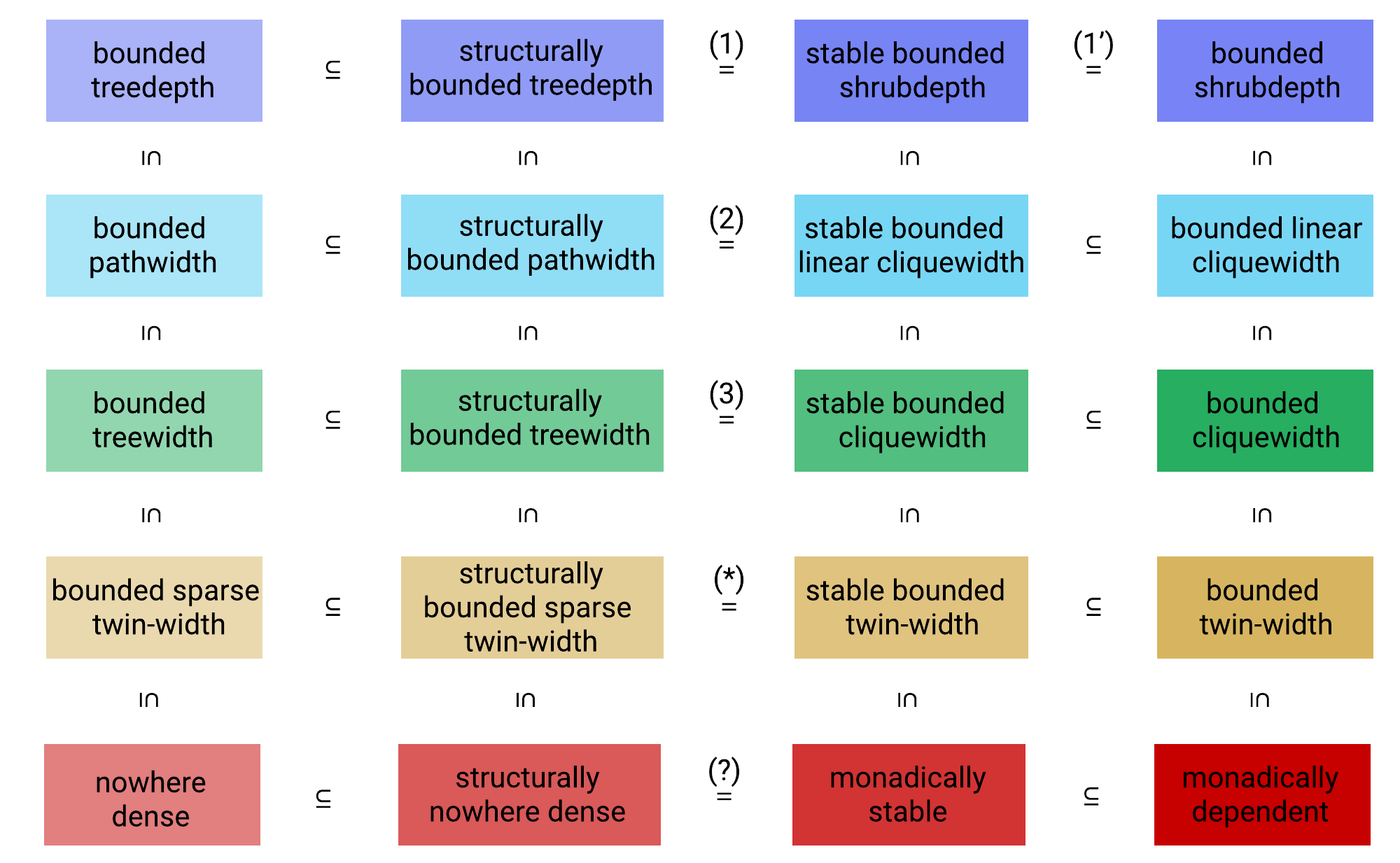}
    \caption{A roadmap of properties of hereditary graph classes. 
    If a class satisfies one of the properties in the first column, then it is weakly sparse, that is, excludes some biclique as a subgraph. 
    The property `structurally $\cal P$' consists of transductions of classes with property $\cal P$. The property `stable $\cal P$' consists of those classes that satisfy property $\cal P$ and are stable, that is, exclude 
    some ladder (cf. Fig.~\ref{fig:ladders}) as a semi-induced subgraph.
    Each property in the second, third and fourth column is a transduction ideal (is downward closed with respect to $\fole$).
    For every row $(\cal P_1,\cal P_2,\cal P_3, \cal P_4)$ in the table,
    the property $\cal P_1$ consists of all classes in $\cal P_4$ that are weakly sparse;
    the property $\cal P_2$ is the property of being `structurally $\cal P_1$'; the property $\cal P_3$ consists of all classes in $\cal P_4$ which are stable.
    The inclusion $\cal P_2\subset \cal P_3$ holds in each row, and equality $\cal P_2=\cal P_3$ holds for the first four rows $(1)$, $(2)$, $(3)$, $(*)$,
    with $(*)$ being our main result, Theorem~\ref{thm:main}.
    Equality $(?)$ is
    Conjecture~\ref{conj:sparsification}. All remaining inclusions in the figure are strict.
    }
    \label{fig:intro}
\end{figure}

The most general notion of well-structuredness that one can consider in this context is \emph{monadic dependence}, defined as follows: a class of graphs $\CC$ is \emph{monadically dependent} if it is not transduction equivalent to the class of all graphs. Here \emph{monadically} refers not to the logic, but to the ability of transductions to apply arbitrary colorings which can be then accessed by the formulas.
It appears that all the mentioned properties of graph classes, in particular nowhere denseness and having bounded cliquewidth, imply monadic dependence.
See Fig.~\ref{fig:intro} for a roadmap of various properties of graph classes which we will discuss later.

Remarkably, it turns out that monadic dependence projects to nowhere denseness in the same sense as was discussed for cliquewidth and treewidth: every weakly sparse class that is monadically dependent is actually nowhere dense~\cite{Dvorak18} (see also~\cite{nesetril2021linrw_stable}). Thus, we have the following equivalence of notions of combinatorial, logical, and algorithmic nature:
\begin{theorem}\label{thm:nd}
Assuming $\mathsf{AW}[\star]\neq \mathsf{FPT}$, the following  conditions are equivalent for every weakly sparse hereditary class of graphs $\CC$:
\begin{enumerate}[nosep]
    \item $\CC$ is nowhere dense,  
    \item $\CC$ is monadically dependent,
and
\item  model checking first-order logic is fixed-parameter tractable on $\CC$.
\end{enumerate}
\end{theorem}
Since both nowhere denseness and having bounded cliquewidth imply fixed-parameter tractability of model checking FO on a given class of graphs, while monadic dependence is their common generalization, this suggests the following conjecture\footnote{This conjecture has been
circulating in the community for some time, see e.g. 
the open problem session at the workshop on Algorithms, Logic and Structure in Warwick in 2016. See also \cite[Conjecture 8.2]{gajarsky2020bd_interp}.}.

\begin{conjecture}\label{conj:dependent-mc}
    For every hereditary class of graphs $\CC$,
    model checking first-order logic on $\CC$ is fixed-parameter tractable if, and only if  $\CC$ is monadically dependent. 
\end{conjecture}

A positive verification of Conjecture~\ref{conj:dependent-mc} would place the dividing line for algorithmic tractability of FO on graph classes exactly at the notion of monadic dependence.

\paragraph*{Stability.} 
Observe that the discussed properties of classes of sparse graphs --- having bounded treewidth and nowhere denseness --- are not closed under taking FO transductions, as witnessed by edge complementation. On the other hand, monadic dependence and having bounded cliquewidth are closed under taking FO transductions. Hence, here is a natural question: every image of a class of bounded treewidth under an FO transduction has bounded cliquewidth, but is it the case that every class of bounded cliquewidth can be transduced from a class of bounded treewidth? The same can be asked about nowhere denseness and monadic dependence.

The answer here is negative and is delivered by another important dividing line originating in model theory: {\em{stability}}. We say that a class of graphs $\CC$ is {\em{monadically stable}} if $\CC \not\foge\textit{Ladders}$, where $\textit{Ladders}$ is the class of all ladders\footnote{Ladders are often also called {\em{half-graphs}} in the literature.}, as depicted in Fig.~\ref{fig:ladders}. Clearly, monadic stability is a property of a graph class that is preserved by FO transductions. Further, it turns out that every nowhere dense class is monadically stable~\cite{adler2014interpreting}, hence by applying an FO transduction to a nowhere dense class one can only obtain classes which are monadically stable. This explains the second and third column in Fig.~\ref{fig:intro}. Note that the class of ladders has bounded cliquewidth but is not monadically stable, hence it can serve as an example distinguishing notions from the third (stable) column and the fourth (dependent) column.

Let us remark that even though monadic stability is a notion originating in model theory, in case of monadically dependent classes of graphs it can be understood in purely graph-theoretical terms. As proved in~\cite{nesetril2021rw_stable}, a monadically dependent class is monadically stable if and only if it excludes some fixed ladder as a {\em{semi-induced subgraph}}, that is, as an induced subgraph except that we allow any adjacencies within the sides of the ladder. This means that the notions in the third column of Fig.~\ref{fig:intro} can be obtained from the notions in the fourth column by restricting attention to monadically stable classes, or equivalently to classes that exclude a fixed ladder as a semi-induced subgraph.

\begin{figure}\centering
    \includegraphics[page=2,scale=1]{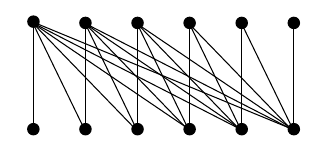}
    \caption{A ladder of length $6$.}
    \label{fig:ladders}
\end{figure}

Is it then the case that monadic stability exactly characterizes classes of graphs that can be transduced from classes of sparse graphs? The following conjecture says that this is the case.

\begin{conjecture}[\cite{Mendez21}]\label{conj:sparsification}
 For every monadically stable class of graphs $\CC$ there exists a nowhere dense class $\DD$ such that $\CC\fole \DD$. 
\end{conjecture}

One could intuitively understand Conjecture~\ref{conj:sparsification} as follows: whenever $\CC$ is monadically stable, for each $G\in \CC$ one can find a sparse ``skeleton'' graph $H$ such that $G$ can be encoded in $H$ in a way that is decodable by an FO transduction. The class $\DD$ comprising all skeleton graphs $H$ is nowhere dense.

Conjecture~\ref{conj:sparsification} is corroborated by the following two results on more restrictive properties.

\begin{theorem}[\cite{nesetril2021linrw_stable}]\label{thm:sparse-lin-cw}
 Every class of graphs that is monadically stable and has bounded linear cliquewidth is transduction equivalent to a class of bounded pathwidth.
\end{theorem}

\begin{theorem}[\cite{nesetril2021rw_stable}]\label{thm:sparse-cw}
 Every class of graphs that is monadically stable and has bounded cliquewidth is transduction equivalent to a class of bounded treewidth.
\end{theorem}

Here, {\em{linear cliquewidth}} is a linear variant that relates to cliquewidth in a similar way as pathwidth relates to treewidth. Theorem~\ref{thm:sparse-lin-cw} and~\ref{thm:sparse-cw} correspond to equalities in the second and third row in Fig.~\ref{fig:intro}. 

Let us remark that the works~\cite{nesetril2021linrw_stable,nesetril2021rw_stable} claim only one direction of the implications: that every monadically stable class of bounded cliquewidth (resp. linear cliquewidth) can be transduced from a class of bounded treewidth (resp. pathwidth). The equivalence stated in Theorems~\ref{thm:sparse-lin-cw} and~\ref{thm:sparse-cw} follows by combining these results with the main result of~\cite{gajarsky2020sbe}; see the proof of Theorem~\ref{thm:main} in Section~\ref{sec:main-theorem} where we use the same argument.

\paragraph*{Twin-width.} Looking at the picture sketched above from a perspective, there seems to be a need for a combinatorially defined concept that would on one hand generalize the notion of bounded cliquewidth, and on the other hand capture classes of well-behaved, but not tree-like graphs, like planar graphs or graphs excluding a fixed minor. Such a concept has been introduced very recently by Bonnet et al.~\cite{bonnet2020tww} through the {\em{twin-width}} graph parameter. Intuitively, a graph has twin-width $d$ if it can be constructed by merging larger and larger parts so that at any moment during the construction, every part has a non-trivial interaction
with at most $d$ other parts
(trivial interaction between two parts means that either no edges, or all edges span across the two parts). Here are some facts proved in~\cite{bonnet2020tww} that may help the reader to properly place classes of bounded twin-width in Fig.~\ref{fig:intro}:
\begin{itemize}[nosep]
 \item Every class of bounded cliquewidth has also bounded twin-width.
 \item Every class that excludes a fixed minor has bounded twin-width. This in particular applies to planar graphs, or graphs embeddable in any fixed surface.
 \item The class of all graphs of maximum degree at most $3$ has unbounded twin-width. Thus, not all nowhere dense classes have bounded twin-width.
 \item Having bounded twin-width is preserved by applying FO transductions.
 \item Every class of bounded twin-width is monadically dependent (this follows from the last two items).
\end{itemize}
Classes that have bounded twin-width and are weakly sparse are said to have {\em{bounded sparse twin-width}}. As proved in~\cite{bonnet2021tww2}, every class of bounded sparse twin-width has {\em{bounded expansion}}, which is a more restrictive property than nowhere denseness. See also~\cite{dreier2021twinwidth} for concrete constructions and bounds in this context. 

Let us also remark that the notion of twin-width is not only applicable to graphs, but more generally to relational structures over binary signatures. Thus, we can for instance speak about the twin-width of permutations (sets equipped with two total orders) or ordered graphs (graphs equipped with a total order on the vertices).

As explained in Theorem~\ref{thm:nd}, monadic dependence equals nowhere denseness if one assumes that the class in question is weakly sparse.
It turns out that for classes of ordered graphs, monadic dependence is equivalent to having bounded twin-width.
\begin{theorem}[\cite{tww4a}]\label{thm:ordered}
Assuming $\mathsf{AW}[\star]\neq \mathsf{FPT}$, 
the following conditions are equivalent
for every hereditary class $\CC$ of ordered graphs:
\begin{enumerate}[nosep]
    \item $\CC$ has bounded twin-width,
    \item $\CC$ is monadically dependent,
and
\item  model checking first-order logic is fixed-parameter tractable on $\CC$.
\end{enumerate}
\end{theorem}

Theorem~\ref{thm:ordered} suggests a possible route of approaching Conjecture~\ref{conj:dependent-mc}. Namely, a ladder of length $k$ encodes, through its adjacency relation, a total order of length $k$. Thus, monadically stable classes can be equivalently defined as classes from which one cannot transduce all total orders. The other extreme are classes of ordered graphs, where a total order on all the vertices is explicitly present. It is conceivable that every structure from a monadically dependent class can be, in some sense, decomposed into parts that are either ``orderless'' or ``orderfull'', in the sense of definability of a total order on their elements. While Theorem~\ref{thm:ordered} could deliver twin-width-related tools for handling the orderfull parts, it is an imperative to understand also the other side of the spectrum: monadically stable classes.
Conjecture~\ref{conj:sparsification} suggests 
a way of understanding those classes.

\paragraph*{Our results.} In this work we prove Conjecture~\ref{conj:sparsification} for classes of bounded twin-width. More precisely, the main result is the following.

\begin{theorem}\label{thm:main}
 Every class of graphs that is monadically stable and has bounded twin-width is transduction equivalent to a class of bounded sparse twin-width. 
\end{theorem}

An immediate corollary of Theorem~\ref{thm:main} is the following.

\begin{corollary}\label{cor:sparsification-tww}
    Let $\cal P$ be any $\fole$-downward closed property of classes of graphs such that every class enjoying $\cal P$ has bounded twin-width.
    Then every monadically stable class  $\CC\in\cal P$
    is transduction equivalent to some weakly sparse class $\DD\in\cal P$.
\end{corollary}

Note that Theorems~\ref{thm:sparse-lin-cw} and~\ref{thm:sparse-cw} follow from Corollary~\ref{cor:sparsification-tww}, where as $\cal P$ we consider the properties of having bounded linear cliquewidth and having bounded cliquewidth, respectively. 

\medskip

Our proof of Theorem~\ref{thm:main} is actually very different from the proofs of Theorems~\ref{thm:sparse-lin-cw} and~\ref{thm:sparse-cw}, presented in~\cite{nesetril2021linrw_stable} and~\cite{nesetril2021rw_stable}. These proofs heavily rely on suitable decompositions for the linear cliquewidth and cliquewidth parameters that expose  respectively the path-like and the tree-like structure. The main combinatorial component is a Ramseyan tool --- Simon's factorization~\cite{simon90factorization} and its deterministic variant~\cite{colcombet2007factorization} --- using which the decomposition is analyzed. The assumption about stability is exploited in a rather auxiliary way within this analysis. On the other hand, our reasoning leading to the proof of Theorem~\ref{thm:main} places stability in the spotlight: we use the largest length of a ladder that can be found in a given graph as a complexity measure bounding the depth of induction. Thus, the proof is completely new, more general, and arguably simpler than the ones presented in~\cite{nesetril2021linrw_stable,nesetril2021rw_stable} for classes of bounded (linear) cliquewidth.

A priori, Theorem~\ref{thm:main} provides no direct implications for Conjecture~\ref{conj:sparsification}. However, we believe that the general scheme of reasoning, and in particular the form of a decomposition implicitly constructed in the proof, may be insightful for the future work in the context of arbitrary monadically stable classes.

\medskip

Finally, we observe that our work has implications in the context of $\chi$-boundedness. We say that a graph class $\CC$ is {\em{$\chi$-bounded}} if there exists a function $f\colon \N\to \N$ such that for every graph $G\in \CC$ we have $\chi(G)\leq f(\omega(G))$, where $\chi(G)$ is the {\em{chromatic number}} of $G$ --- the minimum number of colors needed for a proper coloring of $G$ --- and $\omega(G)$ is the {\em{clique number}} of $G$ --- the maximum number of pairwise adjacent vertices in $G$. The concept of $\chi$-boundedness was introduced by Gy\'arf\'as in~\cite{gyarfas} as a relaxation of perfectness, and has since grown to be one of major notions of interest in contemporary structural graph theory. The reason is that $\chi$-boundedness typically witnesses the well-structuredness in the considered graph class, and trying to establish this property is a perfect excuse to understand the structure of studied graphs better. Also, there is a variety of $\chi$-bounded graph classes originating from different settings, for instance geometric intersection graphs, graphs admitting certain decompositions, or graphs excluding fixed induced subgraphs. We invite the reader to the recent survey of Scott and Seymour~\cite{ScottS20} for a broader introduction.

Coming back to our work, we note that by combining Theorem~\ref{thm:main} with the results of~\cite{gajarsky2020sbe} one can conclude that monadically stable classes of bounded twin-width are {\em{linearly $\chi$-bounded}}, that is, $\chi$-bounded with a linear $\chi$-bounding function $f$.

\begin{theorem}\label{thm:lin-chi}
 Let $\CC$ be a class of graphs that is monadically stable and has bounded twin-width. Then there exists a constant $c\in \N$ such that
 $\chi(G)\leq c\cdot \omega(G)$, for all $G\in \CC$.
\end{theorem}

It is known that classes of bounded twin-width are $\chi$-bounded~\cite{bonnet2020tww3}. Without the assumption of monadic stability, the $\chi$-bounding function cannot be expected to be linear, see~\cite{BonamyP20,nesetril2021linrw_stable}, but it is open whether it can be polynomial~\cite{bonnet2020tww3}. Linear $\chi$-boundedness of monadically stable classes of bounded cliquewidth has been established in~\cite{nesetril2021rw_stable} using a reasoning similar to the one presented here.

While Theorem~\ref{thm:lin-chi} can be seen as a consequence of Theorem~\ref{thm:main}, in Section~\ref{sec:lin-chi} we give a self-contained proof of this result. This proof can be seen as a light-weight and purely combinatorial version of the proof of Theorem~\ref{thm:main}, which nevertheless contains many of the key ideas. Therefore, the reader might consider reading Section~\ref{sec:lin-chi} first in order to gather intuition before the main argument, presented in Sections~\ref{sec:main-lemma} and~\ref{sec:main-theorem}.

\paragraph{Structure of the paper and order of reading.}
In Section~\ref{sec:overview}
we give a high-level overview of the main proof, explaining the main ideas. This overview assumes a basic understanding of twin-width and transductions, which are introduced more formally in the preliminaries in Section~\ref{sec:prelims}. In Section~\ref{sec:main-lemma} we present the proof of the main lemma, while in Section~\ref{sec:main-theorem} we use it to prove the main result, Theorem~\ref{thm:main}. In Section~\ref{sec:lin-chi} we directly prove that monadically stable classes of bounded twin-width are linearly $\chi$-bounded. The proof there is independent of the main proof, and can be read independently of Sections \ref{sec:overview},
\ref{sec:main-lemma}, and \ref{sec:main-theorem}.

We finish with Section~\ref{sec:discussion}, where we discuss the broader context of Fig.~\ref{fig:intro}, state multiple conjectures related to it, and make some preliminary observations towards those conjectures.

\section{Overview of the proof}
\label{sec:overview}
We now present the main ideas behind the proof of Theorem~\ref{thm:main}. 
Beware that the description below is not completely accurate, but it should convey the main ideas. All the notions  discussed below are introduced formally in the preliminaries in Section~\ref{sec:prelims}.

Let $\CC$ be a monadically stable class of graphs of bounded twin-width.
Our task is to exhibit two transductions  $S$ and $T$ such that $S(\CC)$ is a class of bounded sparse twin-width and $\CC \subseteq T(S(\CC))$.  We focus on proving the following weaker statement: There exists a class $\DD$ of bounded sparse twin-width and a transduction $T$ such that $\CC \subseteq T(\DD)$. The stronger statement then follows easily from results of~\cite{gajarsky2020sbe}
(see proof of Theorem~\ref{thm:main}
on p. \pageref{proof:main}).

\medskip

Our goal is therefore the following. Given a graph $G \in \CC$, construct a graph $R$ such that:
\begin{enumerate}[nosep]
    \item $R$ omits some biclique as a subgraph,
    \item $R$ has small twin-width,  and
    \item $G$ can be obtained from $R$ by some transduction $T$.
\end{enumerate}
Crucially, the excluded biclique, the bound on twin-width, and the transduction $T$ should depend only on $\CC$ and not on the particular choice of $G$. 

For technical reasons it is more convenient to work with bipartite graphs $G$ rather than usual graphs. As every class  of graphs 
is transduction equivalent with a class of bipartite graphs (see Lemma~\ref{lem:bipartite-reduction}), and transductions preserve stability and bounded twin-width,
this allows us to reduce our  problem to the case of bipartite graphs.

\paragraph{Ladder index.} A key conceptual ingredient of our approach is to measure the complexity of bipartite graphs on which we induct in terms of the largest size of a ladder that can be found in them. More precisely, if $G$ is a bipartite graph with sides $L$ and $R$, then the {\em{ladder index}} of $G$ is the largest size of a ladder that can be found in $G$ where one side is contained in $L$ and the other in $R$. For technical reasons, in the actual proof we work with a functionally equivalent notion of the {\em{quasi-ladder index}}; the difference is immaterial for the purpose of this overview.

Since $G$ belongs to the fixed class $\CC$ that is monadically stable, in particular $G$ excludes some ladder as a semi-induced subgraph, so the ladder index of $G$ is bounded by a constant depending on $\CC$ only. This allows us to use the ladder index as a measure of progress in an inductive argument, as always inducting on subgraphs with a smaller ladder index yields a reasoning with constant induction depth.

\paragraph{High level description.}
Let us now describe the main construction, of $R$ from $G$, on a high level in order to introduce the necessary concepts.
Using the contraction sequence (sequence of partitions witnessing bounded twin-width) of $G$ `in reverse order' --- starting from the bipartition of $G$ and in each step splitting one part into two --- we construct
a partition $\cal F$ of $V(G)$ and a graph $G'$ with the same vertex set as $G$. These have the following properties.
First, $G'$  can be obtained 
from $G$ by applying a bounded number of 
\emph{flips} (complementations of the edge relation between a subset of the left side and a subset of the right side). Second, the \emph{quotient  graph} $H\coloneqq \quo{G'}{\cal F}$ is sparse. Here $\quo{G'}{\cal F}$ is a graph on vertex set $\cal F$ where two parts $A,B\in \cal F$ are adjacent if and only if in $G'$ there exists an edge with one endpoint in $A$ and second in $B$.
\emph{Sparsity} of $H$ means in particular that $H$ can be edge-partitioned into a bounded number of induced star forests (disjoint unions of stars), say $H=F_1\cup\cdots\cup F_s$.
Importantly,
each star $S$ of $F_i$, say with vertices $K_0,\ldots,K_m\in\cal F$ and center $K_0$,
induces in $G$ a bipartite subgraph $G_S\coloneqq G[K_0,K_1\cup\cdots\cup K_m]$ 
(with $K_0$ on one side and $K_1\cup\cdots\cup K_m$ on the other) of ladder index strictly smaller than that of $G$.
Hence, we can induct on each graph $G_S$, and thus 
 represent it by a sparse graph $R_S$ which has bounded twin-width, omits a fixed biclique as a subgraph, and from which $G_S$ can be recovered using a fixed transduction.
We then combine all the graphs~$R_S$, for all stars $S$ in the star forests $F_1,\ldots,F_s$, yielding the sparse graph $R$ from which $G$ can be recovered by a transduction.

We now give some more details concerning the techniques used to bound the twin-width of $R$ and the sizes of bicliques in $R$. Then we explain the main lemma, which from $G$ produces the graph $G'$ and the partition $\cal F$.

\paragraph{Bounding the twin-width.}
One way of showing that the constructed graph $R$ has bounded twin-width is 
 to explicitly construct a contraction sequence of bounded width for $R$. Another way is to exhibit 
a vertex-ordering which avoids a fixed grid-minor (a certain pattern in the adjacency matrix).
While these approaches could work in our proof, 
  we use yet another approach, namely we show that the graph $R$ can be obtained from $G$ using a fixed transduction. By the results of~\cite{bonnet2020tww}, this implies that the twin-width of $R$ is bounded in terms of the twin-width of $G$ (and the transduction).
Note that our transduction involves additionally a suitable order $\le$ on $V(G)$, which turns $G$ into an ordered bipartite graph $(G,\le)$ of bounded twin-width. Such an order always exists, and is easily obtained from a contraction sequence for $G$.
In fact, we may use any order $\le$ on $V(G)$ such that all parts in the contraction sequence are convex with respect to $\le$. We call such an order a \emph{compatible} order on $G$.

Hence, to accomplish our goal, we achieve the following, alternative goal: given a bipartite graph $G\in \CC$ with a compatible order $\le$, construct a graph $R$ such that:
\begin{enumerate}[nosep]
    \item $R$ omits some biclique as a subgraph,
    \item $R$ can be obtained from $(G,\le)$ by some fixed transduction,  and
    \item $G$ can be obtained from $R$ by some  fixed transduction.
\end{enumerate}
Then by~\cite{bonnet2020tww}, a fixed bound on the twin-width of $G$ entails a fixed bound on the twin-width of $R$.

\paragraph{Bounding the bicliques.}
Instead of directly constructing a graph $R$ which omits a fixed biclique as a subgraph, we
 construct a \emph{$t$-equivalence structure} $\str S$: a set furnished with $t$ equivalence relations, where $t$ is a constant depending only on $\CC$. Such a structure can be represented by a 
  graph $R_{\str S}$ whose vertex set comprises of all the elements of $\str S$, plus for each equivalence class of each of the $t$ equivalence relations we add a vertex representing this class. Every element $e$ of $S$ is adjacent to each of the $t$ vertices representing the $t$ equivalence classes of which $e$ is a member. Thus, by construction,
$R_{\str S}$ omits $K_{t+1,t+1}$ as a subgraph. Moreover, $R_{\str S}$ can be obtained from $\str S$ using a fixed transduction, and vice-versa
(see Lemma~\ref{lem:equivalences}).
So instead of constructing a graph $R$ as in our previous goal, it is enough to construct a $t$-equivalence structure~$\str S$, for some fixed $t$. 

Hence, our new goal can now be rephrased as follows: 
given a bipartite graph $G\in \CC$ with a compatible order $\le$, construct a $t$-equivalence structure $\str S$, for some fixed $t$, such that:
\begin{enumerate}[nosep]
    \item $\str S$ can be obtained from $(G,\le)$ by some fixed transduction,  and
    \item $G$ can be obtained from $\str S$ by some fixed transduction.
\end{enumerate}

\medskip

With the ground prepared, we now explain the statement of our main lemma (Lemma~\ref{lem:main}). Then we describe how the main lemma is applied to achieve the goal outlined above, and finally we sketch the proof of the main lemma.

 \paragraph{Statement of the main lemma.}
 Recall that we are given a bipartite graph $G$ of bounded twin-width, with a compatible order $\leq$, and we assume that the ladder index of $G$ is bounded, say it is equal to $k$. The main lemma intuitively states that by applying a bounded number of flips one can ``sparsify'' $G$ a bit, so that afterwards it can be covered by a sparse network of subgraphs of strictly smaller ladder index. Formally, the main lemma provides a graph $G'$, on the same vertex set as $G$, and a partition $\cal F$ of $V(G)$, with the following properties satisfied:
 \begin{itemize}[nosep]
  \item Every part of $\cal F$ is contained in either the left side or the right side of $G$. Moreover, every part of $\cal F$ is convex in $\leq$. 
  \item $G'$ can be obtained from $G$ by applying a bounded number of flips. Note that thus, $G'$ can be transduced from $G$ using a fixed transduction.
  \item Define the {\em{quotient graph}} $H\coloneqq \quo{G'}{\cal F}$ on vertex set $\cal F$ as described before: parts $A,B\in \cal F$ are adjacent in $H$ if in $G'$ there is an edge with one endpoint in $A$ and second in $B$. Then $H$ is sparse, and in particular it has a bounded {\em{star chromatic number}}: it is possible to color $H$ with a bounded number of colors so that every pair of colors induces a star forest.
  \item Consider any star $S$ in any star forest $F$ among the ones described above. Say $S$ has center $K_0$ and petals $K_1,\ldots,K_m$, where $K_0,K_1,\ldots,K_m\in \cal F$. Then the bipartite subgraph $G[K_0,K_1\cup\ldots\cup K_m]$ induced by $S$ has ladder index strictly smaller than $k$.  
 \end{itemize}
This summarizes the statement of the main lemma.
That the parts of $\cal F$ are convex in $\leq$ will be important for constructing the final $t$-equivalence structure from $G$ by means of a transduction. 

\paragraph{Applying the main lemma.}
 We now explain how the main lemma is used to achieve our final goal: transducing from $(G,\le)$ a $t$-equivalence structure $\str S$, for some fixed $t$, so that $G$ can be recovered from $\str S$ by a transduction.
 This description corresponds to the proof of Lemma~\ref{lem:main-transduction}.
 
 First, $G'$ can be transduced from $G$ by applying a bounded number of flips.
 Thanks to the convexity of the parts in $\cal F$, the equivalence relation corresponding to the partition $\cal F$ can be constructed by a transduction, by using a unary predicate marking the smallest element in each part of $\cal F$. Having $G'$ and $\Ff$, we can interpret the edge relation of the quotient graph $H=\quo{G'}{\cal F}$, hence we can imagine that it is available for further transductions. Next, a star coloring of $H$ with a bounded number of colors can be guessed by introducing a bounded number of unary predicates. Let $F_1,\ldots,F_s$ be the star forests induced by pairs of colors of this coloring. Note that for every $i\in \{1,\ldots,s\}$, we can also transduce the equivalence relation of being in the same star of the star forest $F_i$. This is because stars have bounded radius.
 
Summarizing, we can use the main lemma to obtain the following equivalence relations from $G$ by means of a transduction:
\begin{itemize}[nosep]
\item A relation $\sim$ such that $u \sim v$ if and only if $u$ and $v$ are in the same part of $\cal F$.
\item For each $i \in 1,\ldots,s$ a relation $\sim_i$ such that $u \sim_i v$ if and only if $u$ and $v$  belong to the same star of the star forest $F_i$.
\end{itemize}
 As the bipartite graph $G[K_0,K_1\cup\cdots \cup K_m]$ induced by any star in any star forest $F_i$ has a strictly smaller ladder index, we can apply induction on it. Thus we may encode $G[K_0,K_1\cup\cdots \cup K_m]$  using a $t'$-equivalence structure, for some fixed $t'$ obtained from induction for a strictly smaller ladder index. While there can be arbitrarily many stars in each forest $F_i$,  they are disjoint and so their $t'$-equivalence structures can be merged together to form a single $t'$-equivalence structure which represents all edges of $G$ between any two parts $A,B$ in $F_i$.
 This $t'$-equivalence structure is additionally expanded with the equivalence relation $\sim_i$, yielding a $(t'+1)$-equivalence structure.
 Doing this for all star forests $F_i$ and overlaying the results, we obtain the desired $t$-equivalence structure~$\str S$, where $t=s(t'+1)$.
  
  To sum up, the structure $\str S$ can be obtained from $(G,\le)$ using a transduction (here we rely on convexity of the parts of $\cal F$ and the 
  bounded radius of the stars).
  Conversely,
 each of the bipartite graphs $G[K_0,K_1\cup\cdots \cup K_m]$ can be recovered from $\str S$ by inductive assumption. In particular, each of the bipartite graphs $G[A,B]$, for parts $A,B\in \cal F$ which are adjacent in $H$, can be reconstructed from $\str S$, whereas for parts $A,B\in\cal F$ which are non-adjacent in $H$, the graph $G[A,B]$ can be obtained by reverting the bounded number of flips that were used to obtain $G'$ from~$G$. Therefore, we can recover $G$ from $\str S$ using a transduction. Hence, our goal is achieved, proving the main result, Theorem~\ref{thm:main}.

\paragraph{Proof of the main lemma.}
Recall that we work with a bipartite graph $G$ of bounded twin-width, say $d$, and bounded {ladder index}, say $k$.
As $G$ has twin-width $d$, it has an \emph{uncontraction sequence} of width $d$.
This is a sequence of partitions of $V(G)$ which starts with the partition into two parts --- the left and the right side of $G$ --- and in each step splits some part into two, eventually reaching a discrete partition. That the uncontraction sequence has width $d$ means that at every step, every part is {\em{impure}} towards at most $d$ other parts, in the sense that the parts are neither complete nor anti-complete towards each other. Also, at every point, all parts of the current partition are convex in the compatible order $\leq$.

We follow the uncontraction sequence and apply a mechanism of {\em{freezing}} parts; this mechanism is inspired by the proof of $\chi$-boundedness of classes of bounded twin-width~\cite{bonnet2021tww2}. Specifically, when we consider any time moment in the uncontraction sequence, a part $A$ of the current partition gets frozen at this moment if the following condition is satisfied:
\begin{quote}
 For every part $B$ belonging to the other side of $G$, the induced bipartite graph $G[A,B]$ has ladder index strictly smaller than $k$. 
\end{quote}
We remark that once a part $A$ gets frozen, it still participates in further uncontractions, but no descendant part of $A$ will be frozen again. That is, we freeze a part only if none of its ancestors were frozen before. Since the uncontraction sequence ends with a discrete partition, it is not hard to see that the collection of parts which got frozen at any point forms a partition of the vertex set of $G$. This is the partition $\cal F$ provided by the lemma.

Note that every element of $\cal F$ is convex in $\leq$, because at the moment of freezing it was a member of a partition in the uncontraction sequence. Further, the elements of $\cal F$ can be naturally ordered by their freezing times. Denote this order by $\preceq$ and note that it is unrelated with the compatible order $\leq$ on $V(G)$.

The next step in the proof is an analysis of the properties implied by the freezing mechanism, with the goal of understanding the interaction between the parts $\cal F$. Omitting some technicalities, this analysis yields the following conclusion: if for a part $B\in \cal F$ we consider all parts $A\in \cal F$ with $A\prec B$, then there is a set $S(B)\subseteq \{A\colon A\prec B\}$ of {\em{exceptional parts}} that has bounded size, and otherwise $B$ is either complete or anti-complete towards $\bigcup_{A\prec B} A\setminus \bigcup S(B)$. Therefore, with each part $B\in \cal F$ we can associate the {\em{type}} of $B$, which is $+$ if $B$ is complete towards $\bigcup_{A\prec B} A\setminus \bigcup S(B)$, and $-$ if it is anti-complete.

Consider now the sequence of types of the elements of $\cal F$, as ordered by $\preceq$. This is a sequence over symbols $\{+,-\}$. It turns out that there can be only a bounded number of {\em{alternations}} in this sequence --- switches from $+$ to $-$ or vice versa --- for otherwise we can find a large ladder in $G$. Therefore, the sequence of types can be partitioned into a bounded number of blocks, each consisting of the same symbols. From this one can define a bounded number of flips --- one per each block of symbols $+$ --- that intuitively ``flip away'' all the complete interactions signified by $+$ symbols. Applying these flips turns $G$ into the graph $G'$ that the lemma returns.

Once $\cal F$ and $G'$ are defined, it remains to analyze the quotient graph $H=\quo{G'}{\cal F}$. From the construction it follows that whenever parts $A$ and $B$, say with $A\prec B$, are adjacent in $H$, $A$ must be an exceptional part for $B$, that is, $A\in S(B)$. This means that the $\preceq$ ordering is an ordering of bounded degeneracy for the graph $H$, so in particular $H$ is sparse. With more insight into the properties implied by the freezing condition, it is possible to prove that $\preceq$ has not only bounded degeneracy, but even bounded {\em{strong $2$-coloring number}}. From the classic construction of Zhu~\cite{zhu2009colouring} it then follows that $H$ has a bounded star chromatic number.

The star coloring with a bounded number of colors obtained from the argument above is almost what we wanted. More precisely, from the freezing condition it easily follows that for every pair of parts $A,B\in \cal F$ contained in distinct sides of $G$, the induced bipartite graph $G[A,B]$ has ladder index strictly smaller than $k$. This is because if say $A\prec B$, then at the moment of freezing $A$, $B$ was contained in some ancestor part $B'\supseteq B$, and the fact that $A$ got frozen at this point implies that the ladder index of $G[A,B']$ is strictly smaller than $k$. Therefore, if $S$ is a star in any of the induced star forests coming from the star coloring, say with center $K_0$ and petals $K_1,\ldots,K_m$, then each of the induced bipartite subgraphs $G[K_0,K_i]$ has ladder index strictly smaller than $k$. However, the goal was to obtain this conclusion for  the whole subgraph $G[K_0,K_1\cup \ldots\cup K_m]$ induced by the star $S$. A priori this condition may fail, but we can again use the properties provided by the freezing condition to show that each star forest can be edge-partitioned into a bounded number of subforests that already satisfy the desired property.

This concludes the sketch of the proof of the main lemma and this overview.

\section{Preliminaries}\label{sec:prelims}
For every natural $n\ge0$, the set $\set{1,\ldots,n}$ is denoted $[n]$.
By \emph{order} we mean \emph{total order}.
A \emph{convex subset} of an ordered set $X$ is a subset $U$ of $X$ such that $x\le y\le z$ and $x,z\in U$ implies $y\in U$. 

\subsection{Graphs}
We consider finite, undirected, and simple graphs.
The vertex set and the edge set of a graph $G$ are denoted $V(G)$ and $E(G)$, respectively.
If $G$ is a graph and $X\subset V(G)$ is a set of its vertices, then the subgraph of \emph{$G$ induced by $X$}
is the graph $G[X]$ with vertex set $X$ such that two vertices $x,y\in X$ are adjacent in $G[X]$ if and only if they are adjacent in $G$. 
If $G$ and $H$ are graphs, we say that $G$ is \emph{$H$-free} if it does not contain $H$ as a subgraph.

A {\em{bipartite graph}} is a tuple $G=(V,L,R,E)$ such that $(V,E)$ is a graph, $L$ and $R$ form a partition of $V$ and 
every edge in $E$ has one endpoint in $L$ and one endpoint in $R$.
The sets $L$ and $R$ are the \emph{sides} of $G$. Note that whenever we speak about a bipartite graph, the bipartition $(L,R)$ is considered fixed and provided with the graph. When $G$ is a bipartite graph with sides $L$ and $R$, and $X\subseteq L$ and $Y\subseteq R$ are subsets of the sides, then by $G[X,Y]$ we denote the induced bipartite subgraph whose sides are $X$ and $Y$ and whose edge set comprises of all edges of $G$ with one endpoint in $X$ and the other in $Y$. 

An \emph{ordered bipartite graph} is a tuple $G=(V,L,R,E,\le)$ such that $(V,L,R,E)$ is a bipartite graph and $\le$ is a total order on $V$ such that every vertex in $L$ is smaller than every vertex in $R$.

A {\em{division}} of an ordered bipartite graph $G$ with sides $L$ and $R$ is a partition $\cal F$ of the vertex set of $G$ such that each part of $\cal F$ is convex and is entirely contained either in $L$ or in $R$. Then by $\cal F^L$ and $\cal F^R$ we denote the partitions of $L$ and $R$ consisting of parts of $\cal F$ contained in $L$ and $R$, respectively. We also define the {\em{quotient graph}} $\quo{G}{\cal F}$ as the graph on vertex set $\cal F$ where $A\in \cal F^L$ and $B\in \cal F^R$ are adjacent if and only if there are $a\in A$ and $b\in B$ that are adjacent in $G$. Note that thus, $\quo{G}{\cal F}$ is a bipartite graph with sides $\cal F^L$ and $\cal F^R$.


Let $G$ be a graph and $X,Y\subseteq V(G)$ be two disjoint subsets of vertices. We say that the pair $X,Y$ is {\em{complete}} if every vertex of $X$ is adjacent to every vertex of $Y$, and {\em{anti-complete}} if there is no edge with one endpoint in $X$ and the other in $Y$. The pair $X,Y$ is {\em{pure}} if it is complete or anti-complete, and {\em{impure}} otherwise. If $X,Y$ is pure, then its {\em{purity type}} is $+$ if it is complete, and $-$ if it is anticomplete.

A {\em{flip}} of a bipartite graph $G$ with sides $L$ and $R$ is any graph $G'$ obtained from $G$ by taking any subsets $X\subseteq L$ and $Y\subseteq R$ and flipping the adjacency relation in $X\times Y$: all edges $xy$ with $x\in X$ and $y\in Y$ become non-edges, and all such non-edges become edges.
    We shall also say that $G'$ is obtained from $G$ by {\em{flipping the pair $X,Y$}}.
    Note that a flip is still a bipartite graph with sides $L$ and $R$. For $q\in \N$, we say that $G'$ is a {\em{$q$-flip}} of $G$ if $G'$ it can be obtained from $G$ by applying the flip operation at most $q$ times, that is, there is a sequence $G=G_0,G_1,\ldots,G_{q'}=G'$ such that $q'\leq q$ and $G_i$ is a flip of $G_{i-1}$ for each $i\in [q']$.

\paragraph{Generalized coloring numbers.}Let $G$ be a graph and $\le$ be an order on its vertices. Fix a number $r\in\N$.
For two vertices $v$ and $w$ of $G$, we say that 
$w$ is \emph{strongly $r$-reachable} from $v$ (with respect to $\le$) if $w\leq v$ and there is a path of length at most $r$ in $G$ connecting $v$ and $w$ such that all vertices on the path apart from $v$ and $w$ are larger than $v$ in $\leq$.
Similarly, $w$ is \emph{weakly $r$-reachable} from $v$ if $w\leq v$ and there is a path of length at most $r$ in $G$ connecting $v$ and $w$ such that $w$ is the least (with respect to $\le$) vertex on that path.
We define $\SReach_r^{G,\le}[v]$ to be the set of vertices which are $r$-reachable from $v$ and analogously $\WReach_r^{G,\le}[v]$ to be the set of vertices which are weakly $r$-reachable from $v$. Finally, we define $\scol_r(G,\le)$ and $\wcol_r(G,\le)$ as follows:

 $$ \scol_r(G,\le) = \max_{v \in V(G)}|\SReach_r^{G,\le}[v]| \qquad  \wcol_r(G,\le) = \max_{v \in V(G)}|\WReach_r^{G,\le}[v]| $$


As shown by Zhu~\cite{zhu2009colouring}, weak and strong $r$-coloring numbers are functionally equivalent in the following sense: for every graph $G$, order $\leq$ on the vertex set of $G$, and $r\in \N$, we have
\begin{equation}\label{eq:wcol-scol}
\scol_r(G,\le)\le \wcol_r(G,\le)\le \scol_r(G,\le)^r.
\end{equation}
In this work we will need only a bound for the particular case $r=2$.
\begin{proposition}\label{prop:scol-wcol}
For every graph $G$ and order $\leq$ on $V(G)$,
\[
    \wcol_2(G,\le)\le \scol_2(G,\le)+(\scol_1(G,\le)-1)^2.\]
\end{proposition}
\begin{proof}
   Observe that a vertex $w$ is weakly $2$-reachable from a vertex $v$ 
if and only if it is either strongly $2$-reachable from $v$, or is strongly $1$-reachable from a vertex $v'$ which is strongly $1$-reachable from $v$, where $v'$ is different from $v$ and $w$.
\end{proof}

We will also need the connection between weak $2$-coloring number and star colorings. This connection was also established by  Zhu~\cite{zhu2009colouring}, but we repeat his reasoning in order to make some technical assertions explicit. 
\begin{lemma}\label{lem:star-coloring}
    Let $G$ be a graph and $\le$ an order on $V(G)$.
    There is a coloring $\lambda\from V(G)\to [p]$ using $p\coloneqq \wcol_2(G,\le)$ colors such that for every two colors $c,d\in [p]$,
    every connected component $D$ of $G[\lambda^{-1}(\set{c,d})]$ 
is a star whose center is the vertex of $V(D)$ which is least with respect to $\le$. 
\end{lemma}
\begin{proof}    
    Color $V(G)$ greedily using $p=\wcol_2(G,\le)$ colors as follows: process the vertices in the order $\leq$ from smallest to largest, and assign to each $v\in V(G)$ 
    any color that is not present among the vertices $w$ that are different from $v$ and weakly $2$-reachable from $v$. (Note that these vertices were colored earlier.) Since every vertex weakly $2$-reaches at most $p-1$ vertices other than itself, $p$ colors are sufficient to construct such a coloring. Call it $\lambda\colon V(G)\to [p]$.

    Fix $c,d\in[p]$. Observe that no two adjacent vertices have the same color, since one is weakly $1$-reachable from the other. 
    In particular, sets $\lambda^{-1}(c)$ and $\lambda^{-1}(d)$ both induce edgeless subgraphs of $G$. 

    Let $D$ be the vertex set of a connected component of $G[\lambda^{-1}\set{c,d}]$ and let $u$ be the $\le$-minimal element of $D$. Without loss of generality assume $\lambda(u)=c$. Then for every $v\in D$ that is adjacent to $u$ in $G$ we must have $\lambda(v)=d$, implying in particular that all neighbors of  $u$ in $D$ are pairwise non-adjacent in $G$. Suppose now that some $v\in D$ is simultaneously adjacent to $u$ and to some other $w\in D$. By the choice of $u$ we have that $u$ is $2$-weakly reachable from $w$, so $w$ cannot have color $c$ by construction. But $w$ also cannot have color $d$ due to being adjacent to $v$, a contradiction. This implies that $D$ consists only of $u$ and the neighbors of $u$, hence $G[D]$ is a star with the center being the $\le$-minimal element. 
\end{proof}

We will consider classes of {\em{bounded expansion}}, which is a notion of uniform sparsity in graphs. There are multiple equivalent definitions of this notion --- see the monograph of Ne\v{s}et\v{r}il and Ossona de Mendez~\cite{sparsity}, for a broad introduction --- but for the purpose of this paper it is sufficient to rely on a characterization through coloring numbers.

\begin{definition}
    A class of graphs $\CC$ has \emph{bounded expansion} if 
    for every $r\in\N$ there is a constant $c_r$ such that for every graph $G\in\CC$ there is an order $\le$ on $V(G)$ such that $\scol_r(G,\le)$ is at most $c_r$.   
\end{definition}
By~\eqref{eq:wcol-scol}, replacing $\scol_r$ by $\wcol_r$ would yield an equivalent definition.

\paragraph{Ladder index.} We now introduce notions inspired by model theory, which intuitively define graphs where no large total order can be found.

\begin{definition}
    Let $H$ be a bipartite graph with sides $L$ and $R$. 
    A \emph{ladder} of order $k$ in $H$ of consists of two sequences 
    $x_1,\ldots,x_k\in L$ and $y_1,\ldots,y_k\in R$ such that for all $i,j\in \set{1,\ldots,k}$,
    $x_i$ is adjacent to  $y_j$ if and only if $i\leq j$ (see Fig.~\ref{fig:ladders}).
    A \emph{quasi-ladder} order $k$ in $H$ consists of two sequences
    $x_1,\ldots,x_k\in L$ and $y_1,\ldots,y_k\in R$ such that
    for each $i\in \set{1,\ldots,k}$, 
    one of the following conditions holds:
    \begin{itemize}
        \item
$x_i$ is adjacent to all of $y_1,\ldots,y_{i-1}$ and $y_i$ is non-adjacent to all of $x_1,\ldots,x_{i-1}$; or
\item $x_i$ is non-adjacent to all of $y_1,\ldots,y_{i-1}$ and $y_i$ is adjacent to all of $x_1,\ldots,x_{i-1}$.
    \end{itemize}
    The \emph{ladder index} (resp. \emph{quasi-ladder index}) of a bipartite graph 
    is the largest $k$ such that $H$ contains a ladder (resp. a quasi-ladder) of order $k$.
\end{definition}

Note that in the definition above we do {\em{not}} require vertices $x_1,\ldots,x_k$ or $y_1,\ldots,y_k$ to be different.

\medskip

The ladder index is more commonly used in the literature, but in this work we will find it useful to work with the quasi-ladder index. The next lemma clarifies that the two notions are functionally equivalent.

\begin{lemma}
\label{lem:quasiladder}
    The following inequalities hold for every bipartite graph $H$:
\[\textit{ladder-index}(H)\le 
\textit{quasi-ladder-index}(H)
\le 4\cdot \textit{ladder-index}(H)+4.\]
\end{lemma}
\begin{proof}
    The first inequality is immediate, since every ladder is also a quasi-ladder. For the second inequality we show that a quasi-ladder of order $\ell=4(k+1)$ contains a ladder of order $k$.

Let $x_1,\ldots,x_{\ell}$ and $y_1,\ldots, y_{\ell}$ form a quasi-ladder of order $\ell$ in $H$. Let $I\subset\set{1,\ldots,\ell}$ 
    be the set of those indices $i\in\set{1,\ldots,\ell}$ for which 
    $x_i$ is adjacent to all of $y_1,\ldots,y_{i-1}$ and $y_i$ is non-adjacent to all of $x_1,\ldots,x_{i-1}$,
    and let $J=\set{1,\ldots,\ell}\setminus I$.
Let $A\subseteq\set{1,\ldots,\ell}$ be the set of those indices $i$ for which $x_i$ is adjacent to $y_i$, and let $B=\set{1,\ldots,\ell}\setminus A$ be its complement.
    Then $\set{1,\ldots,\ell}$ is the disjoint union of $I\cap A, J\cap A, I\cap B$, and $J\cap B$,
    so one of those sets must contain at least $k+1$ elements, by the choice of $\ell$.

If $|I\cap A|\ge k+1$, then for any distinct elements $\set{i_1,\ldots,i_k}$ of $I\cap A$, 
 the sequences $x_{i_1},\ldots,x_{i_k}$ and $y_{i_1},\ldots,y_{i_k}$ form a ladder of order $k$ in $H$.
If $|I\cap B|\ge k+1$, then for any distinct $k+1$ elements $\set{i_0,\ldots,i_k}$ of $I\cap B$, the sequences $x_{i_0},\ldots,x_{i_{k-1}}$ and $y_{i_1},\ldots,y_{i_k}$ form a ladder of order $k$ in $H$.
We proceed symmetrically in the cases when $|J\cap A|\ge k+1$ and $|J\cap B|\ge k+1$, in each case concluding that $H$ has a ladder of order $k$.
\end{proof}

\begin{corollary}
    A class of bipartite graphs $\CC$ has bounded ladder index if and only if it has bounded quasi-ladder index.
\end{corollary}    



As mentioned, in this paper it will be more convenient to work with the quasi-ladder index.
Henceforth, by \emph{index} we mean the quasi-ladder index.

\medskip 

The notion of the (quasi-)ladder index, as discussed above, applies only to bipartite graphs. We may extend the notation to general graphs as follows: if $G$ is a graph, then its (quasi-)ladder index is defined as the (quasi-)ladder index of the bipartite graph whose partite sets are two copies of $V(G)$, and where a vertex $u$ from the first copy is adjacent with a vertex $v$ from the second copy if and only if $u$ and $v$ are adjacent in $G$. Note that, maybe a bit counterintuitively, the (quasi-)ladder index of a bipartite graph $H$ is not necessarily equal to its (quasi-)ladder index when it is considered as a general graph (but it is not hard to see that the two quantities are functionally equivalent). A graph class $\CC$ is called {\em{graph-theoretically stable}} if there is a constant~$k$ such that the ladder index (equivalently, the quasi-ladder index) of all members of $\CC$ is upper bounded by $k$.

\subsection{Logic}
\paragraph{Structures.}
We only consider relational signatures consisting 
of unary and binary relation symbols.
We may say that a structure $\str A$ is a \emph{binary} structure to emphasize that its signature is such. A binary structure $\str A$ is \emph{ordered}
if its signature contains the symbol $\le$ which is interpreted in $\str A$ as a total order.

Graphs are viewed as binary structures over the signature consisting of one binary relation~$E$ signifying the adjacency relation. 
Similarly, bipartite graphs are viewed as binary structures equipped with the binary relation $E$ and unary relations $L$ and $R$ marking the two parts of the graph.
Ordered bipartite graphs are viewed as binary structures equipped with the binary relations $E$ and $\le$ and unary relations $L$ and $R$.

\paragraph{Interpretations.}
Interpretations are a means of producing 
new structures out of old ones, where
 each relation of the new structure is defined
 by a fixed first-order formula.
In this work we only consider a restricted fragment which are sometimes called
\emph{simple interpretations}.

Fix relational signatures $\Sigma$ and $\Gamma$.
A (simple) \emph{interpretation}  $I\from \Sigma\to \Gamma$ consists of a \emph{domain} formula $\delta(x)$ and for each $R\in\Gamma$ of arity $k$, 
a formula $\phi_R(x_1,\ldots,x_k)$.
The \emph{output} of such an interpretation $I$ on a given $\Sigma$-structure $\str A$ is the $\Gamma$-structure $I(\str A)$
with domain $B$ consisting of all
elements $a$ of $\str A$ satisfying $\delta(x)$ in $\str A$, 
and in which 
every relation symbol $R\in\Gamma$ of arity $k$ is interpreted as the set of tuples $(a_1,\ldots,a_k)\in B^k$ satisfying $\phi_R(x_1,\ldots,x_k)$ in $\str A$. If $\CC$ is a class of $\Sigma$-structures then $I(\CC)$ denotes the class of structures $I(\str A)$, for all $\str A\in\CC$.

A class $\DD$ is an \emph{interpretation} of a class $\CC$ if there is an interpretation $I$ such that $\DD\subset I(\CC)$.

\medskip
The following standard result says that interpretations are closed under compositions.
\begin{proposition}\label{prop:composition}
    Let $I\from \Sigma\to\Gamma$  and $J\from \Gamma\to\Delta$ be  interpretations. There is an interpretation $J\circ I\from  \Sigma\to \Delta$ such that  $(J\circ I)(\str A)=J(I(\str A))$, for all $\Sigma$-structures $\str A$.
\end{proposition}
\begin{corollary}
\label{cor:interps_compose}
    If $\CC''$ is an interpretation of $\CC'$ and $\CC'$ is an interpretation of $\CC$ then $\CC''$ is an interpretation of $\CC$.
\end{corollary}

\paragraph{Transductions.}
 For a $\Sigma$-structure $\str A$ and $k\in\N$, let $k\times \str A$ denote the structure 
 obtained from $\str A$ by taking the disjoint union 
 of $k$ copies of $\str A$, and expanding it by a fresh binary relation $M$ which relates any two copies of the same element.

 A transduction is an operation which inputs a structure, copies it a fixed number of times, then
 adds some unary predicates in an arbitrary way,
 and finally, applies a fixed interpretation, thus obtaining an output structure. This is formalized below.

 A \emph{transduction} $T\from  \Sigma\to \Gamma$ consists of:
 \begin{itemize}[nosep]
     \item a number $k\in\N$,
     \item unary relation symbols $U_1,\ldots,U_\ell$,
     \item an interpretation $I\from \Sigma'\to \Gamma$, where $\Sigma'=\Sigma\sqcup\set{U_1,\ldots,U_\ell}\sqcup\set{M}$.
 \end{itemize}
 The transduction $T$ is \emph{copyless} if $k=1$. It is \emph{domain-preserving} if it is copyless, and the the domain formula of the underlying interpretation is $(x=x)$.

 Given a $\Sigma$-structure $\str A$ and a $\Gamma$-structure $\str B$, we say that $\str B$ is an \emph{output} of $T$ on $\str A$ if there is 
 some unary expansion $\wh{\str A}$ of the structure $k\times \str A$ such that $\str B=I(\wh{\str A})$.
 Let $T(\str A)$ denote the set of all structures $\str B$ which are outputs of $T$ on $\str A$.
 If $\CC$ is a class of structures then $T(\CC)$ denotes $\bigcup_{\str A\in\CC} T(\str A)$. For a class $\DD$, we say that $\DD$ can be \emph{transduced} from $\CC$ if there is a transduction $T$ such that $\DD\subseteq T(\CC)$. Two classes $\CC$ and $\DD$ are {\em{transduction equivalent}} if each can be transduced from the other.
 
 Like interpretations, transductions are closed under compositions. The following result is standard, see e.g. \cite[Lemma 2]{gajarsky2020sbe} for a proof.
 \begin{lemma}
    [Composition of transductions]
     Let $T\from\Sigma\to\Gamma$ and  $T'\from \Gamma\to\Delta$ be  transductions. There is a transduction \[(T'\circ T)\from \Sigma\to\Delta\] such that $K(\str A)=T'(T(\str A))$, for all $\Sigma$-structures $\str A$.
 \end{lemma}

 \begin{corollary}\label{cor:transductions-compose}
     If $\CC''$ can be transduced from  $\CC'$ and $\CC'$ can be transduced from $\CC$, then $\CC''$ can be transduced from $\CC$.
 \end{corollary}

 We will use two additional operations on transductions.
Let $\str B_1,\str B_2$ be a structures over signatures $\Gamma_1$ and $\Gamma_2$, respectively. Suppose furthermore that $\str B_1$ and $\str B_2$ have the same domain $B$.
Let $\str B_1\&\str B_2$ denote the structure over the signature $\Gamma_1\sqcup \Gamma_2$
(the disjoint union of $\Gamma_1$ and $\Gamma_2$)
with domain $B$, obtained by superimposing the structures $\str B_1$ and $\str B_2$. The following lemma is straightforward.
\begin{lemma}[Combination of transductions]\label{lem:super-impose}
    Suppose $T_i\from \Sigma\to\Gamma_i$ are domain-preserving transductions for $i=1,2$.
    Then there is a domain-preserving transduction \[T_1\&T_2\from \Sigma\to(\Gamma_1\sqcup\Gamma_2)\]
    such that $(T_1\&T_2)(\str A)=\setof{\str B_1\&\str B_2}{\str B_1\in T_1(\str A),\str B_2\in T_2(\str A)}$.
\end{lemma}
\begin{proof}
    For $i\in\set{1,2}$, 
    let the interpretation underlying $T_i$ consist of formulas $\phi_R$ for ${R\in \Gamma_i}$.
    Suppose $T_1$ introduces unary predicates $U_1,\ldots,U_\ell$ while $T_2$ introduces unary predicates $U_{\ell+1},\ldots,U_m$.
    Then $T_1\&T_2$ introduces unary predicates $U_1,\ldots,U_m$ and its underlying interpretation consists of the formulas $\phi_R$, for ${R\in \Gamma_1\sqcup\Gamma_2}$.
\end{proof}

The following lemma allows to apply a single transduction in parallel on pairwise disjoint subsets of a given structure $\str A$, coming from a definable partition of $\str A$.
If $\approx$ is an equivalence relation on a subset $X$ of the domain of a $\Sigma$-structure $\str A$
and $J\from \Sigma\to\Gamma$ is a transduction,
then by \[\coprod_{C\in X/{\approx}}J(\str A[C])\]
we denote the set of $\Gamma$-structures  of the form $\coprod_{C\in X/{\approx}}\str B_C$, where $\str B_C\in J(\str A[C])$ for $C\in X/{\approx}$,
and $\str A[C]$ denotes the substructure of $\str A$ induced by $C$. Note that the domains of the structures $\str B_C$ are pairwise disjoint, for $C\in X/{\approx}$, since the sets $C\in X/{\approx}$ are pairwise disjoint and transductions preserve disjointness of domains.

The following lemma is essentially \cite[Lemma 29]{gajarsky2020sbe}.
\begin{lemma}[Parallel application of transductions]\label{lem:parallel-application}
    Suppose $J\from\Sigma\to\Gamma$ is a transduction and $\Sigma'=\Sigma\sqcup\set{\approx}$, where $\approx$ is a binary relation symbol.
    Then there is a transduction 
    \[(\coprod_\approx J)\from \Sigma'\to\Gamma\]
    such that for every $\Sigma'$-structure $\str A$ in which $\approx$ defines an equivalence relation on a subset $X$ of the domain of $\str A$,
    \[(\coprod_\approx J)(\str A)=\coprod_{C\in X/{\approx}}J(\str A[C])\]
    (the $\Sigma'$-structure $\str A[C]$ above is naturally treated as a $\Sigma$-structure).
\end{lemma}

\paragraph{Monadic stability and structurally bounded expansion.}
 
The following definition and theorem are from~\cite{gajarsky2020sbe}. We remark that the notion of transduction used there slightly differs from our definition, as it uses unary functions. However, unary functions can be modelled as binary relations and so the results from~\cite{gajarsky2020sbe} apply in our setting.  More explicitly, Theorem~\ref{thm:SBE} below follows immediately from Proposition~18 of~\cite{gajarsky2020sbe} after removing the adjectives ``quantifier-free'' and ``almost quantifier-free''.
\begin{definition}
A class $\CC$ of graphs has \emph{structurally bounded expansion} if it is a transduction of a class of bounded expansion.
\end{definition}


\begin{theorem}[\cite{gajarsky2020sbe}]
\label{thm:SBE}Let $\CC$ be a class of graphs which has structurally bounded expansion.
There is a pair of transductions $S,S'$ such that 
$S(\CC)$ is a class of bounded expansion and 
$G\in S'(S(G))$, for all $G\in \CC$.
In particular, $\CC$ is transduction equivalent with a class of bounded expansion.
\end{theorem}

We now come to one of the central notions in this paper, which originates from the work of Baldwin and Shelah~\cite{baldwinshelah}.

\begin{definition}
    A class of structures $\CC$ is \emph{monadically stable}
    if there is no transduction $T$ such that $T(\CC)$ contains 
    the class of all ladders (see Fig.~\ref{fig:ladders}).
\end{definition}

It is known that every nowhere dense class of graphs, in particular, every class of bounded expansion, is monadically stable~\cite{adler2014interpreting}. Hence, every class of structurally bounded expansion is monadically stable, by Corollary~\ref{cor:transductions-compose}.

We will use the following straightforward characterization of monadic stability.
\begin{lemma}\label{lem:mon-stable-transduction}
    Let  $\CC$ be a class of structures.
    Then $\CC$ is not monadically stable if and only if there is a transduction  $T$ which outputs bipartite graphs such that $T(\CC)$ has unbounded ladder index.
\end{lemma}
\begin{proof}
The left-to-right implication is trivial, since if $\CC$ is not monadically stable then 
there is a transduction $T$ such that $T(\CC)$ is the class of all ladders, and those (viewed as bipartite graphs) have unbounded ladder index.

For the right-to-left implication, suppose $T(\CC)$ is a class of bipartite graphs of unbounded ladder index, for some transduction $T$. 
As a bipartite graph of ladder-index $k$ contains a ladder of length $k$ as an induced substructure, there is a transduction $S$ such that $S(T(\CC))$ contains all ladders. 
    By Corollary~\ref{cor:transductions-compose}, $\CC$ is not monadically stable.
\end{proof}

Obviously, if a class of graphs is monadically stable, then it is also graph-theoretically stable. The converse is not necessarily true, as witnessed by the class of $1$-subdivided ladders. However, it turns out that if one restrict attention to monadically dependent classes, the two notions coincide.

\begin{theorem}[\cite{nesetril2021rw_stable}]\label{thm:mon-stab}
    A class of graphs $\CC$ is monadically stable if and only if it is monadically dependent and is graph-theoretically stable.
\end{theorem}

\subsection{Twin-width}
Let $\str A$ be a binary structure.
Generalizing the graph notation, we say that a pair of disjoint subsets $X,Y$ of the domain of $\str A$ is \emph{pure} if for every binary relation symbol $R\in \Sigma$, either $R(x,y)$ holds for all $x\in X$ and $y\in Y$, or $\lnot R(x,y)$ holds for all $x\in X$ and $y\in Y$.



\begin{definition}
An \emph{uncontraction sequence of width $d$} of a binary structure $\str A$ is a sequence $\cal P_1,\ldots,\cal P_n$ of partitions of the domain of $\str A$ such that:
\begin{itemize}[nosep]
    \item $\cal P_1$ is a partition with one part only;
    \item $\cal P_n$ is a partition into singletons;
    \item for $t=1,\ldots,n-1$,
    the partition $\cal P_{t+1}$ is obtained from $\cal P_t$  by splitting exactly one of the parts into two;
    \item for every part $U\in \cal P_t$ there are at most $d$  parts $W\in \cal P_t$ other than $U$ for which the pair $U,W$ is not pure.
\end{itemize} 
The \emph{twin-width} of $\str A$ is the least $d$ such that there is an uncontraction sequence of $\str A$ of width~$d$.
\end{definition}

We remark that the original definition of~\cite{bonnet2020tww} considers {\em{contraction sequences}}, which are reversals of uncontraction sequences. In this work it will be convenient to reverse the way of thinking, similarly as in~\cite{bonnet2021tww2,bonnet2020tww3,dreier2021twinwidth}.

Note that our definition of an uncontraction sequence completely ignores the unary predicates in the structure $\str A$. 

In the case of ordered bipartite graphs we will consider uncontraction sequences tailored to them:

\begin{definition}
A \emph{convex uncontraction sequence of width $d$} of an ordered bipartite graph $G$ with sides $L$ and $R$ is a sequence $\cal P_1,\ldots,\cal P_{n}$ of divisions of the vertex set of $G$ such that:
\begin{itemize}[nosep]
    \item $\cal P_1$ is a division with two parts $L$ and $R$;
    \item $\cal P_{n}$ is a division into singletons;
    \item for $t=1,\ldots,n-1$,
    the division $\cal P_{t+1}$ is obtained from $\cal P_t$  by splitting exactly one of the parts into two; and
    \item for every part $U\in \cal P^L_t$, there are at most $d$ parts $W\in \cal P^{R}_t$ for which the pair $U,W$ is impure, and the symmetric condition holds also for the parts of $\cal P^R_t$.
\end{itemize}  The \emph{convex twin-width} of a bipartite graph $G$ is the least $d$ such that $G$ has a convex uncontraction sequence of width $d$.
\end{definition}
 
Note that in convex uncontraction sequences of bipartite graphs we have $n=|L| + |R| -1 = |V(G)|-1$.

It is easily seen that one can turn a convex uncontraction sequence $\cal P_1,\ldots,\cal P_{n}$ of width $d$ of an ordered bipartite graph $G$ into an uncontraction sequence $\cal P_1',\ldots,\cal P_{n+1}'$  of $G$ regarded as a binary structure by setting $\cal P_1' = V(G)$ and $\cal P_{i+1}'= \cal P_{i}$ for $i=1,\ldots,n$. This transformation preserves the width.
This immediately implies the following.

\begin{lemma}
\label{lem:btww-tww}
If an ordered bipartite graph $G$ has convex twin-width at most $d$, then regarded as a binary structure, it has twin-width at most $d$.
\end{lemma}

We also need a converse.
\begin{lemma}
\label{lem:conv_tww}
    If a bipartite graph $G$ has twin-width at most $d$ when regarded as a binary structure,
then there is an ordering $\le$ on $V(G)$ 
such that $G$ equipped with $\le$ is an ordered bipartite graph 
of convex twin-width at most $d+1$.
\end{lemma}
\begin{proof}
Let $\cal P_1, \ldots, \cal P_n$ be an uncontraction sequence of $G$ of width $d$ and let $L$ and $R$ be the  sides of $G$. In particular, $n=|V(G)|$.  We will first construct a sequence $\cal P_1', \ldots, \cal P_{n-1}'$ of partitions of width at most $d+1$ such that $\cal P_1' = \{L,R\}$ and each part of any $\cal P_i'$ is a subset of $L$ or $R$, and then we will construct an ordering $\le$ of $V(G)$ such that each part of each $\cal P_i'$ is convex with respect to $\le$.

For any subset $S$ of $V(G)$ let $S^L$ and $S^R$ denote the sets $S \cap L$ and $S \cap R$,  respectively.
For every $i \in [n]$ let $\cal P_i'$ denote the partition of $V(G)$ obtained from $\cal P_i$  by replacing each part $A \in \cal P_i$ by two parts, $A^L$ and $A^R$.  Note that if $A\in \cal P_i$ is impure with respect to parts $B_1,\ldots,B_k\in\cal P_i$, where $k\le d$, then each part $B'$ in $\cal P_i'$ such that $A^L$ is impure with respect to $B'$ 
is among $B_1^R,\ldots,B_k^R$ and~$A^R$.
A symmetric statement holds for $A^R$.
Hence, each part in $\cal P_i'$ is impure towards at most $d+1$ parts.

The sequence of partitions $\cal P_1', \ldots, \cal P_n'$ does not have the property that for each $i$ the partition $\cal P_{i+1}'$ is obtained from $\cal P_i'$ by splitting exactly one part of $\cal P_i'$ into two.  We therefore adjust $\cal P_1', \ldots, \cal P_n'$ as follows for each $i\in [n]$: 
\begin{itemize}[nosep]
\item If $\cal P_{i+1}'=\cal P_i'$, then we remove $\cal P_{i+1}'$ from the sequence.
\item If $\cal P_{i+1}'$ is obtained from $\cal P_i'$ by splitting exactly one part, then we do nothing.
\item If $\cal P_{i+1}'$ differs from $\cal P_i'$ by splitting parts $A \subseteq L$ and $B \subseteq R$ into $A_1,A_2$ and $B_1,B_2$, then we add an intermediate partition between $\cal P_{i}'$  and $\cal P_{i+1}'$ which differs from $\cal P_i'$ by splitting $A$ into $A_1,A_2$.
\end{itemize}

After this, we adjust the indices to account for removed and added partitions and obtain $\cal P_1', \ldots, \cal P_{n-1}'$ in which each $\cal P_{i+1}'$ is obtained from $\cal P_i'$ by splitting exactly one part of $\cal P_i'$ into two and in which every part of every $\cal P_i'$ is either in $L$ or in $R$. The width of $\cal P_1', \ldots, \cal P_{n-1}'$ remains bounded by $d+1$.

It remains to construct an ordering of $V(G)$ so that each part in each $\cal P_i'$ is convex. 
Let $\le_1$  be an order in which all vertices in $L$ are before all vertices in $R$, and within $L$ and $R$ the vertices are ordered arbitrarily. 
For $i>1$ we construct $\le_i$ from $\le_{i-1}$ as follows. If $\cal P_i'$ is obtained from $\cal P_{i-1}'$ by splitting $A$ into $A_1$ and $A_2$, then we reorder the vertices in the convex interval corresponding to $A$ so that all vertices in $A_1$ are before all vertices in $A_2$ (and the vertices within the convex subintervals corresponding to $A_1$ and $A_2$ are ordered arbitrarily).  We take $\le$ to be $\le_{n-1}$. It follows from the construction that each $\cal P_i'$ is a division with respect to this order.
\end{proof}

\medskip

Bounded twin-width is preserved by transductions:
\begin{theorem}[\cite{bonnet2020tww}]
\label{thm:interp_bd_tww}
If a class of binary structures $\CC$ can be transduced from a class of bounded twin-width, then $\CC$ also has bounded twin-width.
\end{theorem}

Since there are graphs of arbitrarily high twin-width, from Theorem~\ref{thm:interp_bd_tww} it follows that every class of bounded twin-width is monadically dependent. Hence, by Theorem~\ref{thm:mon-stab}, the notions of monadic stability and graph-theoretic stability coincide for classes of bounded twin-width.

Since by Lemma~\ref{lem:btww-tww} bounded convex twin-width implies bounded twin-width, we also get the following corollary.
\begin{corollary}
\label{cor:interp_tww}
If a class of binary structures $\CC$ can be transduced from a class of ordered bipartite graphs of bounded convex twin-width, then $\CC$ has bounded twin-width.
\end{corollary}

Finally, let us remark that not every class of bounded twin-width is stable,
as witnessed by the class of ladders.


\paragraph{Sparse twin-width.}
The following definition, proposed in~\cite{bonnet2021tww2}, introduces a restriction of the concept of twin-width to sparse graphs.
    
\begin{definition}
    A class $\CC$ of graphs has  \textit{bounded sparse twin-width} if there exist integers $d$ and $s$ such that every $G \in \CC$ has twin-width at most $d$ and does not contain $K_{s,s}$ as a subgraph.
        \end{definition}

It turns out that classes of bounded sparse twin-width are also sparse in the bounded expansion sense.
        
\begin{theorem}[\cite{bonnet2021tww2}]
\label{thm:sparse_tww_be}
    Every class of graphs of bounded sparse twin-width has bounded expansion.
\end{theorem}
The converse implication does not hold,
as witnessed by the class of cubic graphs which has bounded expansion, but does not have bounded twin-width~\cite{bonnet2020tww}.

\begin{proposition}
\label{prop:int_stable}
If a class of binary structures $\CC$ can be transduced from a class of graphs of bounded sparse twin-width, then $\CC$ is monadically stable and has bounded twin-width.
\end{proposition}
\begin{proof}
    Classes transducible from classes of bounded expansion (even from nowhere dense classes) are monadically stable~\cite{adler2014interpreting}, while classes transducible from classes of bounded twin-width have bounded twin-width by Theorem~\ref{thm:interp_bd_tww}.
\end{proof}

Our main result, Theorem~\ref{thm:main}, proves the converse to Proposition~\ref{prop:int_stable}:
every monadically stable class of bounded twin-width is a transduction of a class of bounded sparse twin-width.

\medskip
Note that in general, classes of bounded twin-width 
 are monadically dependent but are not necessarily monadically stable, as the class of all ladders has bounded twin-width. In particular, by Theorem~\ref{thm:mon-stab}, a class of bounded twin-width is graph-theoretically stable if and only if it is monadically stable. Hence, for simplicity,
we will sometimes 
talk about \emph{stable classes of bounded twin-width}, referring to monadically stable classes of bounded twin-width.
And so, our main result states that every stable class of bounded twin-width can be obtained from a class of bounded sparse twin-width by a  transduction, proving a converse of Proposition~\ref{prop:int_stable}.

\newcommand{\rv}[1]{\overline{#1}}

\section{Main lemma}\label{sec:main-lemma}
The following  lemma is our main technical tool. It says that every ordered bipartite graph of bounded convex twin-width and bounded (quasi-ladder) index
has a certain decomposition.
 This decomposition will be used in the next section to prove Theorem~\ref{thm:main}. 

 %





\begin{lemma}\label{lem:main}
    For all $k,d\in\N$, $k,d\geq 2$, there are $\ell,q\in \N$
    satisfying the following.
    Let $G$ be an ordered bipartite graph of convex  twin-width at most $d$ and quasi-ladder index at most $k$, with sides $L$ and $R$.
    Then there is a division $\cal F$ of $G$, 
    sets $\cal U_1,\ldots,\cal U_\ell\subset \cal F$, and an $q$-flip $G'$ of $G$ such that the following holds for $H\coloneqq \quo{G'}{\cal F}$:
    \begin{enumerate}[label=(\arabic*),ref=(\arabic*),leftmargin=*]          
        \item\label{c:impure} For every edge $AB$ of $H$ there exists $i\in [\ell]$ such that $A,B\in \cal U_i$.
        \item\label{c:star-simpler} Each set $\cal U_i$, $i\in [\ell]$, induces in $H$ a star forest. Moreover, for each star in this star forest, say with center $C$ and leaves $K_1,\ldots,K_m$, the index of $G[C,K_1\cup\cdots\cup K_m]$ is smaller than $k$.
    \end{enumerate}
\end{lemma}

The remainder of this section is devoted to the proof of Lemma~\ref{lem:main}. Whenever we speak about purity or impurity of some pair of sets of vertices, we mean purity or impurity in the graph $G$. Recall that if a pair of nonempty subsets $A\subseteq L$ and $B\subseteq R$ is pure, then its {\em{purity type}}  is $+$ if $A,B$ is complete, and $-$ if $A,B$ is anti-complete. A pair $A\subseteq L$ and $B\subseteq R$ {\em{matches}} a purity type $\sigma\in \{+,-\}$ if the pair is pure and of purity type $\sigma$. Otherwise $A,B$ {\em{mismatches}} $\sigma$. Note that if $A,B$ is impure, then it mismatches both purity types. 
Finally, when $Q\in \{L,R\}$, then by $\rv{Q}$ we denote $R$ if $Q=L$ and $L$ if $Q=R$.

\medskip

Fix $k,d\in\N$ with $k,d\geq 2$. We first resolve a corner case when $|L|\leq d$ or $|R|\leq d$. By symmetry suppose that $|L|\leq d$. Observe that the edgeless bipartite graph $G'$ with sides $L$ and $R$ is a $d$-flip of $G$. Indeed, it suffices flip the pairs $\{u\},N_G(u)$ for all $u\in L$. So, in this case we may take $\cal F$ to be the trivial division that puts every vertex into a separate part, $G'$ and $H$ to be edgeless graphs, and $\ell=0$ (that is, there are no sets $\cal U_1,\ldots,\cal U_\ell$). Hence, from now on we assume that $|L|>d$ and $|R|>d$.

Recall that $G$ is an ordered bipartite graph, hence the sides $L,R$ of $G$ are convex in $G$. Further, since $G$ has convex  twin-width at most $d$, there is a convex uncontraction sequence $\cal P_1,\ldots,\cal P_n$ of $G$ of width at most $d$, where $n\coloneqq |L\cup R|-1$.
Note that for each $t\in \{2,\ldots,n\}$, the difference between divisions $\cal P_t$ and  $\cal P_{t-1}$
is that one part of $\cal P_{t-1}$ is replaced by two its subsets in $\cal P_t$ (and all the other parts are the same).


For $s\leq t$, we say that a part $A\in\cal P_{s}$ is an \emph{ancestor} of a part $B\in\cal P_t$ if $A\supseteq B$. Then also $B$ is a {\em{descendant}} of $A$. Note that if $B\in \cal P_t$, then for each $s\leq t$ there is a unique ancestor of $B$ in $\cal P_s$. Note also that every part is considered an ancestor and a descendant of itself.

The following definition is crucial in our reasoning and is inspired by the proof of the $\chi$-boundedness of graphs of bounded twin-width, presented in~\cite{bonnet2020tww3}.
A part $A\in \cal P_t$, say belonging to $\cal P_t^{Q}$ where $Q\in \{L,R\}$, is \emph{frozen} at time $t$ if the following conditions hold:
\begin{itemize}
    \item no ancestor of $A$ was frozen at any time $s<t$, and
    \item for every $B\in \cal P^{\rv{Q}}_t$, the index of $G[A,B]$ is smaller than $k$.
\end{itemize}
For $Q\in \{L,R\}$, let $\cal F^Q_t\subseteq \cal P^Q_t$ be the set of parts of $\cal P^Q_t$ frozen at time $t$.
We note the following.

\begin{lemma}\label{lem:freezing-tiny}
 For every $t\in [n]$ and $Q\in \{L,R\}$, we have $|\cal F^Q_t|\leq d$.
\end{lemma}
\begin{proof}
We prove the claim for $Q=L$, the proof in the other case is symmetric. 

Since $|\cal P^L_1|=1$, the claim holds trivially for $t=1$, hence assume $t>1$.
Note that $\cal P_t$ differs from $\cal P_{t-1}$ in that there are two parts $C,D\in \cal P_t$ that in $\cal P_{t-1}$ are replaced by $C\cup D$, and otherwise all the parts of $\cal P_t$ and $\cal P_{t-1}$ are the same. We consider two cases: either $C,D\in \cal P^L_t$ or $C,D\in \cal P^{R}_t$.

In the first case we have $\cal P^{R}_{t-1}=\cal P^{R}_t$ and $\cal P^L_{t-1}=\cal (P^L_t\setminus \{C,D\})\cup \{C\cup D\}$. Consider any $A\in \cal F^L_t$. Since $A$ got frozen at time $t$ and not at time $t-1$, it must be the case that $A\in \{C,D\}$. It follows that $|\cal F^L_t|\leq 2\leq d$.

In the second case we have $\cal P^{R}_{t-1}=\cal (P^{R}_t\setminus \{C,D\})\cup \{C\cup D\}$ and $\cal P^L_{t-1}=\cal P^L_t$. Again, consider any $A\in \cal F^L_t$ and note that since $A$ got frozen at time $t$ and not at time $t-1$, it must be the case that the index of $G[A,C\cup D]$ is equal to $k$. As $k\geq 2$, this implies that the pair $A,C\cup D$ is impure. By the assumption on the width of the uncontraction sequence, there are at most $d$ such parts $A$ in $\cal P^L_{t-1}$, implying that $|\cal F^L_t|\leq d$.
\end{proof}

For $t\in [n]$ we denote
$$\cal F_t\coloneqq \cal F^L_t\cup \cal F^R_t.$$
Further, let $\cal F$ be the set of all parts frozen at any moment, that is,
$$\cal F\coloneqq \cal F_1\cup \cdots \cup \cal F_n.$$
We observe the following.

\begin{lemma}\label{lem:Fdivision}
    $\cal F$ is a  division of $G$.
\end{lemma}
\begin{proof}
Consider a vertex $u\in L$.
As $\set u$ is a part of $\cal P^L_n$ that
satisfies the second condition in the definition of a frozen part, it follows that either $\set u$ is frozen at time $n$, or some ancestor of $\set u$ got frozen at some earlier time. Either way, $u$ belongs to some frozen part, so we conclude that $\bigcup \cal F\supseteq L$. A symmetric argument shows that $\bigcup \cal F\supseteq R$ as well.

Next, we argue that the elements of $\cal F$ are pairwise disjoint. Consider any $A,B\in \cal F$, $A\neq B$.
If $A,B\in \cal F_t$ for some $t\in [n]$, then $A$ and $B$ are different parts of the division $\cal P_t$, hence they are disjoint. Suppose then that $A\in \cal F_s$ and $B\in \cal F_t$ for some $s<t$. Note that $B$ has an ancestor $B'\in \cal P_s$. Since $B$ is frozen at time $t$, it follows that $B'$ is not frozen at time $s$, hence $A\neq B'$. Then $A$ and $B'$ are different parts of the division~$\cal P_s$, hence they need to be disjoint, which implies that $A$ and $B$ are disjoint as well.

Finally, all the elements of $\cal F$ are convex and entirely contained either in $L$ or in $R$, as they originate from the divisions~$\{\cal P_t\colon t\in [n]\}$.
\end{proof}

Note that the proof of Lemma~\ref{lem:Fdivision} relies only on the property that elements of $\cal F$ are pairwise not bound by the ancestor/descendant relation. The particular choice of the freezing condition --- which in our case is based on measuring the indices of subgraphs induced by pairs of parts --- is motivated by the following observation.

\begin{lemma}\label{lem:index-single-edge}
 For each $A\in \cal F^L$ and $B\in \cal F^R$, the index of $G[A,B]$ is smaller than $k$.
\end{lemma}
\begin{proof}
 Let $s,t\in [n]$ be such that $A\in \cal F^L_s$ and $B\in \cal F^R_t$. Without loss of generality assume that $s\leq t$. Let $B'$ be the unique ancestor of $B$ in $\cal F^R_s$. As $A$ is frozen at time $s$, the index of $G[A,B']$ is smaller than $k$, which implies that the index of $G[A,B]$ is smaller than $k$ as well.
\end{proof}

We remark that in the sequel we will not rely only on Lemma~\ref{lem:index-single-edge}, but also on its stronger variants that take multiple parts of $\cal F$ into account.


%
%

Let $\tau$ be the least positive integer such that
$$\left|\cal P^L_\tau\right|>d\qquad\textrm{and}\qquad \left|\cal P^R_\tau\right|>d.$$
Note that since $|\cal P^L_n|=|L|>d$ and $|\cal P^R_n|=|R|>d$, $\tau$ is well-defined and we have $\tau\leq n$. Also, by minimality we have
\begin{equation}\label{eq:wydra0}
\left|\cal P^L_\tau\right|=d+1\qquad\textrm{or}\qquad \left|\cal P^R_\tau\right| = d+1.
\end{equation}

%

We now analyze the properties of $\cal F$ implied by the construction. 
The first lemma presents a key observation about the adjacencies between  parts that are not yet frozen and the rest of the graph.

\begin{lemma}\label{lem:index}
 Let $Q\in \{L,R\}$, $t\geq \tau$, and $B\in \cal P^Q_t$ be such that $B$ is not frozen at time $t$, and no ancestor of $B$ was frozen at any time $s<t$. Let $\cal N$ be the set of all those parts $A\in \cal P^{\rv{Q}}_t$ for which the pair $A,B$ is impure. Denote
 $$W\coloneqq \rv{Q}\setminus \bigcup \cal N.$$
 Then the pair $B,W$ is pure.
\end{lemma}
\begin{proof}
 We give a proof for the case $Q=L$, the other case is symmetric. Thus, we have $B\in \cal P^L_t$ and $W\subseteq R$.

 First, consider any $u\in W$ and suppose $u$ is impure towards $B$; that is, the pair $\{u\},B$ is impure. Then the part of $\cal P^R_t$ to which $u$ belongs must form an impure pair with $B$, hence it is contained in $\cal N$. But $\bigcup \cal N$ is disjoint with $W$. This contradiction shows that every $u\in W$ is pure towards $B$.
 
 We now prove that the pair $B,W$ is pure.
 Suppose this is not the case. Then from the observation of the previous paragraph it follows that there exist vertices $u^-,u^+\in W$ such that $u^-$ is non-adjacent to all the vertices of $B$, while $u^+$ is adjacent to all the vertices of $B$.
 
 Since $B$ is not frozen at time $t$, nor it has an ancestor frozen earlier, there exists $C\in \cal P^R_t$ such that $G[B,C]$ has index exactly $k$. Let then $x_1,\ldots,x_k\in B$ and $y_1,\ldots,y_k\in C$ be a quasi-ladder of length $k$ in $G[B,C]$. Since $t\geq \tau$, we have $|\cal P_t^L|>d$.
 On the other hand, by the assumption on the width of the uncontraction sequence, there are at most $d$ parts $M\in \cal P_t^L$ for which the pair $M,C$ is impure. Therefore, there exists $D\in \cal P^L_t$ such that the pair $D,C$ is pure. Now if the pair $D,C$ is complete, then by selecting any $x_{k+1}\in D$ and setting $y_{k+1}\coloneqq u^-$, we obtain sequences $x_1,\ldots,x_k,x_{k+1}$ and $y_1,\ldots,y_k,y_{k+1}$ that form a quasi-ladder of length $k+1$ in $G$, a contradiction. If the pair $D,C$ is anti-complete, then setting $y_{k+1}\coloneqq u^+$ yields a contradiction in the same way.
\end{proof}

Let
$$\cal S\coloneqq \bigcup_{s\leq \tau} \cal F_s.$$
Consider any $A\in \cal F\setminus \cal S$, say $A\in \cal F^Q_t$ for some $t>\tau$ and $Q\in \{L,R\}$. We define the {\em{type}} of $A$, denoted $\tp(A)\in \{+,-\}$, as follows. Let $B$ be the unique ancestor of $A$ in $\cal P^Q_{t-1}$. Noting that $B$ satisfies the prerequisites of Lemma~\ref{lem:index}, we let $\tp(A)$ be the purity type of the pair $B,W$, where $W$ is defined as in Lemma~\ref{lem:index} (the lemma also asserts that this pair is pure).
Note here that it will never be the case that $W$ is empty. This is because due to $\tau>t$ we have $|\cal P^{\rv{Q}}_{t-1}|>d$, while the set $\cal N$ defined in the statement of Lemma~\ref{lem:index} has cardinality at most $d$.

The next observation will be a crucial combinatorial tool for the analysis of pairs $A,B\in \cal F$ that mismatch the type of $B$, where $B$ is frozen later than $A$. It will be reused several times in the sequel.

\begin{lemma}\label{lem:sandwich}
 Let $Q\in \{L,R\}$ and let $x,y,z\in [n]$ be such that $x\leq y<z$ and $y\geq \tau$. Further, let $X,Y,Z\subseteq V(G)$ be such that:
 \begin{itemize}
  \item $X\in \cal F^{\rv{Q}}_x$, $Y\in \cal P^Q_y$, and $Z\in \cal F^Q_z$;
  \item $Y$ is an ancestor of $Z$; and
  \item the pair $X,Z$ mismatches the type $\tp(Z)$.
 \end{itemize}
 Then there exists $U\in \cal P^{\rv{Q}}_y$ such that $U$ is a descendant of $X$ and the pair $U,Y$ is impure.
\end{lemma}
\begin{proof}
 We consider the case $Q=L$, the proof in the other case is symmetric.

 Let $Z'$ be the unique ancestor of $Z$ in $\cal P^L_{z-1}$. 
 Note that $Z'$ is also a descendant of $Y$.
 Let us define sets $\cal N_{Y}$ and $\cal N_{Z'}$ as in the statement of Lemma~\ref{lem:index}:
 \begin{itemize}
  \item $\cal N_{Y}$ comprises all parts $D\in \cal P^R_{y}$ such that the pair $D,Y$ is impure.
  \item $\cal N_{Z'}$ comprises all parts $D\in \cal P^R_{z-1}$ such that the pair $D,Z'$ is impure.
 \end{itemize}
 Observe the following: for each $D\in \cal N_{Z'}$, the unique ancestor $D'$ of $D$ in $\cal P^R_{y}$ belongs to $\cal N_{Y}$. Indeed, the pair $D,Z'$ is impure by the definition of $\cal N_{Z'}$, so as $D'$ is an ancestor of $D$ and $Y$ is an ancestor of $Z'$, it follows that the pair $D',Y$ is impure as well.
 
 This observation implies that if we define
 $$W_{Y}\coloneqq R\setminus \bigcup \cal N_{Y}\qquad \textrm{and}\qquad W_{Z'}\coloneqq R\setminus \bigcup \cal N_{Z'},$$
 then $W_{Z'}\supseteq W_{Y}$. Noting that $Y$ and $Z'$ satisfy the prerequisites of Lemma~\ref{lem:index}, we infer that the pairs $Y,W_{Y}$ and $Z',W_{Z'}$ are pure.
 Further, as $y\geq \tau$, we have $|\cal P^R_{y}|>d$, which together with $|\cal N_Y|\leq d$ implies that $W_{Y}$ is non-empty.
 As $Y\supseteq Z'$ and $W_{Y}\subseteq W_{Z'}$, we conclude that the pairs $Y,W_{Y}$ and $Z',W_{Z'}$ have the same purity type. In other words, the purity type of the pair $Y,W_{Y}$ is equal to $\tp(Z)$.
 
 Recall that the pair $X,Z$ mismatches the type $\tp(Z)$. 
 From $Z \subseteq Y$ and the fact that the pair $Y,W_{Y}$ matches $\tp(Z)$, it follows that $X$ must have a descendant $U$ among the parts of $\cal P^R_{y}$ that are not contained in $W_{Y}$, that is, among the elements of $\cal N_{Y}$.
\end{proof}

Observe that if in an application of Lemma~\ref{lem:sandwich} we have $x=y$, then we necessarily have $U=X$, because $X$ has only one descendant in $\cal P^{\rv{Q}}_y$, namely $X$ itself. Hence, in this case we can simply conclude that the pair $X,Y$ is impure. We will use this particular variant of Lemma~\ref{lem:sandwich} a few times in the sequel.

\medskip

Before we continue with the analysis, we need to take a closer look at the set $\cal S$, which consists of all sets frozen until the time $\tau$ --- the first moment when both partitions $\cal P^L_\tau$ and $\cal P^R_\tau$ contain more than $d$ parts. As observed in~\eqref{eq:wydra0}, at least one of the sets $\cal P^L_\tau$ or $\cal P^R_\tau$ has size $d+1$. We now break the symmetry and assume without loss of generality that the first case holds:
\begin{equation}\label{eq:wydra}
\left|\cal P^L_\tau\right|=d+1.
\end{equation}
With this in mind, we analyze the structure of $\cal S$. 
Denote $\cal S^L\coloneqq \cal S\cap \cal F^L$ and $\cal S^R\coloneqq \cal S\cap \cal F^R$.

\begin{lemma}\label{lem:vc}
We have $|\cal S^L|\leq d+1$.
\end{lemma}
\begin{proof}
 Note that every element of $\cal S^L$ must have at least one descendant in $\cal P^L_\tau$, and these descendants must be pairwise different due to the elements of $\cal S^L$ being pairwise disjoint. It follows that $|\cal S^L|\leq |\cal P^L_\tau|$, and by~\eqref{eq:wydra} we have $|\cal P^L_\tau|=d+1$.
\end{proof}

Note that Lemma~\ref{lem:vc} still leaves the possibility that the cardinality of $\cal S^R$ is very large compared to $d$. This can indeed be the case, but the next lemma shows that this may happen only due to having a large number of twins. Here, two vertices $u,v$ are {\em{twins}} if they belong to the same side of $G$ ($L$ or $R$) and have exactly the same neighbors on the other side.

\newcommand{\sm}{\mathsf{simple}}
\newcommand{\hd}{\mathsf{hard}}

\begin{lemma}\label{lem:simple-hard}
 The set $\cal S^R$ can be partitioned into $\cal S^R_\hd$ and $\cal S^R_\sm$ so that:
 \begin{itemize}
  \item $|\cal S^R_\hd|\leq d(d+1)$; and
  \item $\bigcup \cal S^R_\sm$ can be partitioned into at most $2^{d+1}$ parts, each consisting of twins.
 \end{itemize}
\end{lemma}
\begin{proof}
 Let $\cal N\subseteq \cal P^R_\tau$ be the family of all those sets $D\in \cal P^R_\tau$ for which there is $C\in \cal P^L_\tau$ such that the pair $C,D$ is impure. As $|\cal P^L_\tau|=d+1$ and the uncontraction sequence has width at most~$d$, we have $|\cal N|\leq d(d+1)$. Let $\cal S^R_\hd$ comprise all the elements of $\cal S^R$ that have a descendant in~$\cal N$. Also, let $\cal S^R_\sm\coloneqq \cal S^R\setminus \cal S^R_\hd$. As $\cal S^R\subseteq \cal F$ and every element of $\cal N$ has at most one frozen ancestor, we have $|\cal S^R_\hd|\leq d(d+1)$. We are left with verifying the postulated property of~$\bigcup \cal S^R_\sm$.


 Observe that for every $B\in \cal S^R_\sm$, $b\in B$, and $A\in \cal P^L_\tau$, the pair $A,\{b\}$ is pure. Indeed, otherwise the part $U$ of $\cal P^R_\tau$ that contains $b$ would be a descendant of $B$ such that the pair $A,U$ is impure, implying that $B\in \cal S^R_\hd$. Since $\cal P^L_\tau$ is a partition of $L$, this means that the neighborhood of $b$ in $L$ can be described by stating to which parts of $\cal P^L_\tau$ the vertex $b$ is complete and to which it is anti-complete. As there are at most $2^{|\cal P^L_\tau|}\leq 2^{d+1}$ choices for such a description, it follows that $\bigcup \cal S^R_\sm$ can be partitioned into at most $2^{d+1}$ sets, each consisting only of twins. 
\end{proof}

With all the technical observations prepared, we can proceed to the construction of the graph~$H$.
First, construct an ordering $\preceq$ of $\cal F$ as follows:
\begin{itemize}
 \item The elements of $\cal S$ are placed at the front: $A\prec B$ for all $A\in \cal S$ and $B\in \cal F\setminus \cal S$. Moreover, the elements of $\cal S^L,\cal S^R_\hd,\cal S^R_\sm$ are placed in $\preceq$ in this order, but within each of these sets the elements are ordered arbitrarily.
 \item The elements of $\cal F\setminus \cal S$ are ordered according to their freezing times. That is, whenever $A\in \cal F_s$ and $B\in \cal F_t$ for $\tau<s<t$, we also have $A\prec B$. Note that the elements of a single set $\cal F_t$ are ordered arbitrarily.
\end{itemize}

We define a bipartite graph $H$ with sides $\cal F^L$ and $\cal F^R$ as follows. Consider a pair of distinct sets $A,B\in \cal F$, say $A\prec B$. Then:
\begin{itemize}
 \item If $A,B\in \cal S$, then make $A$ and $B$ adjacent in $H$ if and only if $A\in \cal S^L$, $B\in \cal S^R_\hd$, and the pair $A,B$ is not anti-complete.
 \item Otherwise, that is, if $B\in \cal F\setminus \cal S$, then make $A$ and $B$ adjacent in $H$ if and only if $A\notin \cal S^R_\sm$ and the pair $A,B$ mismatches the type $\tp(B)$.
\end{itemize}
Note that thus, the elements of $\cal S^R_\sm$ are isolated in the graph $H$.
Also, note that the last prerequisite in Lemma~\ref{lem:sandwich} --- that the pair $X,Z$ mismatches the type $\tp(Z)$ --- is implied if we require that $X$ and $Z$ are adjacent in $H$. This is because from the assumptions of the lemma it follows that $Z\in \cal F\setminus \cal S$.

We will later construct a $q$-flip $G'$ of $G$ so that $H=\quo{G'}{\cal F}$, where $q$ is a constant depending only on $d$ and~$k$. However, for now let us focus on studying the properties of $H$ implied by the construction.
The next lemma will be used to control the adjacency in $H$ between parts from a prefix and from the corresponding suffix of the ordering $\preceq$, and it follows quite directly from Lemma~\ref{lem:sandwich}.

\begin{lemma}\label{lem:earlier}
 Let $t>\tau$ and let $B\in \cal P_t$ be such that $B$ has no ancestor frozen at any time $s<t$ (but it may happen that $B$ is frozen at time $t$). Let $\cal F_B\subseteq \cal F$ be the set of frozen descendants of $B$. Then
 $$\left|N_H(\cal F_B)\cap \bigcup_{s<t} \cal F_s\right|\leq d.$$
\end{lemma}
\begin{proof}
 Let $B'$ be the unique ancestor of $B$ in $\cal P_{t-1}$. Consider any sets $A\in \bigcup_{s<t} \cal F_s$ and $C\in \cal F_B$ that are adjacent in $H$. Noting that the prerequisites of Lemma~\ref{lem:sandwich} are satisfied for $(X,Y,Z)=(A,B',C)$, we conclude that there exists a descendant $U$ of $A$ such that $U\in \cal P_{t-1}$ and the pair $B',U$ is impure.
 However, there are at most $d$ such sets $U$, and each of them has at most one frozen ancestor. It follows that the total number of different sets $A$ that can be as above is bounded by $d$.
\end{proof}

From Lemma~\ref{lem:earlier} we can easily derive an upper bound on the degeneracy of $\preceq$.

\begin{lemma}\label{lem:degeneracy}
 For every $B\in \cal F$ there are at most $2d$ sets $A\in \cal F$ such that $A\prec B$ and $A$ and $B$ are adjacent in $H$.
\end{lemma}
\begin{proof}
 Let $t$ be such that $B\in \cal F_t$. We consider two cases: either $t\leq \tau$ or $t>\tau$.
 
 In the first case we have $B\in \cal S$. If $B\in \cal S^L$ then there are no sets $A\prec B$ adjacent to $B$ in $H$. Otherwise $B\in \cal S^R$, so all the sets $A\prec B$ adjacent to $B$ in $H$ are contained in $\cal S^L$, and therefore their number is bounded by $|\cal S^L|\leq d+1$, by~Lemma~\ref{lem:vc}.
 
 In the second case, the sets $A\prec B$ that are adjacent to $B$ in $H$ can be divided into those that belong to $\cal F_t$ and those that belong to $\bigcup_{s<t} \cal F_s$. The number of sets of the first kind is bounded by $d$ by Lemma~\ref{lem:freezing-tiny}, while for the second kind we also have an upper bound of $d$ following from Lemma~\ref{lem:earlier} applied to $B$.
\end{proof}

We can now lift the reasoning presented in Lemmas~\ref{lem:earlier} and~\ref{lem:degeneracy} to give a bound on the strong $2$-coloring number of $\preceq$. 
We remark that a similar reasoning can be also applied to bound the strong $r$-coloring numbers for larger values of $r$, but we will not need this  later on.

\begin{lemma}\label{lem:scols}
 We have
 $$\scol_2(H,\preceq)\leq 2d^2+3d.$$
\end{lemma}
\begin{proof}
 Let $B\in \cal F$.
    We consider three cases: either $B\in \cal S^L$, or $B\in \cal S^R$, or $B\in \cal F\setminus \cal S$.
 
 First, if $B\in \cal S^L$, then $\SReach^{H,\preceq}_2[B]\subseteq \cal S^L$, implying that $|\SReach^{H,\preceq}_2[B]|\leq d+1$.
 
 Second, if $B\in \cal S^R$, then each $A\in \SReach^{H,\preceq}_2[B]$ is either equal to $B$, or belongs to $\cal S^L$, or is an element of $\cal S^R_\hd$. By Lemma~\ref{lem:simple-hard}, the number of elements of the last kind is bounded by $d(d+1)$. Hence, in total we have $|\SReach^{H,\preceq}_2[B]|\leq 1+(d+1)+d(d+1)=d^2+2d+2\leq 2d^2+3d$.
 
 We are left with the main case: $B\in \cal F\setminus \cal S$. Then $B\in \cal F_t$ for some $t> \tau$. 
 The sets $A\in \SReach^{H,\preceq}_2[B]$ can be partitioned into three kinds:
 \begin{itemize}[nosep]
  \item those that belong to $\cal F_t$;
  \item those that belong to $\bigcup_{s<t} \cal F_s$ and are adjacent to $B$ in $H$; and
  \item those that belong to $\bigcup_{s<t} \cal F_s$, are not adjacent to $B$ in $H$, but for which there exists $C\in \cal F$ such that $A\prec B\prec C$ and $C$ is adjacent both to $A$ and to $B$ in $H$.
 \end{itemize}
 Lemmas~\ref{lem:freezing-tiny} and~\ref{lem:earlier} respectively imply upper bounds of $2d$ and $d$ on the numbers of sets of the first two kinds. Therefore, it remains to show that there are at most $2d^2$ sets of the third kind.
 
 For each set $A$ of the third kind let us fix some set $C$ as above. We partition the sets $A$ of the third kind into two further subkinds depending on the placement of $C$:
 \begin{itemize}[nosep]
  \item those for which $C\in \cal F_t$; and
  \item those for which $C\in \bigcup_{t'>t} \cal F_{t'}$.
 \end{itemize}
 For the first subkind, observe that by Lemma~\ref{lem:freezing-tiny}, $\cal F_t$ contains at most $d$ sets $C$ that are adjacent to $B$ in $H$, and each of them has at most $d$ neighbors in $H$ that belong to $\bigcup_{s<t} \cal F_s$, by Lemma~\ref{lem:earlier}. Therefore, there are at most $d^2$ sets $A$ of the first subkind. We are left with proving that the number of sets $A$ of the second subkind is bounded by $d^2$ as well.

 Consider any set $A$ of the second subkind and let $C\in \bigcup_{t'>t} \cal F_{t'}$ be the common neighbor of $A$ and $B$ that has been fixed for $A$. Let $C'$ be the unique ancestor of $C$ in $\cal P_t$. Since $t>\tau$, we may apply Lemma~\ref{lem:sandwich} to $(X,Y,Z)=(B,C',C)$ to infer that the pair $B,C'$ is impure. Now, Lemma~\ref{lem:earlier} applied to $C'$ implies that the descendants of $C'$ that belong to $\cal F$ have at most $d$ different neighbors in $\bigcup_{s<t} \cal F_s$ in total. Noting that there are at most $d$ sets $C'\in \cal P_t$ for which the pair $B,C'$ is impure, we conclude that the total number of sets of the second subkind is at most~$d^2$.
\end{proof}
Let $$p\coloneqq 6d^2+3d.$$ 
Then Proposition~\ref{prop:scol-wcol} together with Lemmas~\ref{lem:degeneracy} and~\ref{lem:scols} implies that $$\wcol_2(H,\preceq)\le \scol_2(H,\preceq)+(\scol_1(H,\preceq)-1)^2\le p.$$
We now may apply Lemma~\ref{lem:star-coloring} to the graph $H$ and thus obtain a coloring $\lambda\colon \cal F\to [p]$ that satisfies the following.
\begin{lemma}\label{lem:star-coloring'}
 For any $i,j\in [p]$, the graph $H[\lambda^{-1}(\{i,j\})]$ is a star forest. Moreover, in each star of this forest that has at least three elements, the $\preceq$-minimum element of the star is the center.
\end{lemma}


Lemma~\ref{lem:star-coloring'} suggests that in the statement of Lemma~\ref{lem:main},
as $\cal U_1,\ldots,\cal U_{\ell}$ we could take sets $\lambda^{-1}(\{i,j\})$ for all $1\leq i<j\leq p$. However, then verifying the last condition of Lemma~\ref{lem:main}, about the index of the bipartite graphs induced by the stars, would be problematic. The next lemma will be used to resolve this issue by appropriately refining the choice of $\cal U_1,\ldots,\cal U_{\ell}$ sketched as above.

\begin{lemma}\label{lem:refine}
 Let $A\in \cal F$ and let $\cal N\subseteq \cal F$ be the set of all $B\in \cal F$ such that $A\prec B$ and $A$ and $B$ are adjacent in $H$. Then $\cal N$ can be partitioned into sets $\cal N_1,\ldots,\cal N_c$ for some $c\leq 2d^2+3d+1$ so that for each $i\in [c]$, the index of $G[A,\bigcup \cal N_i]$ is smaller than $k$.
\end{lemma}
\begin{proof}
 Let $t\in [n]$ be such that $A\in \cal F_t$ and let $t'\coloneqq \max(t,\tau)$.
 Since $\cal N\cap (\cal S^L\cup \cal S^R_\sm)=\emptyset$, by Lemmas~\ref{lem:freezing-tiny} and~\ref{lem:simple-hard} we have
 $$\left|\cal N\setminus \bigcup_{s>t'} \cal F_s\right|\leq d+d(d+1)=d^2+2d.$$
 In the constructed partition of $\cal N$ we put each $B\in \cal N\setminus \bigcup_{s>t'} \cal F_s$ into a separate part $\cal N_i$ consisting only of $B$. Then the index of $G[A,\bigcup \cal N_i]=G[A,B]$ is smaller than $k$ by Lemma~\ref{lem:index-single-edge}. 
 Thus, we are left with partitioning the remaining sets, that is, $\cal N'\coloneqq \cal N\cap \bigcup_{s>t'} \cal F_s$, into at most $d^2+d+1$ parts that satisfy the requested property.
 
 We shall prove the following claim: There exists $\cal M\subseteq \cal P_{t'+1}$ such that $|\cal M|\leq d^2+d+1$ and every $B\in \cal N'$ is a descendant of an element of $\cal M$.
 Before we proceed to the proof, let us verify that the lemma follows from this claim. For each $C\in \cal M$ we define
 $$\cal N_C\coloneqq \cal N'\cap \{B\colon B\textrm{ is a descendant of }C\}.$$
 Then from the claim it follows that $\{\cal N_C\colon C\in \cal M\}$ is a partition of $\cal N'$. Observe that for each $C\in \cal M$ the index of $G[A,C]$ is smaller than $k$, because this holds for the graph $G[A,C']$, where $C'$ is the unique ancestor of $C$ in $\cal P_t$. Since $\bigcup \cal N_C\subseteq C$, it follows that the index of $G[A,\bigcup \cal N_C]$ is smaller than $k$, for each $C\in \cal M$. As $|\cal M|\leq d^2+d+1$, $\{\cal N_C\colon C\in \cal M\}$ is the desired partition.
 
 It remains to prove the claim. We distinguish three cases: either $A\in \cal S^L$, or $A\in \cal S^R$, or $A\in \cal F\setminus \cal S$.
 
 Suppose first that $A\in \cal S^L$. Then $t'=\tau$. Let $\cal M_0$ be the set of all $C\in \cal P^R_{\tau}$ for which there exists $D\in \cal P^L_{\tau}$ such that the pair $C,D$ is impure. Since $|\cal P^L_\tau|=d+1$, we have $|\cal M_0|\leq d(d+1)$. By Lemma~\ref{lem:sandwich}, every $B\in \cal N'$ is a descendant of an element of $\cal M_0$. Hence, as $\cal M$ we can take the set of all descendants of the elements of $\cal M_0$ in $\cal P^R_{\tau+1}$. Since $|\cal M_0|\leq d^2+d$ and $\cal P^R_{\tau+1}$ differs from $\cal P^R_{\tau}$ by splitting exactly one set, we have $|\cal M|\leq d^2+d+1$.
 
 Suppose now that $A\in \cal S^R$. Again $t'=\tau$. Then we can simply take $\cal M$ to be the whole $\cal P^L_{\tau+1}$. Note that as above, $|\cal P^L_{\tau+1}|\leq |\cal P^L_\tau|+1=d+2$.
 
 Finally, suppose that $A\in \cal F\setminus \cal S$. Then $t'=t$. Let $\cal M_0$ be the set of all $C\in \cal P_t$ that together with $A$ form an impure pair. Then $|\cal M_0|\leq d$ and from Lemma~\ref{lem:sandwich} it follows that every $B\in \cal N'$ has an ancestor in $\cal M_0$. Hence, again we define $\cal M$ to be the set of all descendants of the elements of $\cal M_0$ in $\cal P_{t+1}$, and we have $|\cal M|\leq |\cal M_0|+1\leq d+1$.
\end{proof}

Let $r\coloneqq 2d^2+3d+1$. We define sets $\cal U_1,\ldots,\cal U_\ell$ for $\ell\coloneqq \binom{p}{2}r$ as follows. 
For every pair of distinct colors $1\leq i<j\leq p$, let $\cal U^{i,j}\coloneqq \lambda^{-1}(\{i,j\})$. By Lemma~\ref{lem:star-coloring}, $H[\cal U^{i,j}]$ is a star forest, so let $\cal C^{i,j}$ be the set of centers of the stars in $H[\cal U^{i,j}]$ (in each star with two elements, we pick the $\preceq$-smaller one as the center). For each $C\in \cal C^{i,j}$ let us fix the partition $\cal N^C_1,\ldots,\cal N^C_r$ of $\{A\colon C\preceq A\textrm{ and }AC\in E(H)\}$ provided by Lemma~\ref{lem:refine} (where we add some empty parts if necessary). Finally, for each $h\in [r]$ define
$$\cal U^{i,j,h}\coloneqq \cal C^{i,j}\cup \bigcup_{C\in \cal C^{i,j}} \left(\cal N^C_h\cap \cal U^{i,j}\right).$$
Thus, we have obtained $\ell$ sets $\cal U^{i,j,h}$ as above, and we reindex them as $\cal U_1,\ldots,\cal U_\ell$ arbitrarily.

That condition \ref{c:star-simpler} is satisfied follows directly from the construction and from Lemma~\ref{lem:refine}. For condition \ref{c:impure}, observe that for each edge $AB\in E(H)$, say with $A\prec B$, we have that $A,B\in \cal U^{i,j}$ where $(i,j)=(\lambda^{-1}(A),\lambda^{-1}(B))$, and $A$ is the center of the star of $H[\cal U^{i,j}]$ that contains both $A$ and $B$. Then $A,B\in \cal U^{i,j,h}$ where $h$ is such that $B\in \cal N^A_h$. 

\medskip

We are left with the construction of a graph $G'$ such that $H=\quo{G'}{\cal F}$ and $G'$ is a $q$-flip of $G$, for some constant $q$ depending only on $d$ and $k$. In fact, we use
$$q\coloneqq 6k(d+1)+3+2^{d+1}.$$
Intuitively, the idea is to analyze the elements of $\cal F\setminus \cal S$ as ordered by $\preceq$ and partition them into a bounded number of blocks, each behaving in a somewhat homogeneous way.

A {\em{block}} is a nonempty subset of $\cal F\setminus \cal S$ that is convex in the ordering $\preceq$. Call a block $\cal I$
\begin{itemize}
 \item 
{\em{sign-homogeneous}} if either $\tp(A)=-$ for all $A\in \cal I$, or $\tp(A)=+$ for all $A\in \cal I$; and
 \item
{\em{side-homogeneous}} if either $\bigcup_{A \in \cal I} A\subseteq L$ or $\bigcup_{A \in \cal I} A\subseteq R$.
\end{itemize}
A block is {\em{homogeneous}} if it is sign-homogeneous or side-homogeneous. Note that a block $\cal I$ is not homogeneous if and only if it contains a {\em{diagonal}}: a pair of sets $A,B$ such $\tp(A)\neq \tp(B)$ and exactly one of $A$ and $B$ is contained in $L$.

The next observation is crucial: in $\cal F\setminus \cal S$ there is no long sequence of diagonals placed one after the other.

\begin{lemma}\label{lem:many-diagonals}
 Suppose 
 $$A_1\prec B_1\prec A_2\prec B_2\prec \cdots \prec A_s\prec B_s$$
 are elements of $\cal F\setminus \cal S$ such that for each $i\in [s]$, the pair $A_i,B_i$ is a diagonal. Then $s\leq k(4d+1)$. 
\end{lemma}
\begin{proof}
 To reach a contradiction, suppose that $s>k(4d+1)$.
 Call a pair of distinct indices $i,j\in [s]$ {\em{in conflict}} if any of the sets $\{A_i,B_i\}$ is adjacent to any of the sets $\{A_j,B_j\}$ in $H$. By Lemma~\ref{lem:degeneracy}, every index $j$ is in conflict with at most $4d$ indices $i<j$. As $s>k(4d+1)$, by applying a greedy right-to-left procedure we can select a set $J\subseteq [s]$ of cardinality $k+1$ such that the indices in $J$ are pairwise not in conflict. For each $i\in J$, pick arbitrary $a_i\in A_i$ and $b_i\in B_i$. Now, since $A_i,B_i$ is a diagonal for each $i\in J$, the vertices $\{a_i,b_i\colon i\in J\}$ form a quasi-ladder of length $k+1$ in $G$. This is a contradiction with the assumption that $G$ has index at most $k$.
\end{proof}

We may now derive the following corollary of Lemma~\ref{lem:many-diagonals}.

\begin{lemma}\label{lem:few-blocks}
 The set $\cal F\setminus \cal S$ can be partitioned into at most $2k(4d+1)+1$ homogeneous blocks.
\end{lemma}
\begin{proof}
 Let $\cal I_1,\cal I_2,\ldots,\cal I_m$ be a partition of $\cal F\setminus \cal S$ into homogeneous blocks that minimizes $m$, the total number of blocks. Note here that such a partition always exists, as a single element of $\cal F\setminus \cal S$ always forms a homogeneous block. We assume that $\cal I_1,\cal I_2,\ldots,\cal I_m$ are ordered naturally by $\preceq$, that is, if $i<j$ then $A\prec B$ for all $A\in \cal I_i$ and $B\in \cal I_j$.
 
 For contradiction suppose $m\geq 2k(4d+1)+2$.
 By minimality, for each $i\in [k(4d+1)+1]$ the block $\cal I_{2i-1}\cup \cal I_{2i}$ is not homogeneous, hence it contains a diagonal, say $A_i,B_i$ where $A_i\prec B_i$. Now the existence of sets $A_1,B_1,A_2,B_2,\ldots,A_{k(4d+1)+1},B_{k(4d+1)+1}$ stands in contradiction with Lemma~\ref{lem:many-diagonals}.
\end{proof}

Let
$$\cal I_1,\cal I_2,\ldots,\cal I_m$$
be the partition of $\cal F\setminus \cal S$ into homogeneous blocks provided by Lemma~\ref{lem:few-blocks}, where $m\leq 2k(4d+1)+1$ and the blocks are ordered naturally by $\preceq$. We construct a graph $G'$ from $G$ by performing the following flips:
\begin{itemize}
 \item For each $j\in [m]$ such that $\cal I_j$ is side-homogeneous with $\bigcup \cal I_j\subseteq L$, flip the pair
 $$\bigcup\left(\cal S^R_\hd\cup \left(\cal F^R\cap \bigcup_{i<j} \cal I_i\right)\right),\quad\bigcup\{A\in \cal I_j~|~\tp(A)=+\}.$$
 \item For each $j\in [m]$ such that $\cal I_j$ is side-homogeneous with $\bigcup \cal I_j\subseteq R$, flip the pair
 $$\bigcup\left(\cal S^L\cup \left(\cal F^L\cap \bigcup_{i<j} \cal I_i\right)\right),\quad\bigcup\{A\in \cal I_j~|~\tp(A)=+\}.$$
 \item For each $j\in [m]$ such that $\cal I_j$ is not side-homogeneous, but is sign-homogeneous and $\tp(A)=+$ for all $A\in \cal I_j$, flip the pairs
 $$\bigcup\left(\cal S^R_\hd\cup \left(\cal F^R\cap \bigcup_{i<j} \cal I_i\right)\right),\quad\bigcup\{A\in \cal I_j~|~A\subseteq L\}$$
 and
 $$\bigcup\left(\cal S^L\cup \left(\cal F^L\cap \bigcup_{i<j} \cal I_i\right)\right),\quad\bigcup\{A\in \cal I_j~|~A\subseteq R\}$$
 and
 $$\bigcup\{A\in \cal I_j~|~A\subseteq L\},\quad\bigcup\{A\in \cal I_j~|~A\subseteq R\}.$$
 \item Let $\cal T$ be the partition of $\bigcup \cal S^R_\sm$ into at most $2^{d+1}$ sets of twins, provided by Lemma~\ref{lem:simple-hard}. Then for each $C\in \cal T$, flip the pair
 $$C,N_C$$
 where $N_C\subseteq L$ is the common neighborhood of all the twins in $C$. 
\end{itemize}
Observe that thus, we performed at most $3m+2^{d+1}\leq q$ flips. So $G'$ is a $q$-flip of $G$. It now follows directly from the construction that $H=\quo{G'}{\cal F}$.
This concludes the proof of Lemma~\ref{lem:main}.

\section{Proof of Theorem~\ref{thm:main}}\label{sec:main-theorem}

With Lemma~\ref{lem:main} in place, we may proceed to the proof of the main result, Theorem~\ref{thm:main}.
The argument will follow easily from the following three lemmas.

The first lemma is an easy translation of a graph to a bipartite graph.

\begin{lemma}\label{lem:bipartite-reduction}
    There is a pair of transductions $T,T'$ such that 
    for every graph $G$, the structure
     $T(G)$ is a bipartite graph and  $G\in T'(T(G))$.
\end{lemma}

For $t\in\N$, \emph{$t$-equivalence structure} is a set equipped with $t$ equivalence relations $\sim_1,\ldots,\sim_t$. Our main lemma, Lemma~\ref{lem:main}, can be applied recursively  to obtain the following lemma, which is the heart of the proof. Intuitively, it says that the edge relation of a bipartite graph $G$ of bounded twin-width and bounded quasi-ladder index
can be encoded using a bounded number of equivalence relations which can be defined in $G$.

\begin{lemma}\label{lem:main-transduction}
    Let $\DD_\le$ be a class of ordered bipartite  graphs of bounded convex twin-width and bounded quasi-ladder index.
    Then there is  $t\in\N$ and a pair of domain-preserving transductions $I,I'$
      such that     
    for  every $G_\le\in\DD_\le$, 
the set $I(G_\le)$ consists only of $t$-equivalence structures and
 $G\in I'(I(G_\le))$, where $G$ is the bipartite graph  $G_\le$ without the order.
\end{lemma}

Finally, the third lemma provides an encoding of a $t$-equivalence structure in a sparse graph.

\begin{lemma}\label{lem:equivalences}
    Fix $t\in\N$. There is a pair of transductions $K$ and $K'$ such that for every $t$-equivalence structure $\str S$,
     all structures in 
     $K(\str S)$ are $K_{t+1,t+1}$-free graphs, and
     $\str S\in K'(K(\str S))$.
\end{lemma}

The first and third lemma above are rather straightforward,
while the second one follows from our main Lemma~\ref{lem:main}. Before presenting their proofs, we show how Theorem~\ref{thm:main} follows from the above lemmas.

    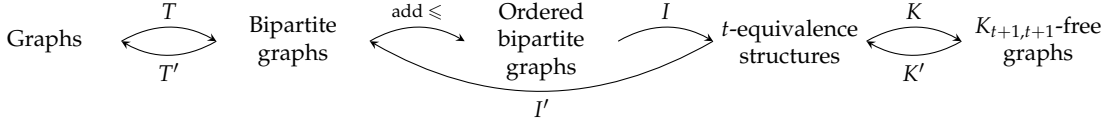
\begin{figure}
    \begin{tikzpicture}[%
        >=stealth,
        node distance=3.8cm,
        on grid,
        auto,
        scale=1, every node/.style={scale=0.86}
      ]
        \node (A) [text width=2.1cm, align=center] {Graphs};
        \node (B) [right of=A,text width=2.1cm,align=center]{Bipartite graphs};
        \node (C) [right of=B,text width=2.1cm,align=center]{Ordered bipartite graphs};
        \node (D) [right of=C,text width=2.1cm,align=center]{$t$-equivalence structures};
        \node (E) [right of=D,text width=2.1cm,align=center]{$K_{t+1,t+1}$-free graphs};
        \path[->] (A.east) edge[bend left] node {$T$}  (B.west);
        \path[<-] (A.east) edge[bend right] node[below] {$T'$} (B.west);
        \path[->] (B.east) edge[bend left] node {\footnotesize{add $\le$}} (C.west);
        \path[->] (C.east) edge[bend left] node {$I$} (D.west);
        \path[<-] (B.east) edge[bend right] node[below] {$I'$} (D.west);
        \path[->] (D.east) edge[bend left] node {$K$} (E.west);
        \path[<-] (D.east) edge[bend right] node[below] {$K'$} (E.west);
      \end{tikzpicture}
      \caption{Proof of Theorem~\ref{thm:main}. 
      The right arrows are transductions or operations on structures which preserve bounded twin-width. The left arrows are transductions.
     The composition of the right arrows, restricted to a stable class $\CC$
     of bounded twin-width, yields a class of graphs $\EE$ with bounded twin-width which excludes $K_{t+1,t+1}$ as a subgraph, so $\EE$ is a class of bounded sparse twin-width. The composition of the left arrows is a transduction which applied to $\EE$ yields a class containing $\CC$.
      }
      \label{fig:main proof}
    \end{figure}

\begin{proof}[Proof of Theorem~\ref{thm:main}]\label{proof:main}
    Let $\CC$ be a class of graphs which is monadically stable and has bounded twin-width.
    The proof is illustrated in Fig.~\ref{fig:main proof}.

Let $T$ and $T'$ be the transductions provided by Lemma~\ref{lem:bipartite-reduction} and
define $\DD\coloneqq T(\CC)$.
Then $\DD$ is a class of bipartite graphs which 
has bounded twin-width (as a binary structure) by Theorem~\ref{thm:interp_bd_tww}, and
has bounded ladder index by Lemma~\ref{lem:mon-stable-transduction}. In particular, $\DD$ has bounded quasi-ladder index, by Lemma~\ref{lem:index}.

Using Lemma~\ref{lem:conv_tww}, equip every bipartite graph $G\in \DD$ with a total order 
$\le$ yielding an ordered bipartite graph $G_\le$
of convex twin-width bounded by some constant.
Let $\DD_\le=\setof{G_\le}{G\in \DD}$.
Then $\DD_\le$ is a class of ordered bipartite graphs 
of bounded convex twin-width
and of bounded quasi-ladder index.

Apply Lemma~\ref{lem:main-transduction} to obtain transductions $I,I'$ and a number $t\in\N$.
Then $I(\DD_\le)$ is a class of $t$-equivalence structures.
Let $K$ and $K'$ be the transductions provided by Lemma~\ref{lem:equivalences} for the number $t$. Then the class $\EE=K(I(\DD_\le))$ 
is a class of $K_{t+1,t+1}$-free graphs which has bounded twin-width, hence
it has bounded sparse twin-width. 
Moreover, by the properties of $T',I'$ and $K'$, we have $\CC\subset T'(I'(K'(\EE)))$.
In particular, $\CC$ can be transduced from the class $\EE$ that has  bounded sparse twin-width.

By Theorem~\ref{thm:sparse_tww_be}, $\EE$ has bounded expansion. Since, $\CC$ can be transduced from $\EE$, it follows that $\CC$ has structurally bounded expansion. By Theorem~\ref{thm:SBE}, $\CC$ is transduction equivalent to some class of bounded expansion $\CC'$. It remains to note that $\CC'$ is $K_{s,s}$-free for some $s\in \N$ (due to being of bounded expansion) and has bounded twin-width (due to being transduction equivalent to a class of bounded twin-width). Hence $\CC'$ has bounded sparse twin-width.
\end{proof}
%
%

We now prove Lemmas~\ref{lem:bipartite-reduction},
~\ref{lem:main-transduction} and~\ref{lem:equivalences}.
We start with Lemma~\ref{lem:bipartite-reduction}.
\begin{proof}[Proof of Lemma~\ref{lem:bipartite-reduction}]
    Given a graph $G=(V,E)$, construct a bipartite graph 
    $\wt G$ as follows. The vertex set of $\wt G$ is $V\times \{0,1,2,3\}$, while the edge set comprises of the following edges:
    \begin{itemize}
     \item edges $(v,0)(w,3)$ and $(w,0)(v,3)$ for each $vw\in E$; and
     \item edges $(v,0)(v,1)$, $(v,1)(v,2)$, and $(v,2)(v,3)$ for each $v\in V$.
    \end{itemize}
    The left side of $\wt G$ is $\{(v,0),(v,2)\colon v\in V\}$ and the right side is $\{(v,1),(v,3)\colon v\in V\}$. 
    
    It is straightforward to see that there is a transduction $T$, independent of $G$, such that for each graph $G$, $T(G)$ consists of bipartite graphs only and $\wt G\in T(G)$. Furthermore, there is a transduction $T'$, independent of $G$ and $\wt G$, such that $G\in T'(\wt G)$. 
%
%
    %
    The transduction $T'$
    introduces two unary predicates $S_0,S_3$ to mark the vertices of the form $V\times\set0$ and $V\times\set3$.
     The result of $T'$ consists only of the vertices satisfying $S_0$, and the edge relation of $G$ is recovered using the edges in $\wt G$ as follows. 
     Observe that there is a formula $\phi(x,y)$ which holds of two vertices $a,b$ in $\wt G$ if and only if $a=(v,0)$ and $b=(v,3)$, for some $v\in V(G)$. 
    The formula $\phi(x,y)$ expresses that $a\in S_0$, $b\in S_3$, and $a$ and $b$ are connected by a path of length $3$ whose internal vertices satisfy neither $S_0$ nor $S_3$. Now the formula $\psi(x,y)\coloneqq S_0(x)\land S_0(y)\land \exists y'.(\phi(y,y')\land E(x,y'))$  holds of two vertices $a,b$ in $\wt G$ if and only if $a=(v,0)$ and $b=(w,0)$ for some edge $vw$ of $G$.
\end{proof}

We now prove Lemma~\ref{lem:equivalences}.
\begin{proof}[Proof of Lemma~\ref{lem:equivalences}]Let $\str S$ be a $t$-equivalence structure. For each $i\in\set{1,\ldots,t}$, let $W_i$ be the set of equivalence classes of $\sim_i$. 
    Let $W_0$ be the domain of $\str S$.
    Consider a graph $G=(V,E)$ constructed as follows: the vertex set $V$ is the disjoint union of the sets $W_0,W_1,\ldots,W_t$, while the edge set $E$ consists of edges $vw$
    such that $v\in W_0$, $w\in W_i$ for some $1\le i\le t$, and $v$ belongs to the $\sim_i$-equivalence class $w$.
By construction, the induced subgraphs $G[W_0]$ and $G[W_1\cup \ldots \cup W_t]$ are edgeless and each $v \in W_0$ has degree equal to $t$. It follows that $G$ does not contain $K_{t+1,t+1}$ as a subgraph.

We now argue that there are transductions $K$ and $K'$, independent of $\str S$, such that $K(\str S)$ contains only $K_{t+1,t+1}$-free subgraphs, and among them
a graph isomorphic to $G$,
while the transduction $K'$ is such that $\str S\in K'(G)$.
The transduction $K$ produces a copy of the input structure for each $i\in \{0,\ldots,t\}$; marks a set of representatives of equivalence classes of $\sim_i$ in the $i$th copy, for $i\in \{1,\ldots,t\}$; creates edges between each representative and all the $\sim_i$-equivalent elements in the original ($0$th) copy; and disposes of all non-representative vertices within copies $1,\ldots,t$. (For formal reasons, $K$ should also verify that the produced graph is $K_{t+1,t+1}$-free, otherwise it outputs an edgeless graph instead.) The transduction $K'$ introduces unary predicates $U_0,U_1,\ldots,U_t$, with the intention that $U_i$ marks the set $W_i$ as described above,
and introduces for each $i=1,\ldots,t$ an equivalence relation $\sim_i$ which holds of two elements of $U_0$ if and only if they have a common neighbor in $U_i$.
\end{proof}

It remains to prove Lemma~\ref{lem:main-transduction}.

\begin{proof}[Proof of Lemma~\ref{lem:main-transduction}]
Throughout the entire proof fix $d\in\N$.
For a given $k\in\N$ denote by  $\DD_k$ the class 
of all ordered bipartite graphs of convex twin-width at most $d$ and of index at most $k$.
We prove the statement of Lemma~\ref{lem:main-transduction} 
for the class $\DD_k$, by induction on $k$.
Precisely, we prove that there is  a pair of domain-preserving transductions $I_k$ and $I_k'$ and a number $t_k\in\N$ such that for every (ordered) $G\in\DD_k$, $I_k(G)$ consists of $t_k$-equivalence structures and $I_k'(I_k(G))$ contains $G$ without the order.

For $k=1$, note that a bipartite graph has index $1$ if and only if it is either complete or edgeless. Therefore,
the interpretation $I_1$ may simply output the domain of the 
given input structure, while the transduction $I_1'$
introduces unary predicates $L$ and $R$ marking the left and right parts, and then nondeterministically either introduces all or none of the edges between $L$ and $R$.

We proceed to the inductive step. Suppose $k\ge 1$ and the statement holds for $\DD_k$.
Let $t_k$ be the number obtained from the inductive assumption, and 
let $\ell,q$ be the numbers provided by Lemma~\ref{lem:main}. Set $t_{k+1}\coloneqq \ell\cdot (t_k+1)$.

Fix an ordered bipartite graph $G\in \DD_{k+1}$ with vertex set $V$.  We will construct transductions 
$J_{k+1}$ and $J_{k+1}'$ such that $J_{k+1}(G)$ consists of $t_{k+1}$-equivalence structures
and $J_{k+1}'(J_{k+1}(G))$ contains $G$ without the order. We will carry out the proof for a fixed $G$, but it will be clear that the transductions constructed throughout the proof do not depend on $G$.

Apply Lemma~\ref{lem:main} to $G$ to obtain a division $\cal F$ of $G$, 
    sets $\cal U_1,\ldots,\cal U_\ell\subset \cal F$, and a $q$-flip $G'$ of $G$ such that the following holds for $H\coloneqq \quo{G'}{\cal F}$:
    \begin{enumerate}[label=(\arabic*),ref=(\arabic*),leftmargin=*]          
        \item\label{cc:impure} For every edge $AB$ of $H$ there exists $i\in \set{1,\ldots,\ell}$ such that $A,B\in \cal U_i$.
        \item\label{cc:star-simpler} Each set $\cal U_i$, $i\in \set{1,\ldots,\ell}$, induces in $H$ a star forest. Moreover, for each star in this star forest, say with center $K_0$ and leaves $K_1,\ldots,K_m$, we have $G[K_0,K_1\cup\cdots\cup K_m]\in\DD_k$.
        \end{enumerate}

Let $\sim$ be the equivalence relation on $V$
signifying membership to the same part of $\cal F$, let $E'$ be the edge set of $G'$,
and let $E^H$ be the binary relation which holds of two vertices $v,w\in V$  if and only if $[v]_{\sim}$ and $[w]_{\sim}$ are adjacent in $H$.
\begin{claim}\label{cl:E^H}
    There is a domain-preserving transduction $P$ 
such that $(V,\sim,E',E^H)\in P(G)$.
\end{claim}
\begin{proof}
As $\cal F$ is a convex partition of $V$
it can be represented by a unary predicate $M_{\cal F}$
such that $M_{\cal F}=\setof{\min C}{C\in\cal F}$.
Then $\sim$
is definable by a formula using the predicate $M$,
expressing that for $u\le v$ in $V$,
$u\sim v$ if and only if the interval $(u,v]$ contains no element of $M_\cal F$. Hence, there is a transduction $T_{\sim}$ such that $T_{\sim}(G)$ contains $(V,{\sim})$.

Let $F_q$ be the transduction which performs an arbitrary $q$-flip of a given bipartite graph. Such a transduction is a composition of $q$ transductions performing single flips.
Since $G'$ is a $q$-flip of~$G$, we have that $G'\in F_q(G)$. 

To define the relation $E^H$,
observe that $E^H(u,v)$ if and only if there exist vertices $u',v'$ such that $u\sim u'$ and $v\sim v'$ and $uv$ is an edge in $E'$. Hence, $E^H$ can be constructed by a transduction,
which performs $F_q$ to define $E'$ and performs $T_{\sim}$ to define $\sim$,
and finally defines $E^H$.
\cqed\end{proof}

For $i=1,\ldots,\ell$, let $\approx_i$ be the equivalence relation 
which holds of two distinct vertices $v,w\in V$
if and only if $[v]_{\sim}$ and $[w]_{\sim}$ 
belong to the same star in the star forest induced by $\cal U_i$. Note that $\approx_i$ can be defined by a first-order formula using the relation $E^H$, since the components of a star forest have bounded radius.
Let $E_i'\subset E'$ be the set of edges of $G'$ whose endpoints are $\approx_i$-equivalent.

Properties~\ref{cc:impure} and~\ref{cc:star-simpler} above translate to the following properties:
\begin{enumerate}[label=(\arabic*'),ref=(\arabic*'),leftmargin=*]          
    \item\label{cc:impure'}         
        $E_1'\cup\cdots\cup E_\ell'=E'$, and
    
    \item\label{cc:star-simpler'} For every $i=1,\ldots,\ell$
    and $\approx_i$-equivalence class $C\subset V$,
    if $G_C'$ is the ordered bipartite subgraph of  $(V,L,R,\le,E_i')$ induced by $C$, then $G_C'$ has a $q$-flip $G_C$ that belongs to $\DD_k$.
    \end{enumerate}

By $\wh G$ denote the structure with 
domain $V$, order $\le$, and binary relations $E_1',\ldots,E_\ell'$ and  $\approx_1,\ldots,\approx_\ell$ defined above.
As the relations $E_1',\ldots,E_\ell',\approx_1,\ldots,\approx_\ell$ can be defined by first-order formulas using the relations $E^H, E',\sim$, as well as the unary relations  $\cal U_i'=\setof{v\in V}{[v]_{\sim}\in \cal U_i}$, for $i=1,\ldots,\ell$, from  Claim~\ref{cl:E^H} we get:
        \begin{claim}\label{cl:whG}
        There is a domain-preserving transduction $Q$ 
such that $\wh G\in Q(G)$.
\end{claim}

We now show, using the inductive assumption for $\DD_k$, that each of the relations $E_i'$ can be represented by equivalence relations.
\begin{claim}\label{cl:J_i}
    Fix $i\in\set{1,\ldots,\ell}$.
    There is a pair of domain-preserving transductions $J_i,J_i'$ such that each structure in $J_i(G)$ is a $(t_k+1)$-equivalence structure and $(V,E_i')\in J_i'(J_i(G))$. 
\end{claim}
\begin{proof}
    Let $I_k,I_k'$ be the transductions obtained by inductive assumption. 
    Let $G$ and $\wh G$ be as described above and fix a $\approx_i$-equivalence class $C\subset V$.
    
    Let $G_C'$ be the ordered bipartite subgraph  of $(V,L,R,\le,E_i')$ induced by $C$. 
By~\ref{cc:star-simpler'},
$G_C'$ has a $q$-flip $G_C$ that belongs to $\DD_k$. Hence, $I_k(G_C)$ contains a $t_{k}$-equivalence structure $\str S^C_i$ with vertices $C$ such that  $G_C\in I_k'(\str S^C_i)$.
Then $G_C'\in F_q(I_k'(\str S^C_i))$, where $F_q$ the transduction performing $q$ flips, as in the proof of Claim~\ref{cl:E^H}.

    Let $\str S_i$ be 
    the disjoint union of the structures $\str S^C_i$, over all $\approx_i$-equivalence classes $C$,
    additionally equipped with the equivalence relation $\approx_i$. Then $\str S_i$ is
    a $(t_k+1)$-equivalence structure with vertices $V$, equipped with equivalence relations $\approx_i$ and $\sim_1,\ldots,\sim_{t_k}$,
    such that the substructure of $\str S_i$ induced by each $\approx_i$-equivalence class $C$ is the structure $\str S^C_i$, extended with the total relation $\approx_i$ on its domain $C$.

    We now observe that there is a transduction $J_i$ such that $\str S_i\in J_i(G)$. The transduction $J_i$ is the composition of the transduction $Q$ from Claim~\ref{cl:whG}, followed by a parallel application (see Lemma~\ref{lem:parallel-application}) of the transduction $I_k\circ F_q$, applied to each substructure induced by some $\approx_i$-equivalence class of the input structure.
    
    Moreover,
    there is a transduction $J_i'$ such that $J_i'(J_i(G))=(V,E_i')$. The transduction $J_i'$ is the parallel application of the transduction $F_q\circ I_k'$, applied to each substructure induced by some $\approx_i$-equivalence class of the input structure.
\cqed\end{proof}

Now, we can combine the transductions provided by Claim~\ref{cl:J_i}.

\begin{claim}There is a pair of domain-preserving transductions $J$ and $J'$ such that each structure in $J(G)$ is an $\ell\cdot (t_k+1)$-equivalence structure and $(V,E')\in J'(J(G))$. 
\end{claim}
\begin{proof}
    The transduction $J$ applies each of the transductions $J_1,\ldots,J_\ell$ from Claim~\ref{cl:J_i} to $G$, obtaining 
    jointly $\ell\cdot (t_k+1)$ equivalence relations 
    $\sim^i_j$, for $i\in \set{1,\ldots,\ell}$ and $j\in\set{1,\ldots,t_k+1}$.
    The structure resulting structure is the set $V$ equipped with all those equivalence relations $\sim^i_j$. Such a combination of transductions is a transduction, see  Lemma~\ref{lem:super-impose}.

    The transduction $J'$ applies each of the transductions $J_1',\ldots,J_\ell'$, where $J_i'$ is applied to the $(t_k+1)$-equivalence structure with equivalence relations $\sim^i_1,\ldots,\sim^i_{t_k+1}$
    and outputs the relation~$E_i'$.
    The result of $J'$ is the structure consisting of $V$ and the relation $E'=E_1'\cup\cdots\cup E_\ell'$.    
\cqed\end{proof}

As noted before, we set $t_{k+1}\coloneqq\ell\cdot(t_k+1)$. Further, we define the transduction $J$ to be $J_{k+1}$, and the transduction $J'$ to be $J_{k+1}'$ followed by the transduction $F_q$ performing $q$ flip operations, and followed by an introduction of two unary predicates $L$ and $R$.
Then $J_{k+1}(G)$ is a $t_{k+1}$-equivalence structure and $(V,L,R,E)\in J_{k+1}'(J_{k+1}(G))$, since $G$ is a $q$-flip of $G'$. 

As the constructed transductions $J_{k+1}$ and $J_{k+1}'$ do not depend on  $G\in\DD_{k+1}$,
this finishes the inductive step,
and the proof of Lemma~\ref{lem:main-transduction}.
\end{proof}

\section{Linear $\chi$-boundedness}\label{sec:lin-chi}

In this section we discuss the implications of our results for $\chi$-boundedness of stable classes of bounded twin-width. Let us first recall the definitions.

Let $G$ be a graph. The {\em{chromatic number}} of $G$, denoted $\chi(G)$, is the least number of colors needed for a {\em{proper coloring}} of $G$: a coloring of vertices of $G$ where no two adjacent vertices receive the same color. The {\em{clique number}} of $G$, denoted $\omega(G)$, is the largest size of a clique in~$G$. Clearly, $\chi(G)\geq \omega(G)$ for every graph $G$. A graph class $\CC$ is {\em{$\chi$-bounded}} if a converse inequality holds in the following sense: there exists a function $f\colon \N\to \N$ such that
$$\chi(G)\leq f(\omega(G))\qquad\textrm{for each }G\in \CC.$$
If $f$ can be chosen to be a polynomial or a linear function, then we respectively say that $\CC$ is {\em{polynomially}} or {\em{linearly $\chi$-bounded}}.

In~\cite{bonnet2021tww2} it was proved that every graph class of bounded twin-width is $\chi$-bounded.
From our main result, Theorem~\ref{thm:main}, it follows that every stable class of bounded twin-width has structurally bounded expansion, as it can be transduced from a class of bounded sparse twin-width, and such classes have bounded expansion~\cite{bonnet2021tww2}. On the other hand, from the results of~\cite{gajarsky2020sbe} it easily follows that classes with structurally bounded expansion are linearly $\chi$-bounded (see the discussion in~\cite{nesetril2021linrw_stable}). Hence, the same can be also concluded about every stable class of bounded twin-width. 

In this section we present a direct proof of this fact, which avoids the need of using the results of~\cite{gajarsky2020sbe} and is based on a small subset of the reasoning presented in Section~\ref{sec:main-lemma}. The additional benefit is that we also obtain precise bounds on $\chi(G)$ in terms of $\omega(G)$.

For the purpose of this section, we will use the notion of (quasi-ladder) index in general graphs, which we recall for convenience: the index of a graph $G$ is the index of the bipartite graph whose sides are two copies of $V(G)$, where a vertex $u$ from the first copy is adjacent to a vertex $v$ from the second copy if and only if $u$ and $v$ are adjacent in $G$. This is easily equivalent to bounding the order of quasi-ladders in $G$, where the elements of a quasi-ladder are not bound to respective sides of a bipartite graph, but can be chosen freely among all the vertices. Clearly, every monadically stable class of graphs has a bounded index in this~sense.

The following notion will be useful. A {\em{cograph}} is a graph that does not contain $P_4$ --- the path on $4$ vertices --- as an induced subgraph. It is well-known that cographs admit the following recursive characterization:
\begin{itemize}
 \item A one-vertex graph is a cograph.
 \item If $G_1,\ldots,G_k$ are cographs, then their disjoint union is also a cograph.
 \item If $G_1,\ldots,G_k$ are cographs, then their {\em{join}} is also a cograph, where the join is obtained from the disjoint union by making every pair of vertices $u\in V(G_i)$ and $v\in V(G_j)$ adjacent whenever $i\neq j$.
\end{itemize}
The class of cographs is very well-understood. In particular, cographs are {\em{perfect}}: $\chi(G)=\omega(G)$ whenever $G$ is a cograph. This fact can be combined with the following result, which we will obtain using a variation on the reasoning from Section~\ref{sec:main-lemma}.

\begin{theorem}\label{thm:cograph-coloring}
 Let $G$ be a graph of twin-width at most $d$ and index at most $k$. Then there is a vertex coloring of $G$ using at most $(2d+4)^{k-1}$ colors such that every color class induces a cograph in $G$.
\end{theorem}

By combining the perfectness of cographs with Theorem~\ref{thm:cograph-coloring} we can immediately derive the following.

\begin{corollary}
 Let $G$ be a graph of twin-width at most $d$ and index at most $k$. Then $$\chi(G)\leq (2d+4)^{k-1}\cdot \omega(G).$$
\end{corollary}

We are left with proving Theorem~\ref{thm:cograph-coloring}. The main tool will be the following variation on Lemma~\ref{lem:main}, which uses only a subset of arguments, but works on not necessarily bipartite graphs.

\begin{lemma}\label{lem:cograph-coloring}
 Let $G$ be a graph of twin-width at most $d$ and index equal to $k$, where $k\geq 2$. Then there is a partition $\cal F$ of vertices of $G$ and a coloring $f\colon \cal F\to [2d+4]$ satisfying the following properties:
 \begin{itemize}
  \item Each $A\in \cal F$ induces a subgraph $G[A]$ of index smaller than $k$.
  \item For every $i\in [2d+4]$, one of the following holds: each pair of distinct $A,B\in f^{-1}(i)$ is complete, or each pair of distinct $A,B\in f^{-1}(i)$ is anti-complete. 
 \end{itemize}
\end{lemma}
\begin{proof}
 Let $\cal P_1,\ldots,\cal P_n$ be an uncontraction sequence of $G$ of width at most $d$, where $n\coloneqq |V(G)|$.
 As in Section~\ref{sec:main-lemma}, we say that for $s,t\in [n]$ with $s\leq t$, a part $A\in \cal P_s$ is an {\em{ancestor}} of a part $B\in \cal P_t$ if $A\supseteq B$. Then also $B$ is a {\em{descendant}} of $A$.
 
 We use the following freezing mechanism, which differs from the one used in the proof of Lemma~\ref{lem:main}, but is based on a similar principle.
 For $t\in [n]$, a part $A\in \cal P_t$ is {\em{frozen}} at time $t$ if:
 \begin{itemize}
  \item no ancestor of $A$ was frozen at any time $s<t$, and
  \item the induced subgraph $G[A]$ has index smaller than $k$.
 \end{itemize}
 Let $\cal F_t$ be the set of parts of $\cal P_t$ frozen at time $t$. Note the following.
 
 \begin{claim}\label{cl:freezing-really-tiny}
  For each $t\in [n]$, $|\cal F_t|\leq 2$.
 \end{claim}
 \begin{proof}
 Note that $\cal P_t$ differs from $\cal P_{t-1}$ by that a part of $\cal P_{t-1}$ is replaced by two its subsets. Then only those two subsets can belong to $\cal F_t$. 
 \cqed\end{proof}

 Let $\cal F\coloneqq \cal F_1\cup \cal \cdots \cup \cal F_n$ be the set comprising all parts frozen at any moment. The same reasoning as in the proof of Lemma~\ref{lem:Fdivision} yields the following.
 
 \begin{claim}
  $\cal F$ is a partition of $V(G)$.
 \end{claim}

 That $G[A]$ has index smaller than $k$ for each $A\in \cal F$ follows directly from the construction. We are left with constructing a suitable coloring $f$.
 
 The following claim is a simple analogue of Lemma~\ref{lem:sandwich} that will be sufficient for our needs.
 
 \begin{claim}\label{cl:light-sandwich}
  Let $t\in [n]$ and let $B\in \cal F_t$. Then there is a set $\cal N\subseteq \bigcup_{s<t} \cal F_s$ such that $|\cal N|\leq d$ and denoting
  $$W\coloneqq \bigcup\left(\bigcup_{s<t} \cal F_s \setminus \cal N\right),$$
  the pair $B,W$ is pure.
 \end{claim}
 \begin{proof}
  We may assume that $t>1$, as the claim holds for $t=1$ trivially. Let $B'$ be the unique ancestor of $B$ in $\cal P_{t-1}$.
  Let $\cal M$ be the set of all those parts $D\in \cal P_{t-1}$ for which the pair $B,D$ is impure. Further, let $\cal N$ comprise all the sets in $\bigcup_{s<t} \cal F_s$ that have a descendant in $\cal M$. By the assumption on the width of the uncontraction sequence, we have $|\cal M|\leq d$. Since every element of $\cal M$ has at most one frozen ancestor, it follows that $|\cal N|\leq d$. We are left with proving that the pair $B,W$ is pure.
  
  First, we observe that every vertex $u\in W$ is pure towards $B'$, that is, the pair $\{u\},B'$ is pure.
  Indeed, if this was not the case, then the part $U\in \cal P_{t-1}$ that contains $u$ would form an impure pair with $B'$, implying that $U\in \cal M$. Hence, the set $A\in \bigcup_{s<t} \cal F_s$ that contains $u$ would belong to $\cal N$, but $W$ is disjoint with $\bigcup \cal N$; a contradiction.
  
  Next, suppose for contradiction that the pair $B,W$ is impure. The observation of the previous paragraph implies that there must exists $u^-,u^+\in W$ such that $u^-$ is non-adjacent to all the vertices of $B'$ while $u^+$ is adjacent to all the vertices of $B'$. Since $B$ is frozen at time $t$, $B'$ was not frozen at time $t-1$, hence in $G[B']$ there exists a quasi-ladder of length $k$, say formed by sequences $x_1,\ldots,x_k$ and $y_1,\ldots,y_k$. Now $x_1,\ldots,x_k,u^-$ and $y_1,\ldots,y_k,u^+$ form a quasi-ladder of length $k+1$ in $G$, a contradiction. 
 \cqed\end{proof}
 
 For every $B\in \cal F$ define the type $\tp(B)\in \{+,-\}$ as the purity type of the pair $B,W$, where $W$ is defined as in Claim~\ref{cl:light-sandwich}. Note that it may happen that $W$ is empty; in this case $\tp(B)$ is chosen arbitrarily. Further, let $\preceq$ be any ordering of $\cal F$ that respects the freezing times: whenever $A\in \cal F_s$ and $B\in \cal F_t$ for $s<t$, then $A\prec B$. Note that for every $t\in [n]$, the (at most two) elements of $\cal F_t$ are ordered arbitrarily in $\preceq$. Finally, define a graph $H$ on vertex set $\cal F$ as follows: for distinct sets $A,B\in \cal F$, say $A\prec B$, make $A$ and $B$ adjacent in $H$ if and only if the pair $A,B$ mismatches the type $\tp(B)$.
 
 From Claim~\ref{cl:light-sandwich} we immediately obtain a bound on the degeneracy of $\preceq$.
 
 \begin{claim}\label{cl:light-degeneracy}
  For every $B\in \cal F$ there are at most $d+1$ sets $A\in \cal F$ such that $A\prec B$ and $A$ and $B$ are adjacent in $H$.
 \end{claim}
 \begin{proof}
  Let $t$ be such that $B \in \cal F_t$.
  Each set $A$ for which $A \prec B$ holds and which is adjacent to $B$ in $H$ belongs to either $\cal F_t$ or $\bigcup_{s<t} \cal F_s$. Then Claims~\ref{cl:freezing-really-tiny} and~\ref{cl:light-sandwich} respectively imply upper bounds of $1$ and of $d$ on the number of sets $A$ falling into these cases.
 \cqed\end{proof}
 
 By Claim~\ref{cl:light-degeneracy}, we may apply a greedy left-to-right procedure on the ordering $\preceq$ to find a coloring $g\colon \cal F\to [d+2]$ where each color class is an independent set in $H$. We may further refine $g$ to a coloring $f\colon \cal F\to [2d+4]$ by splitting every color class of $g$ into two according to the types, as follows:
 $$f(A)\coloneqq \begin{cases}2\cdot g(A)&\qquad \textrm{if }\tp(A)=-,\\2\cdot g(A)-1&\qquad \textrm{if }\tp(A)=+.\end{cases}$$
 That $f$ satisfies the required properties follows directly from the construction.
\end{proof}
%

We can now prove Theorem~\ref{thm:cograph-coloring} using Lemma~\ref{lem:cograph-coloring}.

\begin{proof}[Proof of Theorem~\ref{thm:cograph-coloring}]
 We apply induction on $k$.
 For the base case, every graph of index~$1$ is either complete or edgeless, and hence a cograph. Thus, for $k=1$ one color suffices.
 
 Assume that $k\geq 2$. Apply Lemma~\ref{lem:cograph-coloring}, yielding a suitable partition $\cal F$ of $V(G)$ and coloring $f\colon \cal F\to [2d+4]$. By induction, for each $A\in \cal F$ we find a coloring $h_A\colon A\to [(2d+4)^{k-2}]$ in which every color class induces a cograph in $G[A]$. Construct a coloring $h$ of $G$ by overlaying colorings $f$ and $\{h_A\colon A\in \cal F\}$ as follows: for each vertex $u$, say belonging to $A\in \cal F$,  set
 $$h(u)\coloneqq (f(u),h_A(u)).$$
 To see that each color class $(i,j)\in [2d+4]\times [(2d+4)^{k-2}]$ induces a cograph in $G$, observe that $G[h^{-1}(i,j)]$ is either the disjoint union or the join of graphs $\{G[h_A^{-1}(j)]\colon A\in f^{-1}(i)\}$, and these graphs are cographs by induction.
\end{proof}

\section{Discussion}
\label{sec:discussion}

\begin{figure}[h!]\centering
    \includegraphics[page=3,scale=0.60]{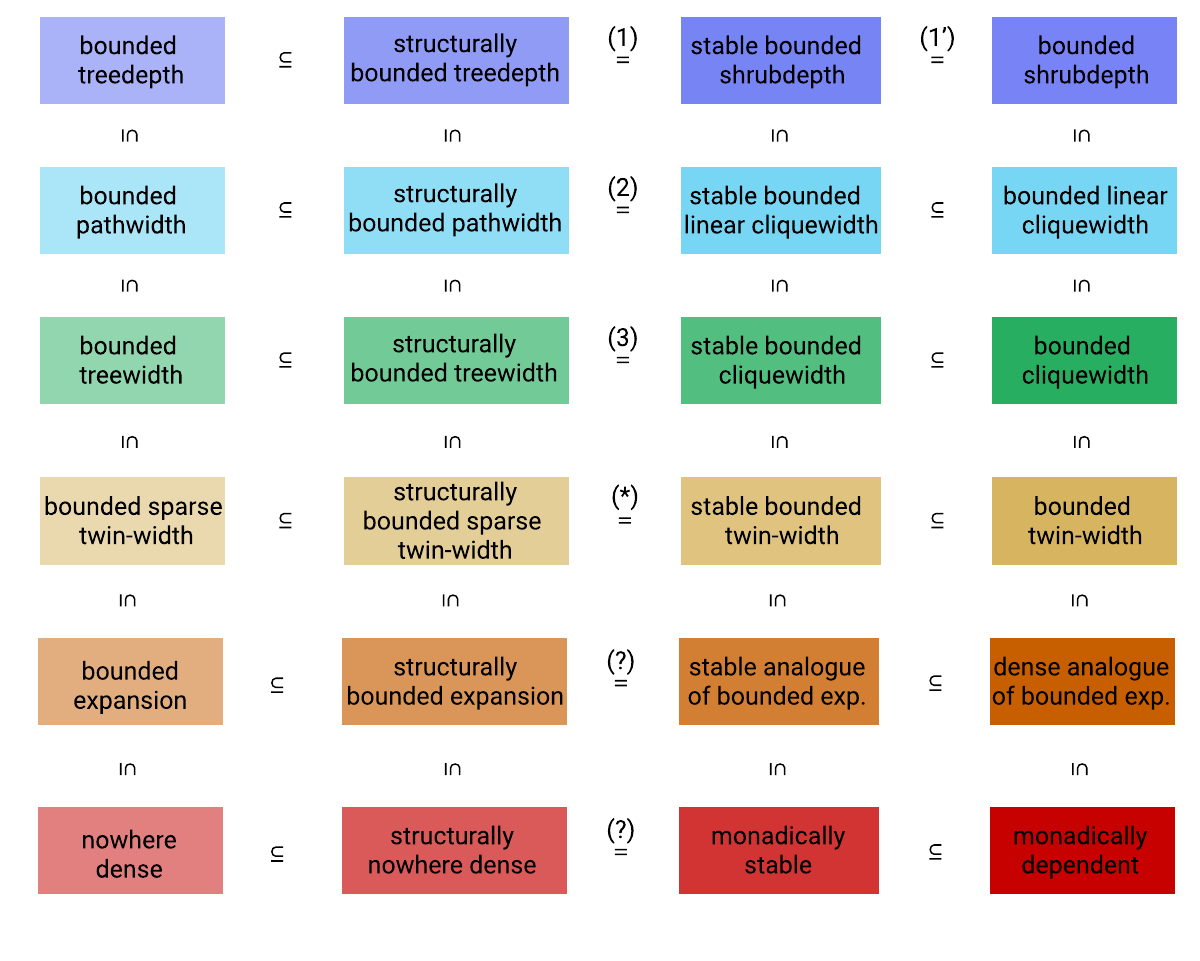}
    \caption{Some properties of graph classes that form transduction 
    ideals (in the second, third, fourth column), or weakly sparse transduction ideals (in the first column).
    For every row $(\cal P_1,\cal P_2,\cal P_3, \cal P_4)$ in the table,
    the property $\cal P_1$ consists of all classes in $\cal P_4$ that are weakly sparse;
    the property $\cal P_2$ is the property of being `structurally $\cal P_1$'; the property $\cal P_3$ consists of all classes in $\cal P_4$ which are stable.
    The inclusion $\cal P_2\subset \cal P_3$ holds in each row, and equality $\cal P_2=\cal P_3$ holds for the first four rows $(1)$, $(2)$, $(3)$, $(*)$,
    with $(*)$ being our main result, Theorem~\ref{thm:main}.
    The first equality $(?)$ is our
    Conjecture~\ref{dense-sparse be},
    and the second equality $(?)$ is 
    Conjecture~\ref{conj:sparsification}.
    All remaining inclusions in the figure are strict.
    }
    \label{fig:discussion}
\end{figure}

We now place our result in the broader context of monadically dependent graph classes. Our discussion is based on Figure~\ref{fig:discussion}, which is a version of Figure~\ref{fig:intro}, extended with the `bounded expansion' row. One of the goals of this discussion is to solve this crossword puzzle: 
 propose candidate notions of properties that can be put in that row, which we dub
`dense analogue of bounded expansion', and `stable analogue of bounded expansion'.
In fact, we will propose a generic construction of lifting properties of sparse graph classes to their dense analogues.  This will lead to several conjectures, that are related to some important known conjectures.

Let us start by analyzing the figure.
The second, third and fourth column in the figure consist of \emph{transduction ideals}, that is, properties of graph classes that are closed under taking transductions.
Properties $\cal P$ in the first column consists of \emph{sparse transduction ideals}:
if $\CC$ has property $\cal P$ and $\CC$  transduces a weakly sparse class $\DD$, then $\DD$ has property $\cal P$. 
In fact, the properties in the first column are precisely the weakly sparse parts of the corresponding properties in the last column (ignore the `bounded expansion' row for the moment).

Our main result says that 
stable classes of bounded twin-width are precisely transductions (even interpretations) of sparse classes of bounded twin-width.
This proves equality among properties in the second column and in the third column of the figure, within the first four rows. 
Conjecture~\ref{conj:sparsification} predicts that such an equality holds also for the last row. Thus we confirm this conjecture in the case of classes of bounded twin-width.

So the properties in the third column can be defined --- conjecturally at least --- in terms of properties in the first column, by taking the transduction closure. 
Properties in the first and third column can be defined in terms of properties in the last column, by restricting to weakly sparse/stable classes, respectively.
Can properties in the last column be defined in terms of properties in the first column? Specifically, is there a {\em{generic way}} of generalizing properties of sparse classes to properties of unstable classes such that boundedness of pathwidth/treewidth is mapped to boundedness of (linear) cliquewidth, boundedness of sparse twin-width is mapped to boundedness of twin-width, and nowhere denseness is mapped to monadic dependence?

A  recipe for defining the properties in the fourth column from the properties in the first one would in particular allow to answer the following question: what is the dense, unstable analogue of classes with bounded expansion? 
In other words, we are seeking
to define the properties in the  `bounded expansion' row. The dense analogue of classes with bounded expansion should have the following properties:
\begin{itemize}[nosep]
 \item it should form a transduction ideal (be closed under taking transductions);
 \item it should contain all classes with bounded expansion and all classes with bounded twin-width;
 \item it should consists only of monadically dependent classes; and
 \item its weakly sparse (resp. stable) classes should be exactly the classes with bounded expansion (resp. structurally bounded expansion).
\end{itemize}
There is an easy answer to this question: the property of having structurally bounded expansion or having bounded twin-width. 
This answer is not very illuminating, however. It is the smallest possible transduction ideal which satisfies the above requirements.
But what if we consider the \emph{largest} such transduction ideal instead?




\medskip
This motivates the following attempt at
defining the properties in the fourth column in terms of the properties in the first column.
For a sparse transduction ideal $\mathcal P$ define the \emph{dense analogue} of $\cal P$, denoted $\overline{\mathcal P}$, as the property consisting of all classes $\CC$ such that every weakly sparse class $\DD$ which can be transduced from $\CC$ belongs to  $\mathcal P$.
Note that if $\mathcal P$ is any transduction ideal then $\mathcal P\subset \overline{\mathcal P\cap\mathcal W}$, where $\mathcal W$ consists of all weakly sparse classes.
We conjecture that the transduction ideals in the last column in Figure \ref{fig:intro}
are dense analogues of the corresponding sparse transduction ideals in the first column.
More precisely, consider the following five statements:
\begin{enumerate}
    \item\label{st:1} the dense analogue of `bounded treedepth' is `bounded shrubdepth',
    \item\label{st:2} the dense analogue of `bounded pathwidth' is `bounded linear cliquewidth',
    \item\label{st:3} the dense analogue of `bounded treewidth' is `bounded cliquewidth',
    \item\label{st:4} the dense analogue of `bounded sparse twin-width' is `bounded twin-width', and
    \item\label{st:5} the dense analogue of `nowhere dense' is `monadically dependent'.
\end{enumerate}
As we argue below, statements \eqref{st:1} and \eqref{st:5} hold, and we conjecture that statements \eqref{st:2}, \eqref{st:3}, and \eqref{st:4} are true.

Note that each of those statements is a duality statement: for instance, the statement \eqref{st:3} says that every graph class~$\CC$ either has bounded cliquewidth, or transduces a weakly sparse class of unbounded treewidth (those are mutually exclusive).
In other words, the statements claim that $\cal P=\overline{\cal P\cap\cal W}$,
for each of the properties $\cal P$ among `bounded shrubdepth', `boudned linear cliquewidth',
`boudned cliquewidth', `bounded twin-width', and `monadically dependent'.
The  inclusion $\cal P\subset \overline{\cal P\cap\cal W}$ is clear, since $\cal P$ is a transduction ideal, so the relevant claims concern the other inclusion.
This is equivalent to stating that every class $\CC$ that does \emph{not} have the property $\cal P$, transduces some weakly sparse class~$\DD$ that also does not have the property $\cal P$.

The statement \eqref{st:5} holds: if a class $\CC$ is not monadically dependent then it transduces every class, in particular, it transduces the class of $1$-subdivisions of all graphs, which is weakly sparse but not monadically dependent. Hence, `monadically dependent' is the dense analogue of `nowhere dense'.
The statement \eqref{st:1} is equivalent to the following, recent result.

\begin{theorem}[\cite{SD-paths}]\label{thm:sd}If $\CC$ is a class which does not have bounded shrubdepth, then $\CC$ transduces the class of all paths.
\end{theorem}

We now show that Theorem~\ref{thm:sd} is equivalent to the statement~\eqref{st:1}.
As remarked,~\eqref{st:1} is equivalent to the statement (1') that every class of unbounded shrubdepth transduces some weakly sparse class of unbounded shrubdepth (equivalently, of unbounded treedepth).

Theorem~\ref{thm:sd} implies the statement~(1'), since the class of all paths is weakly sparse and has unbounded shrubdepth. In the other direction, suppose that from $\CC$ one can transduce a weakly sparse class $\CC'$ of unbounded shrubdepth. If $\CC'$ is monadically stable then, being weakly sparse, $\CC'$ is also nowhere dense.
By~\cite[Proposition 8.2]{sparsity}\footnote{More generally, this statement holds for weakly sparse classes $\CC'$, as proved by Atminas et al.~\cite[Theorem~3]{AtminasLR12}.},
if $\CC'$ contains arbitrarily long paths as subgraphs then it contains arbitrarily long paths as induced subgraphs, which yields the conclusion of Theorem~\ref{thm:sd}. Otherwise, if $\CC'$ does not contain arbitrarily long paths as subgraphs then $\CC'$ has bounded treedepth, in particular has bounded shrubdepth, a contradiction.
Finally, if $\CC'$ is not monadically stable then $\CC'$ transduces the class of all ladders, which in turn transduces the class of all paths. In any case, the class of all paths can be transduced from $\CC$.

Hence, Theorem~\ref{thm:sd} is equivalent to the statements \eqref{st:1} and (1').
In particular, this implies that bounded shrubdepth is the largest transduction ideal that can be put into the upper-right corner of Fig.~\ref{fig:intro}, which in conjunction with weak sparsity implies bounded treedepth.
Since classes of bounded shrubdepth are monadically stable, this explains why the two properties in that row of the figure are equal.

\medskip
The statements \eqref{st:2} and \eqref{st:3}
can be formulated more explicitly, as follows. 
A \emph{subdivision} of a graph $G$ is any graph obtained from $G$ by replacing each edge by some path of positive length.
Statement \eqref{st:2} is equivalent to the following, more precise conjecture.

\begin{conjecture}\label{conj:lcw}
    Let $\CC$ be a class with unbounded linear cliquewidth. Then there is a class $\DD$ which can be transduced from $\CC$ and which contains some subdivision of every binary tree.
\end{conjecture}
It is well-known that the class of all binary trees has unbounded pathwidth, and that subdivisions cannot decrease the pathwidth. Hence, if a class $\DD$  contains some subdivision of every binary tree, then $\DD$ has unbounded pathwidth. Therefore, Conjecture~\ref{conj:lcw} implies the statement \eqref{st:2}. 
The converse implication will be shown later below.

Similarly, the statement~\eqref{st:3} is equivalent to the following conjecture, which replaces trees with walls.
A \emph{wall} is a variation of a grid with maximum degree $3$, as depicted in Fig.~\ref{fig:wall}. 
    \begin{figure}\centering
        \includegraphics[page=3,scale=0.6]{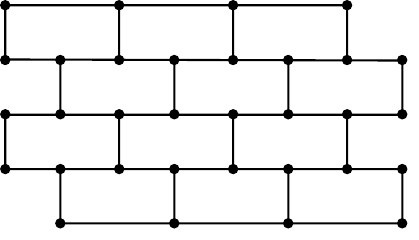}
    \caption{A wall graph.}
    \label{fig:wall}
    \end{figure}

\begin{conjecture}\label{conj:cw}
    Let $\CC$ be a class with unbounded  cliquewidth. Then there is a class $\DD$ which can be transduced from $\CC$ and which contains some subdivision of every wall.
\end{conjecture}
Note that Conjecture~\ref{conj:cw} implies the following well-known conjecture, which is often referred to as Seese's conjecture
(see~\cite{dawar2020mso}). Seese's original conjecture \cite[Problem 1]{SEESE1991169}, regarding 
graph classes with undecidable MSO theory, is an immediate consequence of it.

\begin{conjecture}[Variant of Seese's conjecture]\label{conj:seese}
    Let $\CC$ be a class with unbounded  cliquewidth. Then the class of all grids can be transduced from $\CC$ using an MSO transduction.
\end{conjecture}

Monadic second-order logic (MSO), is a powerful extension of first-order logic, which in particular allows to express the property $\phi(u,v)$ that two vertices $u$ and $v$ are connected by a path consisting of vertices of degree $2$.  Hence Conjecture~\ref{conj:cw} immediately implies Conjecture~\ref{conj:seese}, since if $\DD$ contains some subdivision of every wall then $\DD$ interprets the class of all walls  using an MSO interpretation, namely one whose 
domain formula restricts to vertices of degree~$3$, while edge-defining formula is the formula $\phi(u,v)$ above.
Finally, the class of grids can be obtained from the class of walls using an MSO transduction.

Note that a weaker version of Conjecture~\ref{conj:seese} holds, where MSO is replaced by the slightly more powerful C$_2$MSO logic~\cite{COURCELLE200791}, extending MSO by the capability of counting modulo $2$.

\medskip
A positive answer to Conjectures~\ref{conj:lcw} and~\ref{conj:cw} would resolve 
a question of Blumensath and Courcelle \cite[Open Problem 9.3]{lmcs:1208}, about the MSO-transduction hierarchy. That question asks 
whether, up to MSO-transduction equivalence, every class is equivalent to either a class of trees of depth $n$, for some $n\ge 0$, or to the class of paths, or the class of all trees, or to the class of all graphs. 
By~\cite[Theorem 4.9]{lmcs:5149}, every class of bounded shrubdepth is MSO-transduction equivalent to the class of trees of depth $n$, for some $n\ge 0$.
To answer the question of Blumensath and Courcelle,
it remains to show that:
\begin{itemize}
    \item 
     every class of unbounded shrubdepth MSO-transduces the class of all paths; this is now confirmed by Theorem~\ref{thm:sd},
    \item every class of unbounded linear cliquewidth MSO-transduces the class of all trees (this is implied by 
    Conjecture~\ref{conj:lcw}),
    \item 
    every class of unbounded cliquewidth MSO-transduces the class of all grids (this is the above variant of Seese's conjecture, and is implied by 
    Conjecture \ref{conj:cw}). 
\end{itemize}

\bigskip
We now partially confirm Conjectures \ref{conj:lcw} and~\ref{conj:cw} for a wide range of classes, including all classes with structurally bounded expansion.
We first recall the following notion.

 Say that a class $\DD$ has \emph{star chromatic number} at most $N$ if 
every $G\in \DD$ can be vertex-colored using at most $N$ colors such that every color is an independent set and every pair of colors induces a star forest in $G$. 
Note that every class with bounded expansion has bounded star chromatic number, by Lemma~\ref{lem:star-coloring}. We use the following observation,
due to Ossona de Mendez
(see e.g.~\cite[Lemma 34]{nesetril2020arboretum}, and~\cite[Lemma 2.1]{nesetril2021linrw_stable} for a related result).
\begin{lemma}\label{lem:transducing subgraphs}
Let $\DD$ be a class of graphs with
bounded star chromatic number. Then the subgraph closure of $\DD$ can be transduced from $\DD$. 
\end{lemma}
\begin{proof}[Proof (sketch)]
First consider the case when $\DD$ is the class of star forests. Let $G\in \DD$ and $H\subset G$ be its subgraph. 
Introduce two unary predicates $R$ and $U$, where $R$ consists of the centers of the stars in $G$ (in a two-vertex star we choose any vertex to be the center), and $U$ consists of those vertices $v\in V(H)$ 
such that either $v\in R$, or $v\notin R$ and $vv'\in E(H)$, where $v'$ is the unique neighbor of $v$ in $G$. 
The transduction $T$ first introduces the predicates $U$ and $R$, and then restricts the domain to $U$ and creates an edge between two vertices $u$ and $v$ if and only if $u\in R$, $v\in U$, and $uv\in E(G)$, or symmetrically with $v$ and $u$ replaced. Then $H\in T(G)$.

We now consider the general case. 
Given $G\in\DD$, the transduction first colors $V(G)$ using $N$ colors so that every two colors induce a star forest. Then, for each pair $\set{c,d}\subset [N]$ of colors (in arbitrary order), apply the transduction $T$ defined above
to the subgraph of $G$ induced by those two colors, obtaining a graph $H_{cd}$.
The transduction outputs the edge-union of the resulting graphs $H_{cd}$, for all $\set{c,d}\subset [N]$.
\end{proof}

\begin{theorem}\label{thm:conjectures-hold}
    Conjectures \ref{conj:lcw}, \ref{conj:cw}, and \ref{conj:seese} hold for every class 
  which is transduction equivalent with a class of bounded star chromatic number.
 In particular, they hold for all classes with structurally bounded expansion.
\end{theorem}

\begin{proof}
    The second part of the theorem follows from the first one by Theorem~\ref{thm:SBE} and Lemma~\ref{lem:star-coloring}.
    We prove the first part. The argument below is for the case Conjecture~\ref{conj:cw} (which implies Conjecture~\ref{conj:seese}), while the argument for Conjecture \ref{conj:lcw} is analogous.

    Suppose  $\CC$ is transduction equivalent to a class $\CC'$ which has bounded star chromatic number.  
    If $\CC$ has unbounded cliquewidth then also $\CC'$ has unbounded cliquewidth, and hence has unbounded treewidth. 
    By the grid minor theorem \cite{robertson86gridminor}, graphs from $\CC'$ contain arbitrarily large grids as minors,
    and therefore also subdivisions of all walls as subgraphs.
    Hence, the subgraph closure of $\CC'$ contains a subdivision of every wall.
    By Lemma~\ref{lem:transducing subgraphs}, the subgraph closure of $\CC'$ can be transduced from $\CC'$, and hence also from $\CC$. This proves that $\CC$ transduces a class which contains some subdivision of every wall.

In the case of Conjecture~\ref{conj:lcw}, the argument is the same, but instead of the grid minor theorem we use the fact that every  class with unbounded pathwidth contains every forest as a minor~\cite{ROBERTSON198339}, and hence some subdivision of every binary tree as a subgraph. 
\end{proof}

\begin{example}\label{ex:power graphs}
    We apply Theorem~\ref{thm:conjectures-hold} to a class of unbounded cliquewidth considered in~\cite{10.1007/978-3-662-53174-7_25} (see also \cite[Section 6]{dawar2020mso}) and defined as follows.
    A \emph{power graph} $G_n$ of order $n$ is the graph with vertices $1,\ldots,n$,
    where $i$ is adjacent with $j$ if and only if $|i-j|=1$ or 
    $i$ and $j$ are divisible exactly by the same powers of $2$.

    Let $\CC$ be the hereditary closure of the class of all power graphs.
Then $\CC$ has unbounded cliquewidth~\cite{10.1007/978-3-662-53174-7_25}
and transduces the class of all grids by an MSO transduction~\cite[Theorem 19]{dawar2020mso}.
We show that $\CC$ is (FO) transduction equivalent with a class $\CC'$ of bounded star chromatic number, and thus transduces a class containing some subdivision of every wall.

For each $n$ consider the graph $G_n'$ with vertices $1,\ldots,n,c_0,\ldots,c_{m}$
where $m=\lfloor\log_2n\rfloor$,
and where the vertices $1,\ldots,n$ form a path 
(in that order) and $c_i$ is adjacent with all vertices $k$ such that $1\le k\le n$
and $k=j\cdot 2^i$ for some odd integer $j$.
Thus, $c_i$ represents the equivalence class of those numbers $k\in \set{1,\ldots,n}$ such that $2^i$  is the highest power of $2$ by which $k$ is divisible.
Note that $G_n'$ has a star coloring with $4$ colors, in which the vertices $1\le i\le n$ are colored with color $(i \textrm{ mod } 3)$, while the vertices $c_1,\ldots,c_m$ are colored with color $3$. 

Let $\CC'$ be the hereditary closure of the class of all graphs of the form $G_n'$. It is easy to see that $\CC$ and $\CC'$ are transduction equivalent 
(the transduction from $\CC$ to $\CC'$ uses copying). Hence the class $\CC$ of power graphs is transduction equivalent with the class $\CC'$ which has star chromatic number bounded by $4$.
(In fact, $\CC'$ has bounded expansion and therefore $\CC$ has structurally bounded expansion).
As $\CC$ has unbounded cliquewidth, by Theorem~\ref{thm:conjectures-hold},
$\CC$ transduces a class which contains some subdivision of every wall.
\end{example}

In a similar fashion, it is easy to verify that each of the remaining classes studied by Dawar and Sankaran 
\cite{dawar2020mso} of unbounded clique-width -- bichain graphs, split permutation graphs, bipartite permutation graphs, and unit interval graphs --
transduces some class of unbounded cliquewidth that is transduction equivalent to a class of bounded star-chromatic number.
Thus, each of those classes transduces some class that contains some subdivision of every wall, by Theorem~\ref{thm:conjectures-hold}.
In particular, those classes MSO-transduce the class of all graphs, 
reproving \cite[Theorems 13 and 19]{dawar2020mso}.
Thus, the approach to 
Seese's conjecture proposed by Dawar and Sankanaran~\cite{dawar2020mso}  
--- to consider \emph{minimal} hereditary classes of unbounded cliquewidth, and antichains of unbounded cliquewidth with respect to the induced subgraph relation --- 
may be also applied to Conjecture~\ref{conj:cw}.
Reassuming, a possible route to proving Seese's conjecture is to prove (the much stronger) Conjecture~\ref{conj:cw}.

\medskip
We now show that the statements~\eqref{st:2}, \eqref{st:3}
imply Conjectures \ref{conj:lcw} and \ref{conj:cw}, respectively. We use the following result of 
Dvo\v r\'ak \cite[Theorem 3]{DBLP:journals/ejc/Dvorak18}.
We thank Patrice Ossona de Mendez for suggesting to us this result in connection with the considered implication.
\begin{theorem}[\cite{DBLP:journals/ejc/Dvorak18}]\label{thm:dvorak}
    Let $\CC$ be a weakly sparse class of graphs and suppose there is a graph $H$ such that no subdivision of $H$ is an induced subgraph of a graph in $\CC$. Then $\CC$ has bounded expansion.
\end{theorem}
We now show that the statement \eqref{st:3} implies Conjecture~\ref{conj:cw}.
For the statement \eqref{st:2} and Conjecture~\ref{conj:lcw}, the argument is analogous.

Let $\CC$ be a class with unbounded cliquewidth. According to \eqref{st:3}, there is a weakly sparse class $\DD$ with unbounded cliquewidth that can be transduced from $\CC$. If $\DD$ has bounded expansion then by Theorem~\ref{thm:conjectures-hold}, $\DD$ transduces a weakly sparse class $\DD'$ which contains some subdivision of every wall.
Otherwise, if $\DD$ has unbounded expansion, then by Theorem~\ref{thm:dvorak}, $\DD$ contains some subdivision of every graph $H$ as an induced subgraph. In particular, $\DD$ transduces a class 
$\DD'$ which contains some subdivision of every graph. 



\medskip

Finally, let us propose a solution 
to the crossword puzzle formed by Figure~\ref{fig:discussion},
by filling in the last column of the `bounded expansion' row
with the property `dense analogue of bounded expansion', 
which has now been formally defined. The third column is then the stable counterpart of that property, and we conjecture that this
coincides with structurally bounded expansion. This can be phrased as the following duality statement.

\begin{conjecture}\label{dense-sparse be}
    Exactly one of the following conditions holds for every stable class $\CC$ of graphs:
    \begin{enumerate}
        \item $\CC$ has structurally bounded expansion,
        \item $\CC$ transduces some weakly sparse class that does not have bounded expansion.
    \end{enumerate}
\end{conjecture}

\bibliographystyle{alpha}
\bibliography{ref}

\newcommand{\etalchar}[1]{$^{#1}$}
\begin{thebibliography}{BGK{\etalchar{+}}21b}

\bibitem[AA14]{adler2014interpreting}
Hans Adler and Isolde Adler.
\newblock Interpreting nowhere dense graph classes as a classical notion of
  model theory.
\newblock {\em European Journal of Combinatorics}, 36:322--330, 2014.

\bibitem[ALR12]{AtminasLR12}
Aistis Atminas, Vadim~V. Lozin, and Igor Razgon.
\newblock Linear time algorithm for computing a small biclique in graphs
  without long induced paths.
\newblock In {\em SWAT 2012}, volume 7357 of {\em Lecture Notes in Computer
  Science}, pages 142--152. Springer, 2012.

\bibitem[BC10]{lmcs:1208}
Achim Blumensath and Bruno Courcelle.
\newblock On the monadic second-order transduction hierarchy.
\newblock {\em Log. Methods Comput. Sci.}, 6(2), 2010.

\bibitem[BGK{\etalchar{+}}21a]{bonnet2021tww2}
{\'{E}}douard Bonnet, Colin Geniet, Eun~Jung Kim, St{\'{e}}phan Thomass{\'{e}},
  and R{\'{e}}mi Watrigant.
\newblock Twin-width {II:} small classes.
\newblock In {\em SODA 2021}, pages 1977--1996. {SIAM}, 2021.

\bibitem[BGK{\etalchar{+}}21b]{bonnet2020tww3}
{\'{E}}douard Bonnet, Colin Geniet, Eun~Jung Kim, St{\'{e}}phan Thomass{\'{e}},
  and R{\'{e}}mi Watrigant.
\newblock Twin-width {III:} {M}ax {I}ndependent {S}et, {M}in {D}ominating
  {S}et, and {C}oloring.
\newblock In {\em {ICALP} 2021}, volume 198 of {\em LIPIcs}, pages 35:1--35:20.
  Schloss Dagstuhl --- Leibniz-Zentrum f{\"{u}}r Informatik, 2021.

\bibitem[BGO{\etalchar{+}}21]{tww4a}
{\'{E}}douard Bonnet, Ugo Giocanti, Patrice {Ossona de Mendez}, Pierre Simon,
  St{\'{e}}phan Thomass{\'{e}}, and Szymon Toru\'nczyk.
\newblock Twin-width {IV:} low complexity matrices.
\newblock {\em CoRR}, abs/2102.03117, 2021.

\bibitem[BKTW20]{bonnet2020tww}
{\'{E}}douard Bonnet, Eun~Jung Kim, St{\'{e}}phan Thomass{\'{e}}, and
  R{\'{e}}mi Watrigant.
\newblock Twin-width {I:} tractable {FO} model checking.
\newblock In {\em FOCS 2020}, pages 601--612. {IEEE}, 2020.

\bibitem[BP20]{BonamyP20}
Marthe Bonamy and Micha\l{} Pilipczuk.
\newblock Graphs of bounded cliquewidth are polynomially $\chi$-bounded.
\newblock {\em Advances in Combinatorics}, (2020:8), 2020.

\bibitem[BS85]{baldwinshelah}
J.~T. Baldwin and S.~Shelah.
\newblock {Second-order quantifiers and the complexity of theories}.
\newblock {\em Notre Dame Journal of Formal Logic}, 26(3):229--303, 1985.

\bibitem[CMR00]{courcelle2000cw}
Bruno Courcelle, Johann~A. Makowsky, and Udi Rotics.
\newblock Linear time solvable optimization problems on graphs of bounded
  clique-width.
\newblock {\em Theory Comput. Syst.}, 33(2):125--150, 2000.

\bibitem[CO07]{COURCELLE200791}
Bruno Courcelle and Sang{-}il Oum.
\newblock Vertex-minors, monadic second-order logic, and a conjecture by seese.
\newblock {\em J. Comb. Theory, Ser. {B}}, 97(1):91--126, 2007.

\bibitem[Col07]{colcombet2007factorization}
Thomas Colcombet.
\newblock A combinatorial theorem for trees.
\newblock In {\em ICALP 2007}, volume 4596 of {\em Lecture Notes in Computer
  Science}, pages 901--912. Springer, 2007.

\bibitem[Cou90]{courcelle90tw}
Bruno Courcelle.
\newblock The monadic second-order logic of graphs. {I}. {R}ecognizable sets of
  finite graphs.
\newblock {\em Inf. Comput.}, 85(1):12--75, 1990.

\bibitem[DGJ{\etalchar{+}}22]{dreier2021twinwidth}
Jan Dreier, Jakub Gajarsk\'y, Yiting Jiang, Patrice {Ossona de Mendez}, and
  Jean-Florent Raymond.
\newblock Twin-width and generalized coloring numbers.
\newblock {\em Discrete Mathematics}, 345(3):112746, 2022.

\bibitem[DS20]{dawar2020mso}
Anuj Dawar and Abhisekh Sankaran.
\newblock {MSO} undecidability for some hereditary classes of unbounded
  clique-width.
\newblock {\em CoRR}, abs/2011.02894, 2020.

\bibitem[Dvo18a]{DBLP:journals/ejc/Dvorak18}
Zdenek Dvor{\'{a}}k.
\newblock Induced subdivisions and bounded expansion.
\newblock {\em Eur. J. Comb.}, 69:143--148, 2018.

\bibitem[Dvo18b]{Dvorak18}
Zden\v{e}k Dvo\v{r}{\'{a}}k.
\newblock Induced subdivisions and bounded expansion.
\newblock {\em European Journal of Combinatorics}, 69:143--148, 2018.

\bibitem[GHN{\etalchar{+}}19]{lmcs:5149}
Robert Ganian, Petr Hlin\v{e}n\'y, Jaroslav Ne\v{s}et\v{r}il, Jan
  Obdr\v{z}{\'a}lek, and Patrice {Ossona de Mendez}.
\newblock Shrub-depth: Capturing height of dense graphs.
\newblock {\em Log. Methods Comput. Sci.}, 15(1), 2019.

\bibitem[GHO{\etalchar{+}}20]{gajarsky2020bd_interp}
Jakub Gajarsk{\'{y}}, Petr Hlin\v{e}n{\'{y}}, Jan Obdr\v{z}{\'{a}}lek, Daniel
  Lokshtanov, and M.~S. Ramanujan.
\newblock A new perspective on {FO} model checking of dense graph classes.
\newblock {\em {ACM} Trans. Comput. Log.}, 21(4):28:1--28:23, 2020.

\bibitem[GKMW20]{geelen2020vertex_gridminor}
Jim Geelen, {O-joung} Kwon, Rose McCarty, and Paul Wollan.
\newblock The grid theorem for vertex-minors.
\newblock {\em J. Comb. Theory, Ser. {B}}, 2020.
\newblock In press.

\bibitem[GKN{\etalchar{+}}20]{gajarsky2020sbe}
Jakub Gajarsk{\'{y}}, Stephan Kreutzer, Jaroslav Ne\v{s}et\v{r}il,
  Patrice~Ossona de~Mendez, Micha\l{} Pilipczuk, Sebastian Siebertz, and Szymon
  Toru\'nczyk.
\newblock First-order interpretations of bounded expansion classes.
\newblock {\em {ACM} Trans. Comput. Log.}, 21(4):29:1--29:41, 2020.

\bibitem[GKS17]{grohe2017fo_nd}
Martin Grohe, Stephan Kreutzer, and Sebastian Siebertz.
\newblock Deciding first-order properties of nowhere dense graphs.
\newblock {\em J. {ACM}}, 64(3):17:1--17:32, 2017.

\bibitem[Gy{\'a}87]{gyarfas}
Andr{\'a}s Gy{\'a}rf{\'a}s.
\newblock Problems from the world surrounding perfect graphs.
\newblock {\em Applicationes Mathematicae}, 19:413--441, 1987.

\bibitem[LRZ16]{10.1007/978-3-662-53174-7_25}
Vadim~V. Lozin, Igor Razgon, and Viktor Zamaraev.
\newblock Well-quasi-ordering does not imply bounded clique-width.
\newblock In Ernst~W. Mayr, editor, {\em Graph-Theoretic Concepts in Computer
  Science}, pages 351--359, Berlin, Heidelberg, 2016. Springer Berlin
  Heidelberg.

\bibitem[NO12]{sparsity}
Jaroslav Ne\v{s}et\v{r}il and Patrice {Ossona de Mendez}.
\newblock {\em Sparsity --- {G}raphs, {S}tructures, and {A}lgorithms},
  volume~28 of {\em Algorithms and combinatorics}.
\newblock Springer, 2012.

\bibitem[NOP{\etalchar{+}}21]{nesetril2021rw_stable}
Jaroslav Ne\v{s}et\v{r}il, Patrice {Ossona de Mendez}, Micha\l{} Pilipczuk,
  Roman Rabinovich, and Sebastian Siebertz.
\newblock Rankwidth meets stability.
\newblock In {\em SODA 2021}, pages 2014--2033. {SIAM}, 2021.

\bibitem[NORS21]{nesetril2021linrw_stable}
Jaroslav Ne\v{s}et\v{r}il, Patrice {Ossona de Mendez}, Roman Rabinovich, and
  Sebastian Siebertz.
\newblock Classes of graphs with low complexity: The case of classes with
  bounded linear rankwidth.
\newblock {\em Eur. J. Comb.}, 91:103223, 2021.

\bibitem[NOS20]{nesetril2020arboretum}
Jaroslav Ne\v{s}et\v{r}il, Patrice {Ossona de Mendez}, and Sebastian Siebertz.
\newblock Structural properties of the first-order transduction quasiorder.
\newblock {\em CoRR}, abs/2010.02607, 2020.

\bibitem[{Oss}21]{Mendez21}
Patrice {Ossona de Mendez}.
\newblock First-order transductions of graphs (invited talk).
\newblock In {\em {STACS} 2021}, volume 187 of {\em LIPIcs}, pages 2:1--2:7.
  Schloss Dagstuhl --- Leibniz-Zentrum f{\"{u}}r Informatik, 2021.

\bibitem[PdMS22]{SD-paths}
Michał Pilipczuk, Patrice~Ossona de~Mendez, and Sebastian Siebertz.
\newblock Transducing paths in graph classes with unbounded shrubdepth, 2022.

\bibitem[RS83]{ROBERTSON198339}
Neil Robertson and Paul~D. Seymour.
\newblock Graph minors. i. excluding a forest.
\newblock {\em J. Comb. Theory, Ser. B}, 35(1):39--61, 1983.

\bibitem[RS86]{robertson86gridminor}
Neil Robertson and Paul~D. Seymour.
\newblock Graph minors. {V}. {E}xcluding a planar graph.
\newblock {\em J. Comb. Theory, Ser. {B}}, 41(1):92--114, 1986.

\bibitem[See91]{SEESE1991169}
D.~Seese.
\newblock The structure of the models of decidable monadic theories of graphs.
\newblock {\em Annals of Pure and Applied Logic}, 53(2):169--195, 1991.

\bibitem[Sim90]{simon90factorization}
Imre Simon.
\newblock Factorization forests of finite height.
\newblock {\em Theor. Comput. Sci.}, 72(1):65--94, 1990.

\bibitem[SS20]{ScottS20}
Alex Scott and Paul~D. Seymour.
\newblock A survey of $\chi$-boundedness.
\newblock {\em J. Graph Theory}, 95(3):473--504, 2020.

\bibitem[Zhu09]{zhu2009colouring}
Xuding Zhu.
\newblock Colouring graphs with bounded generalized colouring number.
\newblock {\em Discrete Mathematics}, 309(18):5562--5568, 2009.

\end{thebibliography}

\end{document}